%% file: paper.tex
\newif\ifcomments     %% include author discussion
\newif\ifanonymous    %% include author identities
\newif\ifextended     %% include appendix
\newif\ifsubmission   %% prepare the submitted version
\newif\ifpublic       %% version available for posting / final version

\commentsfalse       %% toggle comments here
\anonymousfalse
\extendedtrue
\submissionfalse      %% if you want to see comments, these should be off
\publicfalse

\ifpublic
\submissionfalse
\commentsfalse
\anonymousfalse
\fi

\ifanonymous
\documentclass[acmsmall,screen,anonymous,review]{acmart}
\else
\ifextended
\documentclass[acmsmall,screen]{acmart}
\acmDOI{}
\makeatletter
\renewcommand\@formatdoi[1]{\ignorespaces}
\makeatother
\renewcommand\footnotetextcopyrightpermission[1]{}
\else
\documentclass[acmsmall,screen]{acmart}
\fi

\usepackage{amsmath}
\usepackage{schemata}
\usepackage{latexsym}
\usepackage{bbm}
\usepackage{supertabular}
\usepackage{stmaryrd}
\usepackage{color}
\usepackage{multirow}
\usepackage{calc}
\usepackage{graphicx}
\usepackage{xspace}
\usepackage{dashrule}  % for \hdashrule
\usepackage{stackrel}
\usepackage{enumerate}
\usepackage{framed}
\usepackage{bussproofs}
\usepackage{mdframed}  % highlight
\usepackage{hyperref}
\usepackage{comment}
\usepackage{ifthen}
\usepackage{pdftexcmds}
\usepackage{mathpartir}
\usepackage{ottalt}
\usepackage{fancyvrb}
\usepackage{listings}
\usepackage{url}
\usepackage{prettyref}

\usepackage{tikz}
\usepackage{microtype}

\newcommand{\pref}[1]{\prettyref{#1}}

\ifcomments
\newcommand{\scw}[1]{\textcolor{blue}{{SCW: #1}}}
\newcommand{\rae}[1]{\textcolor{magenta}{{RAE: #1}}}
\newcommand{\hde}[1]{\textcolor{gray}{{HDE: #1}}}
\newcommand{\pc}[1]{\textcolor[rgb]{0.0, 0.5, 0.0}{{PC: #1}}}
\else
\newcommand{\scw}[1]{}
\newcommand{\rae}[1]{}
\newcommand{\hde}[1]{}
\newcommand{\pc}[1]{}
\fi

\ifextended
\newcommand{\auxiliarymaterial}{the appendix}
\newcommand{\auxref}[1]{Appendix~\ref{#1}}
\else
\ifsubmission
\newcommand{\auxiliarymaterial}{the anonymized supplementary material}
\else
\newcommand{\auxiliarymaterial}{the extended version of this paper}
\fi
\newcommand{\auxref}[1]{\auxiliarymaterial}
\fi

\newtheorem{theorem}{Theorem}[section]

\newtheorem{lemma}[theorem]{Lemma}

\newtheorem{notation}[theorem]{Notation}

\theoremstyle{remark}
\newtheorem{example}[theorem]{Example}

\newcommand{\alt}{\ensuremath{\ |\ }}

\newcommand\ncoverline[1]{\mkern1mu\overline{\mkern-1mu#1\mkern-1mu}\mkern1mu}



\setcopyright{rightsretained}
\acmPrice{}
\acmDOI{10.1145/3434331}
\acmYear{2021}
\copyrightyear{2021}
\acmSubmissionID{popl21main-p408-p}
\acmJournal{PACMPL}
\acmVolume{5}
\acmNumber{POPL}
\acmArticle{1}
\acmMonth{1}

\begin{document}

%% Bibliography style
\bibliographystyle{ACM-Reference-Format}

%% Citation style
%% Note: author/year citations are required for papers published as an
%% issue of PACMPL.
\citestyle{acmauthoryear}   %% For author/year citations

\ifextended
\title{A Graded Dependent Type System with a Usage-Aware Semantics (extended version)}
\else
\title{A Graded Dependent Type System with a Usage-Aware Semantics}
\fi

\author{Pritam Choudhury}
 \affiliation{
   \position{}
   \department{Computer and Information Science}              %% \department is recommended
   \institution{University of Pennsylvania}                   %% \institution is required
%   \streetaddress{}
%   \city{}
%   \state{}
%   \postcode{}
   \country{USA}
 }
 \email{pritam@seas.upenn.edu}

\author{Harley Eades III}
%\orcid{0000-0001-8474-5971}
\affiliation{
%  \position{Professor}
  \department{School of Computer and Cyber Sciences}              %% \department is recommended
  \institution{Augusta University}                   %% \institution is required
  \streetaddress{2500 Walton Way}
  \city{Augusta}
  \state{GA}
  \postcode{30904}
  \country{USA}
}
\email{harley.eades@gmail.com}

\author{Richard A.~Eisenberg}
\affiliation{
  \position{Principal Researcher}
  \institution{Tweag I/O}
  \city{Paris}
  \country{France}
}
 \affiliation{
   \position{Assistant Professor}
   \department{Computer Science}              %% \department is recommended
   \institution{Bryn Mawr College}            %% \institution is required
   \streetaddress{101 N.~Merion Ave}
   \city{Bryn Mawr}
   \state{PA}
   \postcode{19010}
   \country{USA}
 }
 \email{rae@richarde.dev}

\author{Stephanie Weirich}
%\orcid{0000-0002-6756-9168}
\affiliation{
%  \position{Professor}
  \department{Computer and Information Science}              %% \department is recommended
  \institution{University of Pennsylvania}                   %% \institution is required
  \streetaddress{3330 Walnut St}
  \city{Philadelphia}
  \state{PA}
  \postcode{19104}
  \country{USA}
}
\email{sweirich@cis.upenn.edu}

\begin{abstract}
\input{abstract}
\end{abstract}

\maketitle
\ifextended
\thispagestyle{plain}
\pagestyle{plain}
\fi

%%%%%%%%%%%%%%%%%%%%%%%%%%%%%%%%%%%%%%%%%%%%%%%%%%%%%%%%%%%%%%%%%%%%%%%%%%%%%%%%%
%%%%%%%%%%%%%%%%%%%%%%%%%%%%%%%%%%%%%%%%%%%%%%%%%%%%%%%%%%%%%%%%%%%%%%%%%%%%%%%%%

\newcommand{\Langname}{\textsc{GraD}\xspace}
\newcommand{\extendedurl}{\url{http://arxiv.org/abs/2011.04070}}

  \input{qtt-rules}

  \renewottcommands[ott]

\scw{The page limit for the final version is 26 pages, excluding
   references. You may buy up to 4 additional pages for \$100 per page.}

\section{Introduction}
Consider this typing judgement.
\[    \ottmv{x} \! :^{  { \color{black}{1} }  }\! \ottkw{Bool}  ,   \ottmv{y} \! :^{  { \color{black}{1} }  }\! \ottkw{Int}   ,   \ottmv{z} \! :^{  { \color{black}{0} }  }\! \ottkw{Bool}   \vdash \ottkw{if}\, x \, \ottkw{then}\, y + 1\, \ottkw{else}\, y - 1\, :\, \ottkw{Int} \]
Here, the numbers in the context indicate that the variable $x$ is used once
in the expression, the variable $y$ is also used only once (although it
appears twice), and the variable $z$ is never used at all.

This sentence is a judgement of a \emph{graded} type system which ensures
that the \emph{grades} or \emph{quantities} annotating each in-scope variable
reflects how it is used at run time. Graded type systems have been explored in much detail in the literature ~\cite{Ghica:2014,Brunel:2014,Gaboardi:2016,McBride:2016,atkey,10.1145/2628136.2628160,orchard:2019}\scw{add more citations here}. The process of tracking usage through grades is straightforward, but this is a powerful method of instrumenting type systems with analyses of irrelevance and linearity that have practical benefits like erasure of irrelevant terms (resulting in speed-up) and compiler optimizations (such as in-place update of linear resources). This approach is also versatile. 
By abstracting over a domain of resources, the same form of type system can be used to guarantee safe memory usage, or prevent insecure information flow, or quantify information leakage, or
identify irrelevant computations, or combine various modal logics. Several research
languages, such as Idris 2~\cite{brady:2020} and Agda~\cite{agda}, are starting to adopt
ideas from this domain, and new systems like Granule~\cite{orchard:2019} are being developed to explore its possibilities.

Our concrete motivation for studying graded type systems is a desire to merge Haskell's
current form of a linear type system~\cite{linear-haskell} with dependent
types~\cite{weirich:icfp17} in a clean manner. Crucially, the combined system
must support \emph{type erasure}: the compiler must be able to eliminate type arguments
to polymorphic functions. Type erasure is key both to support
parametric polymorphism and to efficiently execute Haskell programs. We discuss this in more detail in Section~\ref{sec:goal}.

% We find that the general structure of quantitative type
% theory, including its parameterization over an arbitrary semiring
% of usages (Section~\ref{sec:semiring}) fits nicely, and so we develop and
% present this more general system, though our needs are, in the end, more modest.
% Section~\ref{sec:haskell} describes how the ideas in this paper will be
% applied toward Haskell.

Although Haskell is our eventual goal, our work remains general. Our designs
are compatible with the current approaches in GHC, but are not specialized to
Haskell.

% In particular, unlike recent related work, we
% favor syntactic techniques for our analysis of the properties of the
% system. As a result, our proofs say are not as deep but, in exchange, are less
% limited in scope. In particular, our design is appropriate for languages that
% include nontermination, such as Haskell.

We make the following contributions in this paper:
\begin{itemize}
\item Our system flexibly abstracts over an
  algebraic structure used to count resources.
  Section~\ref{sec:semiring} describes this structure---a partially-ordered
  semiring---and its properties. This use of a resource algebra is standard,
  although we identify subtle differences in its specification.
\item Section~\ref{sec:simple} presents a simple graded type system, with standard
  algebraic types and a graded modal type. This system is not novel;
  instead, it establishes a foundation for the dependent system. However, even at this
  stage, we identify subtleties in the design space.
\item Because the standard operational semantics does not track resources,
  type safety does not imply that usage tracking is correct.  Section~\ref{sec:heap-semantics} describes
  a heap-based operational semantics, inspired by \citet{turner}. 
  Every variable in the heap has an associated resource tag from our abstract
  structure, modelling how resources are used during computation. We prove
  that our type system is sound with respect to this instrumented semantics.
  This theorem tells us that well-typed terms will not get stuck by running out
  of resources.  In the process of showing that this result holds, we identify
  a key restriction on case analysis that was not forced by the non-resourced
  version of type safety.
\item Using soundness, we show (a generalization of) the \emph{single pointer property} for
  linear resources in Section~\ref{sec:applications}. The single pointer property
  says that a linear resource is referenced by precisely one pointer at runtime.
  Such a property would enable in-place updates of linear resources.
\item Our key contribution is the design of the language, \Langname{},
  extending our ideas to dependent types.  In contrast to other
  approaches~\cite{atkey}, we use the same rules to check
  relevant and irrelevant phrases (that is, terms and types). When computing
  the resources used by the entire term, we discard irrelevant usages. Types
  are irrelevant to computation, so our system ignores these usages. We
  describe the design of the type system in Section~\ref{sec:dependent} and
  extend the soundness proof with respect to a heap semantics in
  Section~\ref{sec:heap-dependent}.

  Our system is thus both simpler and more uniform than prior work that combines
  usage tracking with dependent types. In particular, Quantitative Type Theory
  (QTT)~\cite{McBride:2016,atkey} disables resource checking in types, leading to
  limitations on the sorts of reasoning that can be done in the type
  system. On the other hand, Resourceful Dependent Types~\cite{abel:2018} and
  GrTT~\cite{Moon:2020} maintain separate counts of usages in types and terms,
  incurring additional bookkeeping for less benefit.   Section~\ref{sec:qtt-comparison} provides a 
  detailed comparison of our work with QTT.
\item We have mechanized, in Coq, some intricate syntactic properties of our
  development (substitution, weakening, preservation, progress). These proof
  scripts are available online at
  \url{http://www.github.com/sweirich/graded-haskell}.
\end{itemize} 
\ifextended
This paper is an extended version of ``A graded dependent type system with a usage-aware semantics''.
\else
An extended version of this paper is available at \extendedurl. 
% This longer version includes a reference appendix with the complete specification of the systems described here as well as a more detailed explanation of the .
\fi

\section{Our Goal}
\label{sec:goal}

While the exploration of graded type systems in this paper is applicable to
a wide array of examples (see Section~\ref{sec:semiring-examples}), we were
originally motivated to study such systems in the context of GHC/Haskell, where
we wish to combine its existing support for linearity~\cite{linear-haskell}
(available as of GHC~9.0) with support for dependent types~\cite{gundry-thesis,eisenberg-thesis,weirich:icfp17,role-dependent-haskell}.

A key advantage of a successful combination of graded and dependent types for Haskell is that it allows us to use the $0$ quantity to mean \emph{irrelevant} usage, where an irrelevant sub-term is one that is 
not needed while computing the reduct of the concerned term. Irrelevant sub-terms are quite common in terms derived in dependent type systems. They are essential for type-checking the terms but if left as such, they can make programs run much slower. So we need to track irrelevant sub-terms and erase them before running a program. \citet{weirich:icfp17} use a relevance tag $+/-$ on the $\Pi$ for this purpose. On the other hand, \citet{linear-haskell} use a linearity tag $1/\omega$ on the function domain type to track linear usage. We can combine these two together using $0$, $1$ and $\omega$ to track irrelevant, linear and unrestricted usages respectively. This has an added advantage. It will allow us to provide Haskell programmers the option of annotating arguments with a usage tag that subsumes relevance and linearity tags \cite{proposal-102}. So the use of $0$ to mark irrelevance fits in swimmingly with Haskell's current story around linear types.

Furthermore, given that we plan to implement these ideas concretely inside GHC,
it is essential that the system be as simple as possible. As discussed in
more detail in Section~\ref{sec:qtt-comparison}, our system eliminates features that are not necessary in our case. Doing so will aid in integration with the rest of the GHC implementation. 

Our intentions laid out, we start our exploration by reviewing semirings, the
key algebraic structure used to abstractly represent grades.

\section{The Algebra of Quantities}
\label{sec:semiring}

%In this section we introduce the algebra of quantities. This provides the necessary mathematical %framework to present our system. 

The goal of a graded type theory is to track the demands that computations
make on variables that appear in the context. In other words, the type system
enables a static accounting of runtime resources ``used'' in the evaluation of
terms.  This form of type system generalizes linear types (where linear
resources must be used exactly once)~\cite{Wadler:1990} and bounded linear
types (where bounded resources must be used a finite number of
times)~\cite{Girard:1992}, as well as many, many other type
systems~\cite{reed:2010,Miquel:ICC,pfenning:2001,abadi:dcc,volpano:noninterference}.

This generality derives from the fact that the type system is parametrized
over an abstract algebraic structure of \emph{grades} to model
resources.\footnote{Grades are also called quantities, modalities, resources, coeffects or usages.}  The abstract algebraic structure enables addition and multiplication of resources and these operations conform to our general understanding of resource arithmetic. A  partially-ordered semiring is one such algebraic structure that captures this idea of resource modelling nicely. 

% In this system, variables
% in the context are marked with a \emph{quantity} $q$, drawn from an arbitrary
% partially-ordered semiring that we use for resource accounting.
% \[  \Gamma \mathsf{::=}  \varnothing  \alt  \Gamma ,   \ottmv{x} \!:^{ \ottnt{q} }\! \ottnt{A}   \]

\subsection{Partially-Ordered Semirings}
\label{posemiring}

A \emph{semiring} is a set $Q$ with two binary operations,
$\_{+}\_ : Q \times Q \to Q$ (addition) and $\_{\cdot}\_ : Q \times Q \to Q$
(multiplication), and two distinguished elements, $0$ and $1$, such
that $(Q,+,0)$ is a commutative monoid and $(Q,\cdot,1)$ is a monoid;
furthermore, multiplication is both left and right distributive over addition
and $0$ is an annihilator for multiplication. Note that a semiring is not a
full ring because addition does not have an inverse---we cannot subtract.
% Furthermore, our proofs also do not require the multiplication operation to be
% commutative.

We mark the variables in our contexts with quantities drawn from a semiring to represent demand of resources. In other words, if we have a typing derivation for a term $\ottnt{a}$ with
free variable $\ottmv{x}$ marked with ${ \color{black}{q} }$, we know that $\ottnt{a}$ demands ${ \color{black}{q} }$ uses of $\ottmv{x}$. 

We can weaken the precision of our type system (but increase its flexibility)
by allowing the judgement to express higher demand than is actually
necessary. For example, we may need to use some variable only once but it may
be convenient to declare that the usage of that variable
is unrestricted. To model this notion of \emph{sub-usage}, we need an ordering
on the elements of the abstract semiring, reflecting our notion of
leniency. A partial order captures the idea nicely. Since we work with a semiring, such an order should
also respect the binary operations of the semiring. Concretely, for a partial
order $ \leq $ on $Q$, if ${ \color{black}{q} }_{{\mathrm{1}}}  \leq  { \color{black}{q} }_{{\mathrm{2}}}$, then for any $q \in Q$, we should have
${ \color{black}{q} }  \ottsym{+}  { \color{black}{q} }_{{\mathrm{1}}}  \leq  { \color{black}{q} }  \ottsym{+}  { \color{black}{q} }_{{\mathrm{2}}}$, ${ \color{black}{q} }  \cdot  { \color{black}{q} }_{{\mathrm{1}}}  \leq  { \color{black}{q} }  \cdot  { \color{black}{q} }_{{\mathrm{2}}}$, and ${ \color{black}{q} }_{{\mathrm{1}}}  \cdot  { \color{black}{q} }  \leq  { \color{black}{q} }_{{\mathrm{2}}}  \cdot  { \color{black}{q} }$.
 A semiring with a partial order satisfying this condition is called a
\emph{partially-ordered} semiring.

This abstract structure captures the operations and properties that
the type system needs for resource accounting. Because we are working
abstractly, we are limited to exactly these assumptions. In practice, it means
our design is applicable to settings beyond the simple use of natural numbers to count resources.

\subsection{Examples of Partially-Ordered Semirings}
\label{sec:semiring-examples}

Looking ahead, there are a few semirings that we are interested in. The
\emph{trivial semiring} has a single element, and all operations just return
that element. Our type system, when specialized to this semiring,
degenerates to the usual form of types as the quantities are uninformative.

The \emph{boolean semiring} has two elements, 0 and 1, with the property that
$1 + 1 = 1$. A type system drawing quantities from this semiring distinguishes between variables that are
used (marked with one) and ones that are unused (marked with zero). In such a 
system, the quantity $1$ does \emph{not} correspond to a linear usage: this system
does not count usage, but instead checks \emph{whether} a variable is used or not.

There are two different partial orders that make sense for the boolean
semiring.  If we use the reflexive relation, then this type system
tracks relevance precisely. If a variable is marked with $0$ in the context, then we
know that the variable \emph{must not} be used at runtime, and if it is marked
with $1$, then we know that it \emph{must} be used. On the other hand, if the
partial ordering declares that $ { \color{black}{0} }   \leq   { \color{black}{1} } $, then we still can determine that
$0$-marked variables are unused, but we do not know anything about the usage
of $1$-marked variables.

The \emph{linearity semiring} has three elements, 0, 1 and $\omega$, where addition and multiplication
are defined in the usual way after interpreting $\omega$ as ``greater than $1$''. So, we have $1 + 1 = \omega$, $\omega + 1 = \omega$, and $\omega \cdot \omega = \omega$. A system using the linearity semiring
tracks linearity by marking linear variables with 1 and unrestricted variables with $\omega$. A suitable ordering in this semiring is the reflexive closure of $\{ (0,\omega), (1, \omega) \}$. We do \emph{not} want $0 \leq 1$, since then we would not be able to guarantee that linear variables in the context are used exactly once. This semiring is the one that makes the most sense for Haskell as it integrates linearity (1) with irrelevance (0) and unrestricted usage ($\omega$).

The \emph{five-point linearity semiring} has five elements, 0, 1, Aff, Rel and $\omega$, where addition and multiplication are defined in the usual way after interpreting Aff as ``1 or less", Rel as ``1 or more", and $\omega$ as unrestricted. An ordering reflecting this interpretation is the reflexive transitive closure of $\{ (0,\text{Aff})), (1,\text{Aff}), (1,\text{Rel}), (\text{Aff}, \omega), (\text{Rel},\omega)\}$. This semiring can be used to track irrelevant, linear, affine, relevant, and unrestricted usage. \rae{Can someone double-check that ordering? It doesn't match my expectation.} \hde{Looks right to me.  Think of it as, as you go up the chain you increase in the number of structural rules you can use.  So you can think of it as specifying the subsets of structural rules Aff, Rel, and $\omega$ have starting with weakening (0) and linearity (1).}\scw{also looks correct to me.}

A \emph{security semiring} is based on a lattice of
security levels, with increasing order representing decreasing security. The $+$ and $ \cdot $ correspond to the join and meet operations of the lattice respectively. The partial order corresponds to the lattice order and $0$ and $1$ are the $ {\color{black}{\mathsf{Private} } } $ and $ {\color{black}{\mathsf{Public} } } $ security levels respectively. $ {\color{black}{\mathsf{Public} } } $ can never be as or more secure than $ {\color{black}{\mathsf{Private} } } $, i.e. $ {\color{black}{\mathsf{Public} } }  \nleq  {\color{black}{\mathsf{Private} } } $. This lattice may include additional elements besides $ {\color{black}{\mathsf{Private} } } $ and $ {\color{black}{\mathsf{Public} } } $, corresponding to multiple levels of secrecy. As \citet{abel:icfp2020} describe, security type systems defined in this way differ from the usual convention (such as \citet{abadi:dcc}) in that security levels are relative to 1, the level of the program under execution.

Many other examples of semirings are possible. \citet{orchard:2019} and
\citet{abel:icfp2020} include comprehensive lists of several other
applications including a type system for differential
privacy~\cite{reed:2010} and a type system that tracks
covariant/contravariant use of assumptions.

Partially-ordered semirings have been used to track resource usage in many type systems \\ \cite{Ghica:2014,Brunel:2014,Gaboardi:2016,McBride:2016,atkey,10.1145/2628136.2628160,orchard:2019,abel:icfp2020} but there are some variations with respect to the formal requirements. For example: \citet{Brunel:2014} require the underlying set (of the semiring) along with the order to form a bounded sup-semilattice while \citet{abel:icfp2020} define the order using an additional meet operation on the underlying set; \citet{McBride:2016} uses a hemiring (a semiring without $1$) while \citet{atkey} uses a semiring where zero-usage satisfies a certain condition. Our theory is parametrized by a partially ordered semiring as defined in Section \ref{posemiring}. We add additional constraints as required only while deriving specific properties in Section \ref{sec:applications}.   

\section{A simple graded type system}
\label{sec:simple}

\begin{figure}
\centering
\begin{flushright}
\textit{(Grammar)}
\end{flushright}
\[
\begin{array}{llcl}
\textit{types} & \ottnt{A},\ottnt{B}   & ::=& \ottkw{Unit} \alt   {}^{ { \color{black}{q} } } \ottnt{A} \rightarrow  \ottnt{B}   \alt   \Box^{ { \color{black}{q} } }  \ottnt{A}  \alt  \ottnt{A}  \otimes  \ottnt{B}  \alt  \ottnt{A}  \oplus  \ottnt{B}   \\
\textit{terms} & \ottnt{a}, \ottnt{b}   & ::=& \ottmv{x} % \alt true \alt false \alt \ottkw{if} \, \ottnt{a_{{\mathrm{1}}}} \, \ottkw{then} \, \ottnt{a_{{\mathrm{2}}}} \, \ottkw{else} \, \ottnt{a_{{\mathrm{3}}}} 
                                       \alt  \lambda \ottmv{x} \!:^ { \color{black}{q} } \! \ottnt{A} . \ottnt{a}  \alt \ottnt{a} \, \ottnt{b} \\
                                && \alt & \ottkw{unit}  \alt  \ottkw{let}\,  \ottkw{unit} \,=\, \ottnt{a} \ \ottkw{in}\  \ottnt{b}  \alt  \ottkw{box} _ { \color{black}{q} } \, \ottnt{a}  \alt  \ottkw{let}\, \ottkw{box} \, \ottmv{x} \,=\, \ottnt{a} \ \ottkw{in}\  \ottnt{b}  \\
                                && \alt & \ottsym{(}  \ottnt{a}  \ottsym{,}  \ottnt{b}  \ottsym{)} \alt \ottkw{let} \, \ottsym{(}  \ottmv{x}  \ottsym{,}  \ottmv{y}  \ottsym{)}  \ottsym{=}  \ottnt{a} \, \mathsf{in} \, \ottnt{b} \\
                                && \alt &  \ottkw{inj}_1\,  \ottnt{a}  \alt  \ottkw{inj}_2\,  \ottnt{a}  \alt  \ottkw{case}_ { \color{black}{q} } \,  \ottnt{a} \, \ottkw{of}\,  \ottnt{b_{{\mathrm{1}}}}  ;  \ottnt{b_{{\mathrm{2}}}}  \\
\\
\textit{usage contexts} & \Gamma & ::= &  \varnothing  \alt  \Gamma ,   \ottmv{x} \! :^{ { \color{black}{q} } }\! \ottnt{A}   \\
\textit{contexts} & \Delta & ::= &  \varnothing  \alt  \Delta ,   \ottmv{x} \!\!:\!\! \ottnt{A}  \\
\\
\textit{typing judgement} & & & \boxed{ \Delta \ ;\  \Gamma  \vdash \ottnt{a} : \ottnt{A} }
\end{array}
\]

\caption{The simply typed graded $\lambda$-calculus}
\label{fig:min}
\end{figure}

% A partially-ordered semiring gives us an abstract notion of usage. We want
% to lift this notion to a type system to model the flow of variable usage. A typical
% hypothetical typing judgement gives us a way to construct a term of a certain
% type from assumptions of variables belonging to some types. In our case, a
% hypothetical typing judgement would give us a way to construct a usage of a
% certain type using usages belonging to other types. 
%  where if we ignore the usage
% context $\Gamma$, then $\Delta  \vdash  \ottnt{a} : \ottnt{A}$ says that $\ottnt{a}$ is
% typable in context $\Delta$ without quantities, and then quantities are add
% through $\Gamma$ which is exactly $\Delta$ with quantities.  So the judgment
% $ \Delta \ ;\  \Gamma  \vdash \ottnt{a} : \ottnt{A} $ can be read as ``term $\ottnt{a}$ has type $\ottnt{A}$ in the
% context $\Delta$ with the quantities in $\Gamma$.''

Our goal is to design a \emph{dependent} usage-aware type system. But, for
simplicity, we start with a simply-typed usage-aware system similar to the system of \citet{10.1145/2628136.2628160}. The grammar for this system appear in \pref{fig:min}. It is parametrized
over an arbitrary partially-ordered semiring $(Q, 1,\cdot,0,+,\leq)$ with
grades ${ \color{black}{q} } \in Q$.

The typing judgement for this system has the form $ \Delta \ ;\  \Gamma  \vdash \ottnt{a} : \ottnt{A} $; the rules appear inline below.
This judgement includes both a standard typing context $\Delta$ and a usage context $\Gamma$, a
copy of the typing context annotated with grades. For brevity in examples, we often elide the 
standard typing context as the information is subsumed by the usage context.
Indeed, in any derivation, the typing context and the usage context
correspond:

\begin{notation} ~
\begin{itemize}
\item
  The notation $ \lfloor \Gamma \rfloor $ denotes a typing context $\Delta$ same as $\Gamma$,
  but with no grades. 
\item The notation
  $\ncoverline{\Gamma}$ denotes the vector of grades in $\Gamma$.
\item The notation $ \Delta \vdash \Gamma $ denotes that $\Delta  \ottsym{=}   \lfloor \Gamma \rfloor $.
\end{itemize}
\end{notation}

\begin{lemma}[Typing context correspondence]
  If $ \Delta \ ;\  \Gamma  \vdash \ottnt{a} : \ottnt{A} $, then $ \Delta \vdash \Gamma $.
\end{lemma}

This style of including both a plain typing  context $\Delta$ and its usage
counterpart $\Gamma$ in the judgement is merely for convenience; it allows us
to easily tell when two usage contexts differ only in their quantities. There are many alternative ways to express the same information
in the type system: we could have only one usage context $\Gamma$ and add
constraints, or have a typing context $\Delta$ and a separate vector
of quantities. % Any of these approaches would work.

\subsection{Type System Basics}

We are now ready to start our tour of the typing rules of this system.

\paragraph{Variables}\

\centerline{\drule{ST-Var}\drule{ST-Weak}}
\vspace{1ex}

\scw{I find this explanation about variables awkward, but I'm not sure how to fix it.}
We see here that a variable $\ottmv{x}$ has type
$\ottnt{A}$ if it has type $\ottnt{A}$ in the context---that part is
unsurprising. However, as is typical in this style of systems, the context
is extended to include $  { \color{black}{0} }    \cdot   \Gamma $:
this notation means that all variables in $\Gamma$ must
have a quantity of 0. % This allows weakening to happen at the variable axiom.

\begin{notation}[Context scaling]
  The notation $ { \color{black}{q} }   \cdot   \Gamma $ denotes a context $\Gamma'$ such that, for each
  $ \ottmv{x} \! :^{ { \color{black}{r} } }\! \ottnt{A}  \in \Gamma$, we have $ \ottmv{x} \! :^{ { \color{black}{q} }  \cdot  { \color{black}{r} } }\! \ottnt{A}  \in \Gamma'$.
\end{notation}

The \rref{ST-Var} states that all variables other than $\ottmv{x}$ are
not used in the expression $\ottmv{x}$, that is why their quantity is zero.
Note also that $ \ottmv{x} \! :^{  { \color{black}{1} }  }\! \ottnt{A} $ occurs last in the context. If we wish to use a
variable that is not the last item in the context, the \rref{ST-Weak} allows us
to remove (reading from bottom to top) zero-usage variables at the end of a
context.

\paragraph{Sub-usage}\

\centerline{\drule{ST-Sub}}

We may allow our contexts to provide more resources than is necessary. Sub-usaging, as it is commonly
referred to, allows us to assume more resources in our context than are necessary.

\begin{notation}[Context sub-usage]
The notation $\Gamma_{{\mathrm{1}}}  \leq  \Gamma_{{\mathrm{2}}}$ means $ \lfloor \Gamma_{{\mathrm{1}}} \rfloor  =  \lfloor \Gamma_{{\mathrm{2}}} \rfloor $ where, for every corresponding pair of assumptions $ \ottmv{x} \! :^{ { \color{black}{q} }_{{\mathrm{1}}} }\! \ottnt{A}  \in \Gamma_{{\mathrm{1}}}$ and $ \ottmv{x} \! :^{ { \color{black}{q} }_{{\mathrm{2}}} }\! \ottnt{A}  \in \Gamma_{{\mathrm{2}}}$, the condition ${ \color{black}{q} }_{{\mathrm{1}}}  \leq  { \color{black}{q} }_{{\mathrm{2}}}$ holds.
\end{notation}

\paragraph{Functions}\

\centerline{\drule{ST-Lam}\drule{ST-App}}

Any quantitative type system must be careful around expressions that contain
multiple sub-expressions. Function application is a prime example, so we
examine \rref{ST-App} next. In this rule, we see that the function $\ottnt{a}$ has
type $ {}^{ { \color{black}{q} } } \ottnt{A} \rightarrow  \ottnt{B} $, meaning that it uses its argument, of type $\ottnt{A}$,
${ \color{black}{q} }$ times to produce a result of type $\ottnt{B}$.  Accordingly, we must make
sure that the argument expression $\ottnt{b}$ can be used ${ \color{black}{q} }$
times. Put another way, we must \emph{multiply} the usage required for
$\ottnt{b}$, as recorded in the typing context $\Gamma_{{\mathrm{2}}}$, by ${ \color{black}{q} }$. We see this in the context used in the rule's conclusion:
$\Gamma_{{\mathrm{1}}}  \ottsym{+}   { \color{black}{q} }   \cdot   \Gamma_{{\mathrm{2}}} $. 

This introduces another piece of important notation:
\begin{notation}[Context addition]
Adding contexts $\Gamma_{{\mathrm{1}}}  \ottsym{+}  \Gamma_{{\mathrm{2}}}$ is defined only when $ \lfloor \Gamma_{{\mathrm{1}}} \rfloor   \ottsym{=}   \lfloor \Gamma_{{\mathrm{2}}} \rfloor $.
The result context $\Gamma_{{\mathrm{3}}}$ is obtained by point-wise addition of quantities; i.e. for every $ \ottmv{x} \! :^{ { \color{black}{q} }_{{\mathrm{1}}} }\! \ottnt{A}  \in \Gamma_{{\mathrm{1}}}$
and $ \ottmv{x} \! :^{ { \color{black}{q} }_{{\mathrm{2}}} }\! \ottnt{A}  \in \Gamma_{{\mathrm{2}}}$, we have $ \ottmv{x} \! :^{ { \color{black}{q} }_{{\mathrm{1}}}  \ottsym{+}  { \color{black}{q} }_{{\mathrm{2}}} }\! \ottnt{A}  \in \Gamma_{{\mathrm{3}}}$.
\end{notation}
Our approach using two contexts $\Delta$ and $\Gamma$ works nicely here. Because
both premises to \rref{ST-App} use the same $\Delta$, we know that the required
precondition of context addition is satisfied.
\hde{We have brought the above statement up a bunch, but this is not new, and we can do this for the case when $\Gamma$ is just a vector.}
The high-level idea here is
common in sub-structural type systems: whenever we use multiple
sub-expressions within one expression, we must \emph{split} the context. One
part of the context checks one sub-expression, and the remainder checks other
sub-expression(s).

\begin{example}[Irrelevant application]
\label{example:irrelevant}
Before considering the rest of the system, it is instructive to step through
an example involving a function that does not use its argument in the context of the linearity semiring. We say that
such arguments are \emph{irrelevant}. 
Suppose that we have a function $\ottmv{f}$,
of type $ {}^{  { \color{black}{0} }  } \ottnt{B} \rightarrow    {}^{  { \color{black}{1} }  } \ottnt{A} \rightarrow  \ottnt{A}   $.  (Just from this type, we can see that $\ottmv{f}$
must be constant in $B$.) Suppose also that we want to apply this function
to some variable $\ottmv{x}$. In this case, define the usage contexts
\begin{gather*}
   \Gamma_{{\mathrm{0}}} =   \ottmv{f} \! :^{  { \color{black}{1} }  }\! \ottsym{(}   {}^{  { \color{black}{0} }  } \ottnt{B} \rightarrow    {}^{  { \color{black}{1} }  } \ottnt{A} \rightarrow  \ottnt{A}     \ottsym{)}   \qquad
   \Gamma_{{\mathrm{1}}} =  \Gamma_{{\mathrm{0}}} ,    \ottmv{x} \! :^{  { \color{black}{0} }  }\! \ottnt{B}     \qquad 
   \Gamma_{{\mathrm{2}}} =    \ottmv{f} \! :^{  { \color{black}{0} }  }\! \ottsym{(}   {}^{  { \color{black}{0} }  } \ottnt{B} \rightarrow    {}^{  { \color{black}{1} }  } \ottnt{A} \rightarrow  \ottnt{A}     \ottsym{)}   ,    \ottmv{x} \! :^{  { \color{black}{1} }  }\! \ottnt{B}    
%   \Delta_{{\mathrm{0}}} =  \lfloor \Gamma_{{\mathrm{0}}} \rfloor  \qquad
%   \Delta_{{\mathrm{1}}}  =  \lfloor \Gamma_{{\mathrm{1}}} \rfloor 
\end{gather*}
and construct a typing derivation for the application:
\[
\inferrule*[left=ST-App]
{
  \inferrule*[left=ST-Weak]
  { \inferrule*[left=ST-Var]{\ \ }{  \lfloor \Gamma_{{\mathrm{0}}} \rfloor  \ ;\  \Gamma_{{\mathrm{0}}}  \vdash \ottmv{f} :  {}^{  { \color{black}{0} }  } \ottnt{B} \rightarrow    {}^{  { \color{black}{1} }  } \ottnt{A} \rightarrow  \ottnt{A}    }  % {\mbox{\textsc{ST-Var}}} 
  }  
  {   \lfloor \Gamma_{{\mathrm{1}}} \rfloor  \ ;\  \Gamma_{{\mathrm{1}}}  \vdash \ottmv{f} :  {}^{  { \color{black}{0} }  } \ottnt{B} \rightarrow    {}^{  { \color{black}{1} }  } \ottnt{A} \rightarrow  \ottnt{A}     }
  %{\mbox{\textsc{ST-Weak}}}
  \qquad
  { \inferrule*[left=ST-Var]{\ \ }{  \lfloor \Gamma_{{\mathrm{1}}} \rfloor  \ ;\  \Gamma_{{\mathrm{2}}}  \vdash \ottmv{x} : \ottnt{B} }
  }
  %{\mbox{\textsc{ST-Var}}}
}
{ 
    \lfloor \Gamma_{{\mathrm{1}}} \rfloor  \ ;\  \Gamma_{{\mathrm{1}}}  \ottsym{+}    { \color{black}{0} }    \cdot   \Gamma_{{\mathrm{2}}}   \vdash \ottmv{f} \, \ottmv{x} :  {}^{  { \color{black}{1} }  } \ottnt{A} \rightarrow  \ottnt{A}  
}
%{\mbox{\textsc{ST-App}}}
\]

Working through the context expression $\Gamma_{{\mathrm{1}}}  \ottsym{+}    { \color{black}{0} }    \cdot   \Gamma_{{\mathrm{2}}} $, 
we see that the computed final
context, derived in the conclusion of the application rule is just $\Gamma_{{\mathrm{1}}}$
again. Although the variable $\ottmv{x}$ appears free in the expression
$\ottmv{f} \, \ottmv{x}$, because it is the
argument to a constant function here, this use does not contribute to the
overall result.
\end{example}

\subsection{Data Structures}

%% Because this type system subsumes linear types, we must include both an
%% introduction and an elimination form for every type form. This even
%% includes the $\ottkw{Unit}$ type. We see in the rule that the introduction form,
%% the $\ottkw{unit}$ value, requires no quantities. This value is eliminated by the
%% pattern matching expression $ \ottkw{let}\,  \ottkw{unit} \,=\, \ottnt{a} \ \ottkw{in}\  \ottnt{b} $, that consumes $\ottnt{a}$, a
%% $\ottkw{Unit}$ typed expression, and adds any quantity used in its computation to
%% those required by $\ottnt{b}$, the next expression to evaluate in sequence.

\paragraph{Unit}\

\centerline{\drule{ST-Unit}\drule{ST-UnitE}}
\vspace{2ex}

The $\ottkw{Unit}$ type has a single element, $\ottkw{unit}$. To eliminate a term of this type, we just match it with $\ottkw{unit}$. Since the elimination form requires the resources used for both the terms, we add the two contexts in the conclusion of \rref{ST-UnitE}.

\paragraph{The graded modal type}\ 

%\drules[ST]{$ \Delta \ ;\  \Gamma  \vdash \ottnt{a} : \ottnt{A} $}
%{Typing rules for graded modal type}
%{Box,LetBox}

\centerline{\drule{ST-Box}\drule{ST-LetBox}}
\vspace{2ex}

The type $ \Box^{ { \color{black}{q} } }  \ottnt{A} $ is called a \emph{graded modal type} or \emph{usage
  modal type}. It is introduced by the construct $ \ottkw{box} _ { \color{black}{q} } \, \ottnt{a} $, which uses the
expression ${ \color{black}{q} }$ times to build the box. This box can then be passed
around as an entity. When unboxed (\rref{ST-LetBox}), the continuation has
access to ${ \color{black}{q} }$ copies of the contents.

% \scw{Removing for space, and because it maybe confusing to just drop in here.}
% Modal types are captured by the axioms of modal logic, but in
% the case of usage modal types these axioms are indexed by quantities. 
%  For example,
% below are modal axioms from the modal logic S4 that are provable as implications:
% \[
% \begin{array}{rl}
%    {}^{ { \color{black}{q} }_{{\mathrm{1}}} } \ottsym{(}   \Box^{ { \color{black}{q} } }  \ottsym{(}   {}^{ { \color{black}{q} }_{{\mathrm{2}}} } \ottnt{a} \rightarrow  \ottnt{b}   \ottsym{)}   \ottsym{)} \rightarrow  \ottsym{(}   {}^{ { \color{black}{q} }_{{\mathrm{2}}} }  \Box^{ { \color{black}{q} } }  \ottnt{a}  \rightarrow   \Box^{ { \color{black}{q} } }  \ottnt{b}    \ottsym{)}  & \text{ (K)}\\
%    {}^{ { \color{black}{q} } }  \Box^{ { \color{black}{q} }_{{\mathrm{1}}}  \cdot  { \color{black}{q} }_{{\mathrm{2}}} }  \ottnt{a}  \rightarrow   \Box^{ { \color{black}{q} }_{{\mathrm{1}}} }   \Box^{ { \color{black}{q} }_{{\mathrm{2}}} }  \ottnt{a}    & \text{ (T)}\\
%    {}^{ { \color{black}{q} } }  \Box^{  { \color{black}{1} }  }  \ottnt{a}  \rightarrow  \ottnt{a}  & \text{ (4)}
% \end{array}
% \]

\paragraph{Products}\ 

\centerline{\drule{ST-Pair}\drule{ST-Spread}} 
\vspace{2ex}

The type system includes (multiplicative) products, also known as tensor
products. The two components of these pairs do not share variable
usages. Therefore the introduction rule adds the two contexts together. These
products must be eliminated via pattern matching because both components must
be used in the continuation. An elimination form that projects only
one component of the tuple would lose the usage constraints from the other
component. Note that even though both components of the tuple must be used
exactly once, by nesting a modal type within the tuple, programmers can
construct data structures with components of varying usage.

\paragraph{Sums}\ 

\centerline{\drule{ST-InjOne}\drule{ST-InjTwo}\drule{ST-Case}}
\vspace{2ex}

Last, the system includes (additive) sums and case analysis.  The introduction
rules for the first and second injections are no different from a standard
type system. However, in the elimination form, \rref{ST-Case}, the quantities
used for the scrutinee can be different than the quantities used (and shared
by) the two branches. Furthermore, the case expression may be annotated with a
quantity $q$ that indicates how many copies of the scrutinee may be demanded
in the branches.  Both branches of the case analysis \emph{must} use the
scrutinee at least once, as indicated by the $ { \color{black}{1} }   \leq  { \color{black}{q} }$ constraint. 

\subsection{Type Soundness}

For the language presented above, we define an entirely standard call-by-name
reduction relation $ \ottnt{a}  \leadsto  \ottnt{a'} $, included in \auxref{app:simple-opsem}.
%% \begin{figure}
%% \drules[S]{$ \ottnt{a}  \leadsto  \ottnt{a'} $}{Small-step reduction}{AppCong,Beta,UnitCong,UnitBeta}
%% \caption{Small-step call-by-name reduction (excerpt)}
%% \label{fig:step}
%% \end{figure}
With this operational semantics, a syntactic proof of type soundness follows
in the usual manner, via the entirely standard progress and preservation lemmas.
The substitution lemma that is part of this proof is of particular interest to us, as it must account
for the number of times the substituted variable ($\ottmv{x}$, in our
statement) is used when computing the contexts used in the conclusion of the lemma:
\begin{lemma}[Substitution] \label{SSub} If $ \Delta_{{\mathrm{1}}} \ ;\  \Gamma  \vdash \ottnt{a} : \ottnt{A} $ and
$  \Delta_{{\mathrm{1}}} ,     \ottmv{x} \!\!:\!\! \ottnt{A}  ,  \Delta_{{\mathrm{2}}}    \ ;\   \Gamma_{{\mathrm{1}}} ,     \ottmv{x} \! :^{ { \color{black}{q} } }\! \ottnt{A}  ,  \Gamma_{{\mathrm{2}}}     \vdash \ottnt{b} : \ottnt{B} $, then 
$  \Delta_{{\mathrm{1}}} ,  \Delta_{{\mathrm{2}}}  \ ;\   \Gamma_{{\mathrm{1}}}  \ottsym{+}    { \color{black}{q} }   \cdot   \Gamma   ,  \Gamma_{{\mathrm{2}}}   \vdash \ottnt{b}  \ottsym{\{}  \ottnt{a}  \ottsym{/}  \ottmv{x}  \ottsym{\}} : \ottnt{B} $.
\end{lemma}

%% With this substitution lemma, we can prove the usual preservation and progress theorems showing type soundness.

%% The preservation theorem shows that all quantities are preserved during
%% computation. The usage context is unchanged with each step.

%% \begin{theorem}[Preservation]
%% If $ \Delta \ ;\  \Gamma  \vdash \ottnt{a} : \ottnt{A} $ and $ \ottnt{a}  \leadsto  \ottnt{a'} $ then $ \Delta \ ;\  \Gamma  \vdash \ottnt{a'} : \ottnt{A} $.
%% \end{theorem}

%% Finally, the progress lemma states that in an empty context, if computation
%% has not finished, then the term is not stuck. (The values are shown in
%% Figure~\ref{fig:min}).

%% \begin{theorem}[Progress]
%%   If $ \varnothing \ ;\  \varnothing  \vdash \ottnt{a} : \ottnt{A} $ then either $a$ is a value or there exists
%%   some $a'$ such that $ \ottnt{a}  \leadsto  \ottnt{a'} $.
%% \end{theorem}

\subsection{Discussion and Variations}
\label{sec:variations}

At this point, the language that we have developed recalls
systems found in prior work, such as \citet{Brunel:2014}, \citet{orchard:2019}, \citet{Wood:2020}
and \citet{abel:icfp2020}. Most differences are cosmetic---especially in the
treatment of usage contexts. Of these, the most similar is the
concurrently developed \citet{abel:icfp2020}, which we compare below.

\begin{itemize}
\item First, \citet{abel:icfp2020} include a slightly more expressive form of pattern
  matching. Their elimination forms for the box modality and products 
  multiply each scrutinee by a quantity $q$, providing that many copies of its subcomponents to the
  continuation, as in our \rref{ST-Case}.
  For simplicity, we have omitted this feature; it is not difficult to add.
\item Second, \citet{abel:icfp2020} require that the semiring include least-upper bounds
  for the partial order of the semiring. This allows them to compose case branches with differing     usages.
  \item Third, in the rule for $\ottkw{case}$, like \citet{abel:icfp2020}, we need the
 requirement that $ { \color{black}{1} }   \leq  { \color{black}{q} }$. In our system as well as theirs, it turns out
 that this requirement is not motivated by the standard type soundness theorem: the theorem holds without it.  Their condition was instead motivated by their parametricity theorems.  Our condition is motivated by   the heap soundness theorem that we present in the next section.
\end{itemize}

The standard type soundness theorem is not very
informative because it does not show that the quantities are correctly
used. % Our usage-agnostic small-step reduction relation can go no further.
Therefore, to address this issue, we turn to a heap-based semantics, based on
\citet{launchbury} and \citet{turner}, to account for resource usage during
computation.

\section{Heap semantics for simple type system}
\label{sec:heap-semantics}

A heap semantics shows how a term evaluates when the free variables of the term are assigned other terms. The assignments are stored in a heap, represented here as an ordered list. We associate an \emph{allowed usage}, basically an abstract quantity of resources, to each assignment. We change these quantities as the evaluation progresses. For example, a typical call-by-name reduction goes like this:\footnote{We don't have $\ottkw{Int}$ type and $+$ function in our language, but we use them for the sake of explanation.}
\begin{align*}
& [x \stackrel{3}{\mapsto} 1, y \stackrel{1}{\mapsto} x + x] (x + y)  
&\mbox{\textit{look up value of $x$, decrement its usage}}
\\
\Rightarrow \, & [x \stackrel{2}{\mapsto} 1, y \stackrel{1}{\mapsto} x + x] 1 + y
&\mbox{\textit{look up value of $y$, decrement its usage}} 
\\
\Rightarrow \, & [x \stackrel{2}{\mapsto} 1, y \stackrel{0}{\mapsto} x + x] 1 + (x + x)
&\mbox{\textit{look up value of $x$, decrement its usage}}
\\
\Rightarrow \, & [x \stackrel{1}{\mapsto} 1, y \stackrel{0}{\mapsto} x + x] 1 + (1 + x)
&\mbox{\textit{look up value of $x$, decrement its usage}}
\\
\Rightarrow \, & [x \stackrel{0}{\mapsto} 1, y \stackrel{0}{\mapsto} x + x] 1 + (1 + 1)
&\mbox{\textit{addition step}}
\\
\Rightarrow \, & [x \stackrel{0}{\mapsto} 1, y \stackrel{0}{\mapsto} x + x] 3
\end{align*}

\subsection{The Step Judgement}

\begin{figure}
\drules[Small]{$ [  \ottnt{H}  ]\,  \ottnt{a}  \Rightarrow_{ \ottnt{S} }^{ { \color{black}{r} } } [  \ottnt{H'} \, ;\,  \mathbf{u}' \, ;\,  \Gamma'  ]\,  \ottnt{a'} $}
{Small-step reduction relation (excerpt)}
{Var,AppL,AppBeta,CaseL,CaseOne,CaseTwo,Sub}
\scw{Should we drop UnitL and UnitBeta from this figure?}
\caption{Heap semantics (excerpt)}
\label{fig:heap-semantics}
\begin{notation}
The notation $ \mathbf{0}^{ \ottnt{n} } $ denotes a vector of $0$'s of length $n$. The notation $ \mathbf{u}_{{\mathrm{1}}}  \mathop{\diamond}  \mathbf{u}_{{\mathrm{2}}} $ denotes concatenation. Here $ \,\text{fv}\,  \ottnt{a} $ stands for the free variables of $a$ while $ \mathsf{Var} \,  \ottnt{H} $ stands for the domain of $H$ and the free variables of the terms appearing in the assignments of $H$.
\end{notation}
\end{figure}

The reduction above is expressed informally as a sequence of pairs of heap $\ottnt{H}$ and expression $\ottnt{a}$. We formalize this relation using the following judgement, which appears in Fig \ref{fig:heap-semantics}. 
\[
 [  \ottnt{H}  ]\,  \ottnt{a}  \Rightarrow_{ \ottnt{S} }^{ { \color{black}{r} } } [  \ottnt{H'} \, ;\,  \mathbf{u}' \, ;\,  \Gamma'  ]\,  \ottnt{a'} 
\]

The meaning of this relation is that ${ \color{black}{r} }$ copies of the term $a$ use the
resources of the heap $H$ and step to $r$ copies of the term $a'$, with $H'$
being the new heap. The relation also maintains additional information, which
we explain below.

Heap assignments are of the form $ \ottmv{x}  \overset{ { \color{black}{q} } }{\mapsto} { \Gamma \vdash  \ottnt{a}  :  \ottnt{A} } $, associating an
\textit{assignee variable} with its \textit{allowed usage} $q$ and assignment
$a$. The \textit{embedded context} $\Gamma$ and type $\ottnt{A}$ help in the proof of our soundness
theorem (\ref{thm:heap-soundness}).  For a heap $H$, we use
$\lfloor H \rfloor$ to represent $H$ excluding the allowed usages and
$\llfloor H \rrfloor$ to represent just the list of underlying assignments. We
call $\lfloor H \rfloor$ and $\llfloor H \rrfloor$ the \textit{erased} and
\textit{bare} views of $H$ respectively. For example, for
$H = [  \ottmv{x}  \overset{ { \color{black}{q} } }{\mapsto} { \Gamma \vdash  \ottnt{a}  :  \ottnt{A} }  ]$, the erased view
$\lfloor H \rfloor = [ x \mapsto  \Gamma  \vdash \ottnt{a} : \ottnt{A}  ]$ and the bare view
$\llfloor H \rrfloor = [ x \mapsto a ]$ and $\Gamma$ is the embedded context. The vector of allowed usages of the
variables in $H$ is denoted by $ \ncoverline{ \ottnt{H} } $.

Because we use a call-by-name reduction, we don't evaluate the terms in the heap; we just modify the quantities associated with the assignments as they are retrieved. Therefore, after any step, $\ottnt{H'}$ will contain all the previous assignments of $\ottnt{H}$, possibly with different usages. Furthermore, a beta-reduction step may also add new assignments to $H$.
To allocate new variable names appropriately, we need a support set $\ottnt{S}$ in this relation; fresh names are chosen avoiding the variables in this set.
We keep track of these new variables that are added to the heap along with the allowed usages of their assignments
using the added context $\Gamma'$.\footnote{Instead of full contexts $\Gamma'$, we could have just used a list of variable/usage pairs here; but we pass dummy types along with them for ease of presentation later.}
Therefore, after a step $ [  \ottnt{H}  ]\,  \ottnt{a}  \Rightarrow_{ \ottnt{S} }^{ { \color{black}{r} } } [  \ottnt{H'} \, ;\,  \mathbf{u}' \, ;\,  \Gamma'  ]\,  \ottnt{a'} $, the length of $\ottnt{H'}$ is the sum of the lengths of $\ottnt{H}$ and $\Gamma'$. 

Now, because we work with an arbitrary semiring (possibly without subtraction), this heap semantics is non-deterministic. For example, consider a step $[x \stackrel{q}{\mapsto} a] x \Rightarrow [x \stackrel{q'}{\mapsto} a] a$, where $q = q' + 1$. Here, we are using $x$ once, so we need to reduce its usage by $1$. But in an arbitrary semiring, there may exist multiple new quantities, $q'' \neq q'$, such that $q = q' + 1 = q'' + 1$. For example, in the linearity semiring, we have $\omega = 1 + 1 = \omega + 1$. In this case, $[x \stackrel{\omega}{\mapsto} a] x \Rightarrow [x \stackrel{1}{\mapsto} a] a$ and $[x \stackrel{\omega}{\mapsto} a] x \Rightarrow [x \stackrel{\omega}{\mapsto} a] a$.

The absence of subtraction also means that given an initial heap and a final
heap, we really don't know how much resources have been used by the
computation. The only way to know this is to keep track of resources while
they are being used. The amount of resources used up can be expressed as a
quantity vector $\mathbf{u}'$ called \textit{consumption vector}, with its
components showing usage at the corresponding variables in $\ottnt{H'}$. (The
length of $\mathbf{u}'$ will always be the same as $\ottnt{H'}$.)

Finally, owing to the presence of $ \ottkw{case} $ expressions that can use the
scrutinee more than once, we need to be able to evaluate several copies of the
scrutinee in parallel before passing them on to the appropriate branch. So we
step ${ \color{black}{r} }$ copies of a term $a$ in parallel to get ${ \color{black}{r} }$ copies of
$\ottnt{a'}$. We call ${ \color{black}{r} }$ the \textit{copy quantity} of the step. For the
most part, we shall be interested in copy quantity of $1$. 
\begin{notation}
We use $ [  \ottnt{H}  ]\,  \ottnt{a}  \Rightarrow_{ \ottnt{S} } [  \ottnt{H'} \, ;\,  \mathbf{u}' \, ;\,  \Gamma'  ]\,  \ottnt{a'} $ to denote $ [  \ottnt{H}  ]\,  \ottnt{a}  \Rightarrow_{ \ottnt{S} }^{  { \color{black}{1} }  } [  \ottnt{H'} \, ;\,  \mathbf{u}' \, ;\,  \Gamma'  ]\,  \ottnt{a'} $.
\end{notation}

\subsection{Reduction Relation}

Figure~\ref{fig:heap-semantics} contains an excerpt of the reduction relation. They mirror the ordinary small-step rules, but there are some crucial differences. For example, this relation includes \rref{Small-Var} that allows a variable look-up, provided its usage permits. The look-up consumes the copy quantity from the allowed usage. But the copy quantity cannot be arbitrary here -- we are consuming the resource at least once. So we restrict the copy quantity to be $1$ or more. This is the only rule that modifies the usage of an existing variable in the heap.

In this relation, \rref{AppBeta} loads new assignment into the heap instead of immediate substitution. The substitution happens in steps through variable look-ups. To avoid conflict, we choose new variables excluding the ones already in use. Since we are evaluating $r$ copies, we set the allowed usage of the variable to ${ \color{black}{r} }  \cdot  { \color{black}{q} }$ where $q$ is the usage annotation on the term.

Let us look at \rref{Small-CaseL}. This is interesting since the copy quantity in the premise and the conclusion are different. In fact, we introduced copy quantity to properly handle usages while evaluating $ \ottkw{case} $ expressions. For evaluating ${ \color{black}{r} }$ copies of the $ \ottkw{case} $ expression, we need to evaluate ${ \color{black}{r} }  \cdot  { \color{black}{q} }$ copies of the scrutinee since the scrutinee gets used ${ \color{black}{q} }$ times in either branch.

The \rref{Small-Sub} reduces the allowed usages in the heap and then lets the term take a step. Here, we use $ \ottnt{H_{{\mathrm{1}}}}  \leq  \ottnt{H_{{\mathrm{2}}}} $ to mean $\lfloor \ottnt{H_{{\mathrm{1}}}} \rfloor = \lfloor \ottnt{H_{{\mathrm{2}}}} \rfloor$ and for corresponding pair of assignments $ \ottmv{x}  \overset{ { \color{black}{q} }_{{\mathrm{1}}} }{\mapsto} { \Gamma \vdash  \ottnt{a}  :  \ottnt{A} } $ and $ \ottmv{x}  \overset{ { \color{black}{q} }_{{\mathrm{2}}} }{\mapsto} { \Gamma \vdash  \ottnt{a}  :  \ottnt{A} } $ in $\ottnt{H_{{\mathrm{1}}}}$ and $\ottnt{H_{{\mathrm{2}}}}$ respectively, the condition ${ \color{black}{q} }_{{\mathrm{1}}}  \leq  { \color{black}{q} }_{{\mathrm{2}}}$ holds. 

The multi-step reduction relation is the transitive closure of the single-step relation. In \rref{Multi-Many}, the consumption vectors from the steps are added up and the added contexts of new variables are concatenated. The copy quantity is the same in both the premises and the conclusion since the rule represents parallel multi-step evaluation of ${ \color{black}{r} }$ copies.

\drules[Multi]{$ [  \ottnt{H}  ]\,  \ottnt{a}  \Rightarrow\!\!\!\!\!\Rightarrow^{ { \color{black}{r} } }_{ \ottnt{S} } [  \ottnt{H'} \, ;\,  \mathbf{u}' \, ;\,  \Gamma  ]\,  \ottnt{b} $}
{Multi-Step relation}
{One,Many}

\subsection{Accounting of Resources}

The reduction relation enforces fair usage of resources, leading to the following theorem.

\begin{theorem}[Conservation]
If $ [  \ottnt{H}  ]\,  \ottnt{a}  \Rightarrow\!\!\!\!\!\Rightarrow^{ { \color{black}{r} } }_{ \ottnt{S} } [  \ottnt{H'} \, ;\,  \mathbf{u}' \, ;\,  \Gamma'  ]\,  \ottnt{a'} $, then $  \ncoverline{ \ottnt{H'} }   +  \mathbf{u}'  \leq   \ncoverline{ \ottnt{H} }   \mathop{\diamond}   \ncoverline{ \Gamma' }  $.
\end{theorem}

Here, $ \ncoverline{ \ottnt{H} } $ represents the initial resources and $ \ncoverline{ \Gamma' } $ represents the newly added resources; whereas $ \ncoverline{ \ottnt{H'} } $ represents the resources left and $\mathbf{u}'$ the resources that were consumed. So the theorem says that the initial resources concatenated with those that are added during evaluation, are equal to or more than the remaining resources plus those that were used up. Note that if the partial order is the trivial reflexive order, $ \leq $ becomes an equality. In such a scenario, the reduction relation enforces strict conservation of resources. More generally, this theorem states that we don't use more resources than what we are entitled to.\\\

Unlike the substitution-based semantics, in this heap semantics, terms can ``get stuck'' due to lack of resources. Let us look at the following evaluation:
\begin{align*}
& [x \stackrel{2}{\mapsto} 1, y \stackrel{1}{\mapsto} x + x] (x + y)  
&\mbox{\textit{look up value of $x$, decrement its usage}}
\\
\Rightarrow \, & [x \stackrel{1}{\mapsto} 1, y \stackrel{1}{\mapsto} x + x] 1 + y
&\mbox{\textit{look up value of $y$, decrement its usage}} 
\\
\Rightarrow \, & [x \stackrel{1}{\mapsto} 1, y \stackrel{0}{\mapsto} x + x] 1 + (x + x)
&\mbox{\textit{look up value of $x$, decrement its usage}}
\\
\Rightarrow \, & [x \stackrel{0}{\mapsto} 1, y \stackrel{0}{\mapsto} x + x] 1 + (1 + x)
&\mbox{\textit{look up value of $x$, stuck!}}
\end{align*}

The evaluation gets stuck because the starting heap does not contain enough resources for the evaluation of the term. The term needs to use $x$ thrice; whereas the heap contains only two copies of $x$.

But this is not the only way in which an evaluation can run out of resources. Such a situation may also happen through ``unwise usage'', even when the starting heap contains enough resources. For example, over the linearity semiring, the evaluation: $[ x \stackrel{\omega}{\mapsto} 5 ] x + (x + x) \Rightarrow [ x \stackrel{1}{\mapsto} 5 ] 5 + (x + x) \Rightarrow [ x \stackrel{0}{\mapsto} 5 ] 5 + (5 + x)$ gets stuck because in the first step, $\omega$ was ``unwisely'' split as $1 + 1$ instead of being split as $\omega + 1$.

Our aim, then, is to show that given a heap that contains enough resources, a well-typed term that is not a value, can always take a step such that the resulting heap contains enough resources for the evaluation of the resulting term. We shall formalize what it means for a heap to contain enough resources. But before that, let us explore the relationship between the various possible steps a term can take when provided with a heap. \scw{Should have a forward reference to the soundness theorem here.}

\subsection{Determinism and Alpha-Equivalence}

Earlier, we pointed out that the step relation is non-deterministic. But on a closer look, we find that the non-determinism is limited more or less to the usages. If a term steps in two different ways when provided with a heap, the resulting terms are the same; the resulting heaps, though, may have different allowed usage vectors. Here, we formulate a precise version of this statement.

A term, when provided with a heap, can step either by looking up a variable or by adding a new assignment. \scw{Maybe we don't have to talk about proper heaps} Now, if the heap does not contain duplicate assignments for the same variable, look-up will always produce the same result. We call such heaps \textit{proper}. Note that the reduction relation maintains this property of heaps. So hereafter, we restrict our attention to proper heaps. Next, if a term steps by adding a new assignment, we may choose different fresh variables leading to different resulting terms. But such a difference is reconcilable. Viewed as closures, such heap term pairs are $\alpha$-equivalent.

Given a heap $H$ and a term $a$, let us call $(\llfloor H \rrfloor , a)$ a \textit{machine configuration}. Two heap term pairs are $\alpha$-equivalent if the corresponding machine configurations are identical up to systematic renaming of assignee variables. We denote $\alpha$-equivalence by $\sim_{\alpha}$.

The step relation, then, is deterministic in the following sense: 
 
\begin{lemma}[Determinism]
If $ [  \ottnt{H_{{\mathrm{1}}}}  ]\,  \ottnt{a_{{\mathrm{1}}}}  \Rightarrow_{ \ottnt{S} }^{ { \color{black}{r} }_{{\mathrm{1}}} } [  \ottnt{H'_{{\mathrm{1}}}} \, ;\,  \mathbf{u}'_{{\mathrm{1}}} \, ;\,  \Gamma'_{{\mathrm{1}}}  ]\,  \ottnt{a'_{{\mathrm{1}}}} $ and $ [  \ottnt{H_{{\mathrm{2}}}}  ]\,  \ottnt{a_{{\mathrm{2}}}}  \Rightarrow_{ \ottnt{S} }^{ { \color{black}{r} }_{{\mathrm{2}}} } [  \ottnt{H'_{{\mathrm{2}}}} \, ;\,  \mathbf{u}'_{{\mathrm{2}}} \, ;\,  \Gamma'_{{\mathrm{2}}}  ]\,  \ottnt{a'_{{\mathrm{2}}}} $ and $ ( \ottnt{H_{{\mathrm{1}}}} , \ottnt{a_{{\mathrm{1}}}} ) \sim_{\alpha} ( \ottnt{H_{{\mathrm{2}}}}  , \ottnt{a_{{\mathrm{2}}}} ) $, then $ ( \ottnt{H'_{{\mathrm{1}}}} , \ottnt{a'_{{\mathrm{1}}}} ) \sim_{\alpha} ( \ottnt{H'_{{\mathrm{2}}}}  , \ottnt{a'_{{\mathrm{2}}}} ) $.
\end{lemma}

The lemma above says that, if a term, when provided with a heap, takes a step in two different ways, then the resulting heap term pairs are basically the same. We know that the ordinary small-step semantics is deterministic. The inclusion of allowed usages in the heap semantics is not to produce multiple reducts but just to block evaluation at the point where consumption reaches its permitted limit. This is an important point and needs elaboration.

Call a reduction consisting of $n$ steps an $n$-chain reduction. Also, for a reduction $ [  \ottnt{H}  ]\,  \ottnt{a}  \Rightarrow_{ \ottnt{S} }^{ { \color{black}{r} } } [  \ottnt{H'} \, ;\,  \mathbf{u}' \, ;\,  \Gamma'  ]\,  \ottnt{a'} $, call $ [ \llfloor H \rrfloor ] a \Rightarrow [ \llfloor H' \rrfloor ] a'$ the \textit{machine view} of reduction. Then, the machine view of every $n$-chain reduction of a term in a heap is the same, modulo $\alpha$-equivalence. So, if there \textit{exists} an $n$-chain reduction of $a$ to $a'$, starting with heap $H$, we know that there is a way by which $a$ can reduce to $a'$ without running out of resources, implying the validity of the reduction. In such a scenario, we may as well forget all the usage annotations and evaluate $a$ for $n$ steps starting with $\llfloor H \rrfloor$. By the above lemma, such an evaluation in this machine environment is deterministic and hence unique. The reduced heap term pair that we get is the same (modulo $\alpha$-equivalence). Along with the soundness theorem, this shall give us a deterministic reduction strategy that is correct.

Now that we see the equivalence of all the possible reducts, we explore its relation with the ordinary small-step reduct.

\subsection{Bisimilarity}

The ordinary and the heap-based reduction relations are bisimilar in a way we make precise below. To compare, we need to define some terms. We call a heap \textit{acyclic} iff the term assigned to a variable does not refer to itself or to any other variable appearing subsequently in the heap. Note that the reduction relation preserves acyclicity. Hereafter, we restrict our attention to proper, acyclic heaps.

Now, for a heap $H$, define $a\{ H \}$ as the term obtained by substituting in $a$, in reverse order, the corresponding terms for the variables in the heap. Then we have the following lemmas: 

\begin{lemma} \label{HOrd}
If $ [  \ottnt{H}  ]\,  \ottnt{a}  \Rightarrow_{ \ottnt{S} }^{ { \color{black}{r} } } [  \ottnt{H'} \, ;\,  \mathbf{u}' \, ;\,  \Gamma'  ]\,  \ottnt{a'} $, then $ \ottnt{a}  \{  \ottnt{H}  \}   \ottsym{=}   \ottnt{a'}  \{  \ottnt{H'}  \} $ or $  \ottnt{a}  \{  \ottnt{H}  \}   \leadsto   \ottnt{a'}  \{  \ottnt{H'}  \}  $. Further, if $ [  \varnothing  ]\,  \ottnt{a}  \Rightarrow_{ \ottnt{S} }^{ { \color{black}{r} } } [  \ottnt{H'} \, ;\,  \mathbf{u}' \, ;\,  \Gamma'  ]\,  \ottnt{a'} $, then $ \ottnt{a}  \leadsto   \ottnt{a'}  \{  \ottnt{H'}  \}  $.
\end{lemma}

\begin{lemma} \label{OrdH}
If $ \ottnt{a}  \leadsto  \ottnt{a_{{\mathrm{1}}}} $, then for a heap $H$, we have $ [  \ottnt{H}  ]\,  \ottnt{a}  \Rightarrow_{ \ottnt{S} }^{ { \color{black}{r} } } [   \ottnt{H}  ,  \ottnt{H'}  \, ;\,  \mathbf{u} \, ;\,  \Gamma  ]\,  \ottnt{a_{{\mathrm{2}}}} $ where $ \ottnt{a_{{\mathrm{2}}}}  \{  \ottnt{H'}  \}   \ottsym{=}  \ottnt{a_{{\mathrm{1}}}}$.
\end{lemma} 

\scw{This paragraph is getting redundant. We can just cut it.}
The heap reduction relation splits the ordinary $\beta$-reduction rules into an assignment addition rule and a variable look-up rule. This enables the heap-based rules to substitute one occurrence of a variable at a time while the ordinary $\beta$-rules substitute all occurrences of a variable at once. If we perform substitution immediately after loading a new assignment to the heap, then the heap-based rules and the ordinary step rules are essentially the same. The above lemmas formalize this idea.

The heap-based rules substitute one occurrence of a variable at a time and keep track of usage and obstruct unfair usage. With this constraint in place, we ought to know how much resources shall be necessary before evaluating a term. This will tell us how much resources the starting heap should contain. The type system helps us know this as we see next.  

\subsection{Heap Compatibility}

The key idea behind this language design is that, if the resources contained in a heap are judged to be ``right'' for a term by the type system, the evaluation of the term in such a heap does not get stuck. With the heap-based reduction rules enforcing fairness of usage, this would mean that the type system does a proper accounting of the resource usage of terms. 

The compatibility relation $ \ottnt{H}  \vdash  \Delta ;  \Gamma $ presented below expresses the judgement that the heap $H$ contains enough resources to evaluate any term that type-checks in the usage context $\Gamma$. A heap that is compatible with some context is called a \textit{well-formed} heap. 
\drules[Compat]{$ \ottnt{H}  \vdash  \Delta ;  \Gamma $}
{Heap Compatibility}
{Empty,Cons}

The \rref{Compat-Cons} rule reminds us of the substitution lemma \ref{SSub}. In a way, this rule is converse of the substitution lemma. It loads $q$ potential single-substitutions into the heap and lets the context use the variable $q$ times.

\begin{example} \label{heap-ex}  Consider the following derivation:\footnote{For simplicity, we omit the $\Delta$s from the compatibility and the typing judgements.}
\[
\inferrule*[]
{\inferrule*[]
  {\inferrule*[]{
    \varnothing   \vdash   \varnothing  \\
    \varnothing  \vdash  1  : \ottkw{Int} }
  {  \ottmv{x_{{\mathrm{1}}}}  \overset{ \ottsym{7} }{\mapsto}   1     \vdash    \ottmv{x_{{\mathrm{1}}}} \! :^{ \ottsym{7} }\! \ottkw{Int}  } \\
     \ottmv{x_{{\mathrm{1}}}} \! :^{  { \color{black}{2} }  }\! \ottkw{Int}   \vdash  \ottmv{x_{{\mathrm{1}}}}  +  \ottmv{x_{{\mathrm{1}}}}  : \ottkw{Int} }
  {   \ottmv{x_{{\mathrm{1}}}}  \overset{ \ottsym{7} }{\mapsto}   1    ,   \ottmv{x_{{\mathrm{2}}}}  \overset{ \ottsym{3} }{\mapsto}   \ottmv{x_{{\mathrm{1}}}}  +  \ottmv{x_{{\mathrm{1}}}}      \vdash     \ottmv{x_{{\mathrm{1}}}} \! :^{  { \color{black}{1} }  }\! \ottkw{Int}  ,   \ottmv{x_{{\mathrm{2}}}} \! :^{ \ottsym{3} }\! \ottkw{Int}   } \\
      \ottmv{x_{{\mathrm{1}}}} \! :^{  { \color{black}{1} }  }\! \ottkw{Int}  ,   \ottmv{x_{{\mathrm{2}}}} \! :^{  { \color{black}{2} }  }\! \ottkw{Int}    \vdash  \ottmv{x_{{\mathrm{1}}}}  +  \ottsym{(}   \ottmv{x_{{\mathrm{2}}}}  +  \ottmv{x_{{\mathrm{2}}}}   \ottsym{)}  : \ottkw{Int} 
}
{     \ottmv{x_{{\mathrm{1}}}}  \overset{ \ottsym{7} }{\mapsto}   1     ,      \ottmv{x_{{\mathrm{2}}}}  \overset{ \ottsym{3} }{\mapsto}   \ottmv{x_{{\mathrm{1}}}}  +  \ottmv{x_{{\mathrm{1}}}}     ,    \ottmv{x_{{\mathrm{3}}}}  \overset{  { \color{black}{1} }  }{\mapsto}   \ottmv{x_{{\mathrm{1}}}}  +  \ottsym{(}   \ottmv{x_{{\mathrm{2}}}}  +  \ottmv{x_{{\mathrm{2}}}}   \ottsym{)}         \vdash     \ottmv{x_{{\mathrm{1}}}} \! :^{  { \color{black}{0} }  }\! \ottkw{Int}  ,     \ottmv{x_{{\mathrm{2}}}} \! :^{  { \color{black}{1} }  }\! \ottkw{Int}  ,   \ottmv{x_{{\mathrm{3}}}} \! :^{  { \color{black}{1} }  }\! \ottkw{Int}     
}
\]
\end{example}

The context $ \ottmv{x_{{\mathrm{1}}}} \! :^{ \ottsym{7} }\! \ottkw{Int} $ splits its resources amongst derivations of $\ottmv{x_{{\mathrm{2}}}}  \ottsym{=}   \ottmv{x_{{\mathrm{1}}}}  +  \ottmv{x_{{\mathrm{1}}}} $ (thrice) and $\ottmv{x_{{\mathrm{3}}}}  \ottsym{=}   \ottmv{x_{{\mathrm{1}}}}  +  \ottsym{(}   \ottmv{x_{{\mathrm{2}}}}  +  \ottmv{x_{{\mathrm{2}}}}   \ottsym{)} $ (once). The heap keeps a record, in the form of allowed usages, of how the context gets split. A heap compatible with a context, therefore, satisfies the resource demands of a term derived in this context.
 
We pointed out earlier that the \rref{Compat-Cons} is like a converse substitution lemma. The following lemma formalizes this idea:
\begin{lemma}[Multi-substitution]
\label{multisub}
If $ \ottnt{H}  \vdash  \Delta ;  \Gamma $ and $ \Delta \ ;\  \Gamma  \vdash \ottnt{a} : \ottnt{A} $, then $ \varnothing \ ;\  \varnothing  \vdash  \ottnt{a}  \{  \ottnt{H}  \}  : \ottnt{A} $.
\end{lemma}
Because \rref{Compat-Cons} leads to expansion (or reverse substitution), we can re-substitute while maintaining well-typedness. The compatibility relation is crucial to our development. So we explore it in more detail below. % The details follow.

\subsection{Graphical and Algebraic Views of the Heap}

A heap can be viewed as a memory graph where the assignee variables correspond to memory locations and the assigned terms to data stored in those locations. The allowed usage then, is the number of ways the location can be referenced. This gives us a graphical view of the heap.

A heap $H$ where $ \ottnt{H}  \vdash  \Delta ;  \Gamma $ can be viewed as a weighted directed acyclic graph $G_{H,\Gamma}$. Let $H$ contain $n$ assignments with the $j^{\text{th}}$ one being $ \ottmv{x_{\ottmv{j}}}  \overset{ { \color{black}{q} }_{\ottmv{j}} }{\mapsto} { \Gamma_{\ottmv{j}} \vdash  \ottnt{a_{\ottmv{j}}}  :  \ottnt{A_{\ottmv{j}}} } $.  Then, $G_{H,\Gamma}$ is a DAG with $(n + 1)$ nodes, $n$ nodes corresponding to the $n$ variables in $H$ and one extra node for $\Gamma$, referred to as the source node. Let $v_j$ be the node corresponding to $x_j$ and $v_g$ be the source node. For $x_{i} :^{q_{ji}} A_i \in \Gamma_j$ (where $i < j$) add an edge with weight $w(v_j,v_i) := q_{ji}$ from $v_j$ to $v_i$. (Note that $\Gamma_j$ only contains variables $\ottmv{x_{{\mathrm{1}}}}$ through $x_{j-1}$.) We do this for all nodes, including $v_g$. This gives us a DAG with the topological ordering $v_g, v_n, v_{n-1}, \ldots, v_2, v_1$. 

For example \ref{heap-ex}, we have the following memory graph\footnote{We omit the $0$ weight edge from $v_g$ to $v_1$.}:

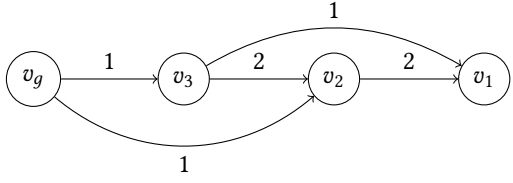
\begin{center}
\begin{tikzpicture}
    \node[shape=circle,draw=black] (A) at (0,0) {$v_g$};
    \node[shape=circle,draw=black] (C) at (2,0) {$v_3$};
    \node[shape=circle,draw=black] (D) at (4,0) {$v_2$};
    \node[shape=circle,draw=black] (E) at (6,0) {$v_1$};

    \path [->](A) edge node[above] {$1$} (C);
    \path [->](A) edge [bend right=40] node[below] {$1$} (D);
    
    \path [->](C) edge node[above] {$2$} (D);
    \path [->](C) edge [bend left] node[above] {$1$} (E);
    
    \path [->](D) edge node[above] {$2$} (E);   
\end{tikzpicture}
\end{center}

For a heap compatible with some context, we can express the allowed usages of the assignee variables in terms of the edge weights of the memory graph. Let us define the length of a path to be the product of the weights along the path. Then, the allowed usage of a variable is the sum of the lengths of all paths from the source node to the node corresponding to that variable. Note that this is so for the example graph.

A path $p$ from $v_g$ to $v_j$ represents a chain of references, with the last one being pointed at $v_j$. The length of $p$ shows how many times this path is used to reference $v_j$. The sum of the lengths of all the paths from $v_g$ to  $v_j$ then gives a (static) count of the total number of times location $v_j$ is referenced. And this is equal to $q_j$, the allowed usage of the assignment for $v_j$ in the heap. This means that the allowed usage of an assignment is equal to the (static) count of the number of times the concerned location is referenced. So, we also call $q_j$ the \textit{count} of $v_j$ and call this property count balance. Below, we present an algebraic formalization of this property of well-formed heaps.

\begin{notation}
We use $ \mathbf{0} $ to denote a row vector of $0$s of length $n$ (when $n$ is clear from the context) and use $ \mathbf{0} ^\intercal$ to denote a column vector of $0$s.
\end{notation}

For a well-formed heap $H$ containing $n$ assignments of the form $ \ottmv{x_{\ottmv{i}}}  \overset{ { \color{black}{q} }_{\ottmv{i}} }{\mapsto} { \Gamma_{\ottmv{i}} \vdash  \ottnt{a_{\ottmv{i}}}  :  \ottnt{A_{\ottmv{i}}} } $, we write $\langle H \rangle$ to denote the $n \times n$ matrix whose $i^{\text{th}}$ row is $  \ncoverline{ \Gamma_{\ottmv{i}} }   \mathop{\diamond}   \mathbf{0}  $. We call $\langle H \rangle$ the transformation matrix corresponding to $H$.
The transformation matrix for example \ref{heap-ex} is:\\
\[
\begin{pmatrix}
0 & 0 & 0 \\
2 & 0 & 0 \\
1 & 2 & 0
\end{pmatrix}
\]

For a well-formed heap $H$, the matrix $\langle H \rangle$ is strictly lower triangular. Note that this is also the adjacency matrix of the memory graph, excluding node $v_g$. The strict lower triangular property of the matrix corresponds to the acyclicity of the graph. With the matrix operations over a semiring  defined in the usual way, the count balance property is:
\begin{lemma}[Count Balance]
If $ \ottnt{H}  \vdash  \Delta ;  \Gamma $, then $ \ncoverline{ \ottnt{H} }  =  \ncoverline{ \ottnt{H} }  \times \langle H \rangle  +  \ncoverline{ \Gamma } $. 
\end{lemma}

\begin{proof}
We show this by induction on $ \ottnt{H}  \vdash  \Delta ;  \Gamma $. The base case is trivial.\\

For the Cons-case, let $  \ottnt{H'}  ,   \ottmv{x}  \overset{ { \color{black}{q} } }{\mapsto} { \Gamma_{{\mathrm{2}}} \vdash  \ottnt{a}  :  \ottnt{A} }    \vdash   \Delta' ,   \ottmv{x} \!\!:\!\! \ottnt{A}   ;   \Gamma_{{\mathrm{1}}} ,   \ottmv{x} \! :^{ { \color{black}{q} } }\! \ottnt{A}   $ where $ \ottnt{H'}  \vdash  \Delta' ;  \Gamma_{{\mathrm{1}}}  \ottsym{+}  \ottsym{(}   { \color{black}{q} }   \cdot   \Gamma_{{\mathrm{2}}}   \ottsym{)} $. By inductive hypothesis, $ \ncoverline{ \ottnt{H'} }  =  \ncoverline{ \ottnt{H'} }  \times \langle H' \rangle +  \ncoverline{ \Gamma_{{\mathrm{1}}}  \ottsym{+}  \ottsym{(}   { \color{black}{q} }   \cdot   \Gamma_{{\mathrm{2}}}   \ottsym{)} } $. Therefore, $  \ncoverline{ \ottnt{H'} }   \mathop{\diamond}   { \color{black}{q} }   =   \ncoverline{ \ottnt{H'} }   \mathop{\diamond}   { \color{black}{q} }   \times \big( \begin{smallmatrix} \langle H' \rangle &  \mathbf{0} ^\intercal \\  \ncoverline{ \Gamma_{{\mathrm{2}}} }  & 0 \end{smallmatrix} \big) +   \ncoverline{ \Gamma_{{\mathrm{1}}} }   \mathop{\diamond}   { \color{black}{q} }  $.  
\end{proof}

For example \ref{heap-ex}, we can check that $ \ncoverline{ \ottnt{H} }  = \begin{pmatrix}7&3&1\end{pmatrix}$ satisfies the above equation. Let us understand this equation. For a node $v_i$ in $G_{H,\Gamma}$, we can express the count $q_i$ in terms of the counts of the incoming neighbours and the weights of the corresponding edges. We have, $q_i = \Sigma_{j} \, q_j w(v_j,v_i) + w(v_g,v_i)$. The right-hand side of this equation represents static estimate of demand, the amount of resources we shall need while the left-hand side represents static estimate of supply, the amount of resources we shall have. So $ \ottnt{H}  \vdash  \Delta ;  \Gamma $ is a static guarantee that the heap $H$ shall supply the resource demands of the context $\Gamma$.

Therefore, if $ \ottnt{H}  \vdash  \Delta ;  \Gamma $ and $ \Delta \ ;\  \Gamma  \vdash \ottnt{a} : \ottnt{A} $, we should be able to evaluate $a$ in $H$ without running out of resources. This is the gist of the soundness theorem.

\subsection{Soundness}

\scw{We shouldn't say that $S$ is sufficiently fresh -- it is the set of variables to avoid, so it needs to contain all of the variables in dom D}
\begin{theorem}[Soundness]
\label{thm:heap-soundness}
If $ \ottnt{H}  \vdash  \Delta ;  \Gamma $ and $ \Delta \ ;\  \Gamma  \vdash \ottnt{a} : \ottnt{A} $ and $S \supseteq  \mathsf{dom} \,  \Delta $, then either $a$ is a value or there exists $\Gamma'$, $\ottnt{H'}$, $\mathbf{u}'$, $\Gamma_{{\mathrm{4}}}$ such that:
\begin{itemize}
\item $ [  \ottnt{H}  ]\,  \ottnt{a}  \Rightarrow_{ \ottnt{S} } [  \ottnt{H'} \, ;\,  \mathbf{u}' \, ;\,  \Gamma_{{\mathrm{4}}}  ]\,  \ottnt{a'} $
\item $ \ottnt{H'}  \vdash   \Delta ,   \lfloor \Gamma_{{\mathrm{4}}} \rfloor   ;  \Gamma' $
\item $  \Delta ,   \lfloor \Gamma_{{\mathrm{4}}} \rfloor   \ ;\  \Gamma'  \vdash \ottnt{a'} : \ottnt{A} $
\item $  \ncoverline{ \Gamma' }   +  \mathbf{u}'  +   \mathbf{0}   \mathop{\diamond}   \ncoverline{ \Gamma_{{\mathrm{4}}} }   \times \langle H'\rangle \leq   \ncoverline{ \Gamma }   \mathop{\diamond}   \mathbf{0}   + \mathbf{u}' \times \langle H' \rangle +   \mathbf{0}   \mathop{\diamond}   \ncoverline{ \Gamma_{{\mathrm{4}}} }  $
\end{itemize}
\end{theorem}

The soundness theorem\footnote{We present the proof for its dependent counterpart later.} states that our computations can go forward with the available resources without ever getting stuck. Note that as the term $a$ steps to $a'$, the typing context changes from
$\Gamma$ to $\Gamma'$. This is to be expected because during the step, resources
from the heap may have been consumed or new resources may have been added.
 For
example,
$[x \stackrel{1}{\mapsto} \ottkw{unit} ]x \Rightarrow [x \stackrel{0}{\mapsto}
\ottkw{unit} ] \ottkw{unit}$ and $  \ottmv{x} \! :^{  { \color{black}{1} }  }\! \ottkw{Unit}   \vdash \ottmv{x} : \ottkw{Unit} $ while $  \ottmv{x} \! :^{  { \color{black}{0} }  }\! \ottkw{Unit}   \vdash \ottkw{unit} : \ottkw{Unit} $. Though the typing context may change, the new context, which type-checks the reduct,
must be compatible with the new heap. This means that we can apply the soundness theorem again and again until we reach a value. At every step of the evaluation, the dynamics of our 
language aligns perfectly with the statics of the language. Graphically, as the evaluation progresses, the weights in the memory graph change but the count balance property is maintained.

Furthermore, the old context and the new context are related according to the fourth clause of the theorem. For the moment being, let the partial order be the trivial reflexive order. Then, the equation stands as:
\begin{alignat*}{3}
&  \ncoverline{ \Gamma' }  &{} \quad + \quad &{} \mathbf{u}' &{} \quad + \quad &{}   \mathbf{0}   \mathop{\diamond}   \ncoverline{ \Gamma_{{\mathrm{4}}} }   \times \langle H'\rangle\\
= \;\; &   \ncoverline{ \Gamma }   \mathop{\diamond}   \mathbf{0}   &{} \quad + \quad &{} \mathbf{u}' \times \langle H' \rangle &{} \quad + \quad &{}   \mathbf{0}   \mathop{\diamond}   \ncoverline{ \Gamma_{{\mathrm{4}}} }   
\end{alignat*}
We can understand this equation through the following analogy. The contexts
can be seen engaged in a transaction with the heap. The heap pays the context
$  \mathbf{0}   \mathop{\diamond}   \ncoverline{ \Gamma_{{\mathrm{4}}} }  $ and gets
$  \mathbf{0}   \mathop{\diamond}   \ncoverline{ \Gamma_{{\mathrm{4}}} }   \times \langle H'\rangle$ resources in return. The
context pays the heap $\mathbf{u}'$ and gets $\mathbf{u}' \times \langle H' \rangle$
resources in return. The equation is the ``balance sheet" of this transaction.

For an arbitrary partial order, the transaction gets skewed in favour of the heap; meaning, the context gets less from the heap for what it pays. This is so because the heap contains more resources than is necessary; so it may ``throw away'' the extra resources.

This soundness theorem subsumes ordinary type soundness. In fact, we can derive the ordinary preservation and progress lemmas from this soundness theorem using bisimilarity of the two reduction relations and the multi-substitution property.
\begin{corollary}
\label{spreservation}
If $ \varnothing \ ;\  \varnothing  \vdash \ottnt{a} : \ottnt{A} $ and $ \ottnt{a}  \leadsto  \ottnt{b} $, then $ \varnothing \ ;\  \varnothing  \vdash \ottnt{b} : \ottnt{A} $.
\end{corollary}  
\begin{proof}
Since $ \ottnt{a}  \leadsto  \ottnt{b} $, for any $S$, we have, $ [  \varnothing  ]\,  \ottnt{a}  \Rightarrow_{ \ottnt{S} } [  \ottnt{H'} \, ;\,  \mathbf{u}' \, ;\,  \Gamma'  ]\,  \ottnt{b'} $ such that $ \ottnt{b'}  \{  \ottnt{H'}  \}   \ottsym{=}  \ottnt{b}$, by lemma (\ref{OrdH}).  Since $ \varnothing \ ;\  \varnothing  \vdash \ottnt{a} : \ottnt{A} $ and $ \varnothing  \vdash  \varnothing ;  \varnothing $ and $a$ is not a value, we have $H, \Gamma, \Gamma_{{\mathrm{4}}}, Q, a'$ such that $ \ottnt{H}  \vdash   \lfloor \Gamma_{{\mathrm{4}}} \rfloor  ;  \Gamma $ and $ [  \varnothing  ]\,  \ottnt{a}  \Rightarrow_{ \ottnt{S} } [  \ottnt{H} \, ;\,  \mathbf{u} \, ;\,  \Gamma_{{\mathrm{4}}}  ]\,  \ottnt{a'} $ and $  \lfloor \Gamma_{{\mathrm{4}}} \rfloor  \ ;\  \Gamma  \vdash \ottnt{a'} : \ottnt{A} $, by the soundness theorem.\\

Now, since $ [  \varnothing  ]\,  \ottnt{a}  \Rightarrow_{ \ottnt{S} } [  \ottnt{H'} \, ;\,  \mathbf{u}' \, ;\,  \Gamma'  ]\,  \ottnt{b'} $ and $ [  \varnothing  ]\,  \ottnt{a}  \Rightarrow_{ \ottnt{S} } [  \ottnt{H} \, ;\,  \mathbf{u} \, ;\,  \Gamma_{{\mathrm{0}}}  ]\,  \ottnt{a'} $, determinism gives us $  \ottnt{b'}  \{  \ottnt{H'}  \}   \ottsym{=}   \ottnt{a'}  \{  \ottnt{H}  \} $. Since $ \ottnt{H}  \vdash   \lfloor \Gamma_{{\mathrm{4}}} \rfloor  ;  \Gamma $ and $  \lfloor \Gamma_{{\mathrm{4}}} \rfloor  \ ;\  \Gamma  \vdash \ottnt{a'} : \ottnt{A} $, by multi-substitution, we have, $  \varnothing \ ;\  \varnothing  \vdash  \ottnt{a'}  \{  \ottnt{H}  \}  : \ottnt{A} $. But $  \ottnt{a'}  \{  \ottnt{H}  \}   \ottsym{=}   \ottnt{b'}  \{  \ottnt{H'}  \} $ and $ \ottnt{b'}  \{  \ottnt{H'}  \}   \ottsym{=}  \ottnt{b}$. Therefore, $ \varnothing \ ;\  \varnothing  \vdash \ottnt{b} : \ottnt{A} $. 
\end{proof}

\begin{corollary}
\label{sprogress}
If $ \varnothing \ ;\  \varnothing  \vdash \ottnt{a} : \ottnt{A} $, then $a$ is a value or there exists b, such that $ \ottnt{a}  \leadsto  \ottnt{b} $.
\end{corollary}  
\begin{proof}
Since $ \varnothing \ ;\  \varnothing  \vdash \ottnt{a} : \ottnt{A} $ and $ \varnothing  \vdash  \varnothing ;  \varnothing $, either $a$ is a value or there exists $H, \Gamma_{{\mathrm{4}}}, Q, a'$ such that $ [  \varnothing  ]\,  \ottnt{a}  \Rightarrow_{ \ottnt{S} } [  \ottnt{H} \, ;\,  \mathbf{u} \, ;\,  \Gamma_{{\mathrm{4}}}  ]\,  \ottnt{a'} $, in which case, by lemma  (\ref{HOrd}), $ \ottnt{a}  \leadsto   \ottnt{a'}  \{  \ottnt{H}  \}  $.
\end{proof}

Next we apply the soundness theorem to derive some useful properties about usage.

\section{Applications}
\label{sec:applications}

\subsection{Irrelevance}

Till now, we have developed our theory over an arbitrary partially-ordered semiring. But an arbitrary semiring is too general a structure for deriving theorems we are interested in. For example, the set $\{0,1\}$ with $1 + 1 = 0$ and all other operations defined in the usual way is also a semiring. But such a semiring does not capture our notion of usage since $0$ is supposed to mean no usage and $1$ (whenever $1 \neq 0$) is supposed to mean some usage. For $0$ to mean no usage in a semiring $Q$, the equation ${ \color{black}{q} }  \ottsym{+}   { \color{black}{1} }  = 0$ must have no solution. We call an element $q' \in Q$ \textit{positive} (respectively \textit{positive-or-more}) iff $q' = q + 1$ (respectively $q + 1 \leq q'$) for some $q \in Q$. The above condition then means that $0$ is not positive. If we also have a partial order, the constraint ${ \color{black}{q} }  \ottsym{+}   { \color{black}{1} }  \leq 0$ must be unsatisfiable; meaning $0$ should not be positive-or-more. We call this the \textit{zero-unusable} criterion. Henceforth, we restrict our attention to semirings that meet this criterion. The following lemmas formalize the idea discussed here.

\begin{lemma}
\label{lemma:zero-dead}
In a zero-unusable semiring, if $ [  \ottnt{H}  ]\,  \ottnt{a}  \Rightarrow_{ \ottnt{S} } [  \ottnt{H'} \, ;\,  \mathbf{u}' \, ;\,  \Gamma_{{\mathrm{4}}}  ]\,  \ottnt{a'} $ and $ \ottmv{x_{\ottmv{i}}}  \overset{  { \color{black}{0} }  }{\mapsto} { \Gamma_{\ottmv{i}} \vdash  \ottnt{a_{\ottmv{i}}}  :  \ottnt{A_{\ottmv{i}}} }  \in H$, then the component $\mathbf{u}' (\ottmv{x_{\ottmv{i}}}) = 0$ and $ \ottmv{x_{\ottmv{i}}}  \overset{  { \color{black}{0} }  }{\mapsto} { \Gamma_{\ottmv{i}} \vdash  \ottnt{a_{\ottmv{i}}}  :  \ottnt{A_{\ottmv{i}}} }  \in H'$. 
\end{lemma}

We see above that locations with count $0$ cannot be referenced during computation. Also, the count for such locations always remains $0$. Now, if they cannot be referenced, what they contain should not matter. In other words, $0$-graded variables do not affect the result of computation. Two initial configurations that differ only in the assignments of some $0$-graded variables produce identical results. This means that such assignments do not interfere with evaluation and are irrelevant.

\scw{I would feel better if we called this lemma something else. It is weaker than the usual non-interference lemma from security type systems}
\begin{lemma}[Zero non-interference]
\label{lemma:zero-nonint}
Let $\ottnt{H_{\ottmv{i}\,{\mathrm{1}}}} =  \ottmv{x_{\ottmv{i}}}  \overset{  { \color{black}{0} }  }{\mapsto} { \Gamma_{{\mathrm{1}}} \vdash  \ottnt{a_{{\mathrm{1}}}}  :  \ottnt{A_{{\mathrm{1}}}} } $ and $\ottnt{H_{\ottmv{i}\,{\mathrm{2}}}} =  \ottmv{x_{\ottmv{i}}}  \overset{  { \color{black}{0} }  }{\mapsto} { \Gamma_{{\mathrm{2}}} \vdash  \ottnt{a_{{\mathrm{2}}}}  :  \ottnt{A_{{\mathrm{2}}}} } $. Then, in a zero-unusable semiring, if $ [    \ottnt{H_{{\mathrm{1}}}}  ,  \ottnt{H_{\ottmv{i}\,{\mathrm{1}}}}   ,  \ottnt{H_{{\mathrm{2}}}}   ]\,  \ottnt{b}  \Rightarrow_{   \ottnt{S}  \, \cup   \,\text{fv}\,  \ottnt{a_{{\mathrm{2}}}}    } [    \ottnt{H'_{{\mathrm{1}}}}  ,  \ottnt{H_{\ottmv{i}\,{\mathrm{1}}}}   ,  \ottnt{H'_{{\mathrm{2}}}}  \, ;\,  \mathbf{u}' \, ;\,  \Gamma_{{\mathrm{4}}}  ]\,  \ottnt{b'} $, \\ then $ [    \ottnt{H_{{\mathrm{1}}}}  ,  \ottnt{H_{\ottmv{i}\,{\mathrm{2}}}}   ,  \ottnt{H_{{\mathrm{2}}}}   ]\,  \ottnt{b}  \Rightarrow_{   \ottnt{S}  \, \cup   \,\text{fv}\,  \ottnt{a_{{\mathrm{1}}}}    } [    \ottnt{H'_{{\mathrm{1}}}}  ,  \ottnt{H_{\ottmv{i}\,{\mathrm{2}}}}   ,  \ottnt{H'_{{\mathrm{2}}}}  \, ;\,  \mathbf{u}' \, ;\,  \Gamma_{{\mathrm{4}}}  ]\,  \ottnt{b'} $.
\end{lemma}

Note that not just $0$-graded resources may be unusable, any $s$-graded resource for which the constraint $q + 1 \leq s$ is unsatisfiable is unusable. With respect to the security semirings described in Section \ref{sec:semiring-examples}, this means that data from any security level $s$ for which $1 \nleq s$ is unusable. This makes sense since the default view of the type system is $1$ or $ {\color{black}{\mathsf{Public} } } $ so that data judged to be more secure (or incomparable) cannot be used at this level. This gives us the following lemma for the class of security lattices described in Section \ref{sec:semiring-examples}:

\begin{lemma}[$s$ non-interference]
\label{lemma:s-nonint}
Let $1 \nleq s$ in a security lattice. Let $\ottnt{H_{\ottmv{i}\,{\mathrm{1}}}} =  \ottmv{x_{\ottmv{i}}}  \overset{ { \color{black}{s} } }{\mapsto} { \Gamma_{{\mathrm{1}}} \vdash  \ottnt{a_{{\mathrm{1}}}}  :  \ottnt{A_{{\mathrm{1}}}} } $ and $\ottnt{H_{\ottmv{i}\,{\mathrm{2}}}} =  \ottmv{x_{\ottmv{i}}}  \overset{ { \color{black}{s} } }{\mapsto} { \Gamma_{{\mathrm{2}}} \vdash  \ottnt{a_{{\mathrm{2}}}}  :  \ottnt{A_{{\mathrm{2}}}} } $. If $ [    \ottnt{H_{{\mathrm{1}}}}  ,  \ottnt{H_{\ottmv{i}\,{\mathrm{1}}}}   ,  \ottnt{H_{{\mathrm{2}}}}   ]\,  \ottnt{b}  \Rightarrow_{   \ottnt{S}  \, \cup   \,\text{fv}\,  \ottnt{a_{{\mathrm{2}}}}    } [    \ottnt{H'_{{\mathrm{1}}}}  ,  \ottnt{H_{\ottmv{i}\,{\mathrm{1}}}}   ,  \ottnt{H'_{{\mathrm{2}}}}  \, ;\,  \mathbf{u}' \, ;\,  \Gamma_{{\mathrm{4}}}  ]\,  \ottnt{b'} $, \\ then $ [    \ottnt{H_{{\mathrm{1}}}}  ,  \ottnt{H_{\ottmv{i}\,{\mathrm{2}}}}   ,  \ottnt{H_{{\mathrm{2}}}}   ]\,  \ottnt{b}  \Rightarrow_{   \ottnt{S}  \, \cup   \,\text{fv}\,  \ottnt{a_{{\mathrm{1}}}}    } [    \ottnt{H'_{{\mathrm{1}}}}  ,  \ottnt{H_{\ottmv{i}\,{\mathrm{2}}}}   ,  \ottnt{H'_{{\mathrm{2}}}}  \, ;\,  \mathbf{u}' \, ;\,  \Gamma_{{\mathrm{4}}}  ]\,  \ottnt{b'} $.
\end{lemma}

\subsection{Garbage Collection}

Now let us look at locations with count $0$ in memory graphs. The sum of the lengths of all paths from the source node to such a node must be $0$. The zero-unusable criterion, along with the count balance property, implies that none of these paths has positive-or-more length. This means that all the edge-weights in any such path cannot be positive-or-more. 

The condition that $0$ is not positive-or-more is a weaker version of a well-known constraint put on semirings. If $0$ is a minimal element, a stronger constraint is zerosumfree\footnote{Our terminology follows \citet{golan}.}. A semiring $Q$ is said to be zerosumfree if for any $q_1 , q_2 \in Q$, the equation $q_1 + q_2 = 0$ implies $q_1 = q_2 = 0$. If we work with a zerosumfree semiring with $0$ as a minimal element, we know that the length of any path from the source node to a node with count $0$ is $0$. But all the edge-weights along such a path may be non-zero. This is so because the product of two non-zero elements may be $0$. If we disallow this, then there is no path from the source node to such a node ($0$ weight edges are omitted). Semirings which satisfy $q_1 \cdot q_2 = 0 \implies q_1 = 0 \text{ or }  q_2 = 0$ are called \emph{entire}\footnote{The zerosumfree property is sometimes called ``positive'' and the entire property is sometimes called ``zero-product'' property.}. With these constraints on the semiring, we have the following lemma:
 
\begin{lemma}
\label{lemma:gc}
In a zerosumfree, entire semiring with $0$ as a minimal element, if $ \ottnt{H}  \vdash  \Delta ;  \Gamma $ and $ \ottmv{x_{\ottmv{i}}}  \overset{  { \color{black}{0} }  }{\mapsto} { \Gamma_{\ottmv{i}} \vdash  \ottnt{a_{\ottmv{i}}}  :  \ottnt{A_{\ottmv{i}}} }  \in H$, then $v_i$ (the node corresponding to $x_i$) belongs to an isolated subgraph (of $G_{H,\Gamma}$) that does not contain the source node.
\end{lemma}

The lemma above says that all the $0$-count assignments lie in isolated islands disconnected from the line of computation. So at any point, it is safe to garbage collect all such assignments. 

\scw{Is there an interpretation of this lemma with respect to security? Why is
  this an important application? What does it mean that this lemma does not
  hold for the five-point security lattice example from the introduction? } \pc{Possible. But I do not see now.}

\subsection{Linearity}

Let us now look at linearity. Just having a $1$ in the semiring is not enough to capture a notion of linearity. For example, $1$ does not really represent linear usage in the boolean semiring since $1 + 1 = 1$. If $1$ must mean linear usage, then it cannot be equal to or greater than the successor of any quantity other than $0$, where \textit{successor} of $q$ is defined as $q + 1$. Formally, the pair of constraints: $q + 1 \leq 1$ and $q \neq 0$ must have no solution. We call this the \textit{one-linear} criterion. In semirings that meet the zero-unusable and one-linear criteria, $1$ represents single usage.

Mirroring our discussion on $0$-usage, we strengthen the one-linear criterion to derive a useful property about nodes with a count of $1$ in memory graphs. Let us call semirings obeying the following constraints linear:
\begin{itemize}
\item $q_1 + q_2 = 1 \implies q_1 = 1 \text{ and } q_2 = 0 \text{ or } q_1 = 0 \text{ and } q_2 = 1$
\item $q_1 \cdot q_2 = 1 \implies q_1 = q_2 = 1$
\end{itemize}      

For entire, zerosumfree, linear semirings with $0$ and $1$ as minimal elements, we have the following property:

\begin{lemma}[Quantitative single-pointer property]
If $ \ottnt{H}  \vdash  \Delta ;  \Gamma $ and $ \ottmv{x_{\ottmv{i}}}  \overset{  { \color{black}{1} }  }{\mapsto} { \Gamma_{\ottmv{i}} \vdash  \ottnt{a_{\ottmv{i}}}  :  \ottnt{A_{\ottmv{i}}} }  \in H$, then in $G_{H,\Gamma}$, there is a single path $p$ from the source node to $v_i$ and all the weights on $p$ are $1$. Further, for any node $v_j$ on $p$, the subpath is the only path from the source node to $v_j$. 
\end{lemma}

Along with the soundness theorem, this gives us a quantitative version of the single pointer property. In words, it means that there is one and only one way to reference a linear resource; any resource along the way has a single pointer to it. This property would enable one to carry out safe in-place update for linear resources.

Now that we have explored a graded simple type system, we move on to dependent types.

\section{Graded Dependent Types}
\label{sec:dependent}

In this section we define \Langname, a language with graded dependent types.
The syntax is presented below. \Langname uses a single
syntactic category for terms and types.
\[
\begin{array}{llcl}
\mathit{terms, types} & a, b, A, B   & ::=&  \textbf{Type}  \alt \ottmv{x}  \\
                      &              &\alt& \ottkw{Unit} \alt \ottkw{unit} \alt  \ottkw{let}\,  \ottkw{unit} \,=\, \ottnt{a} \ \ottkw{in}\  \ottnt{b}   \\
                      &              &\alt&  \Pi  \ottmv{x} \!:^ { \color{black}{q} } \! \ottnt{A} . \ottnt{B}  \alt  \lambda \ottmv{x} \!:^ { \color{black}{q} } \! \ottnt{A} . \ottnt{a}  \alt \ottnt{a} \, \ottnt{b} \\
%                      &              &\alt&  \Box^{ { \color{black}{q} } }  \ottnt{A}  \alt  \ottkw{box} _ { \color{black}{q} } \, \ottnt{a}  \alt  \ottkw{let}\, \ottkw{box} \, \ottmv{x} \,=\, \ottnt{a} \ \ottkw{in}\  \ottnt{b}  \\
                      &              &\alt&  \Sigma  \ottmv{x} \!\!:^ { \color{black}{q} } \!\! \ottnt{A} . \ottnt{B}  \alt \ottsym{(}  \ottnt{a}  \ottsym{,}  \ottnt{b}  \ottsym{)} \alt \ottkw{let} \, \ottsym{(}  \ottmv{x}  \ottsym{,}  \ottmv{y}  \ottsym{)}  \ottsym{=}  \ottnt{a} \, \mathsf{in} \, \ottnt{b} \\
                      &              &\alt&  \ottnt{A}  \oplus  \ottnt{B}  \alt  \ottkw{inj}_1\,  \ottnt{a}  \alt  \ottkw{inj}_2\,  \ottnt{a} 
                                            \alt  \ottkw{case}_ { \color{black}{q} } \,  \ottnt{a} \, \ottkw{of}\,  \ottnt{b_{{\mathrm{1}}}}  ;  \ottnt{b_{{\mathrm{2}}}}  \\

\end{array}
\]

%% SCW: don't include def or weak-def in this figure. Also, use simplified version of 
%% conversion rule, defined in (extra.ott) instead of substituting one.
\begin{figure}
\drules[T]{$ \Delta  ;  \Gamma  \vdash \ottnt{a} : \ottnt{A} $}
{Typing rules for \Langname, a graded dependent type system}
{sub,weak,convert,type,var,Unit,unit,UnitElim,pi,lam,app,
Sigma,Tensor,SigmaElim,sum,injOne,injTwo,SumElim}
\caption{Typing rules for dependent, quantitative type system}
\label{fig:dep-type}
\end{figure}

\begin{figure}
\end{figure}

\subsection{Type System}

The rules of this type system, shown in \pref{fig:dep-type}, are inspired by
the Pure Type Systems of
Barendregt~\cite{barendregt:lambda-calculi-with-types}. However, for
simplicity, this system includes only a single sort, $ \textbf{Type} $ and a single
axiom $ \textbf{Type} : \textbf{Type} $.\footnote{This definition corresponds to
  $\lambda\ast$, which is `inconsistent' in the sense that all types are
  inhabited. However, this inconsistency does not interfere with the syntactic
  properties of the system that we are interested in as a core for Dependent
  Haskell.} We annotate Barendregt's system with quantities, as well as add the
unit type, sums and sigma types. Note that the \rref{T-Convert} uses the definitional 
equivalence relation, which is essentially $\beta$-equivalence. This relation is axiomatically specified in Section \ref{defeq}.

The key idea of this design is that quantities only count the
\emph{runtime} usage of variables. In a judgement $ \Delta  ;  \Gamma  \vdash \ottnt{a} : \ottnt{A} $, the
quantities recorded in $\Gamma$ should be derived only from the parts of $\ottnt{a}$
that are needed during computation. All other uses of these variables, whether
in the type $\ottnt{A}$, in irrelevant parts of $\ottnt{a}$, or in types that appear
later in the context, should not contribute to this count. This distinction is significant because in a dependently-typed system terms may appear in types. As a result, the typing rules must ensure that both terms and types are well-formed during type checking. Therefore, the type system
must include premises of the form $ \Delta \ ;\  \Gamma  \vdash \ottnt{A} :  \textbf{Type}  $, that hold when
$\ottnt{A}$ is a well-formed type. But we don't want to add this usage to the usage of the term.

What this means for the type system is that any usage of a context to check an
irrelevant component, like the type, should be multiplied by $0$, just like the irrelevant
argument in example~\ref{example:irrelevant}. For example, in the rule for
variables \rref{T-Var}, any uses of the context $\Gamma$ to check the type $\ottnt{A}$ are
discarded (multiplied by 0) in the resulting derivation. Similarly, in the
rule for weakening \rref{T-Weak}, we check that the type of the weakened variable is
well-formed using some context $\Gamma_{{\mathrm{2}}}$ that is compatible with $\Gamma_{{\mathrm{1}}}$
(same $\Delta$). Since $\Gamma_{{\mathrm{1}}}  \ottsym{+}     { \color{black}{0} }    \cdot   \Gamma_{{\mathrm{2}}}   = \Gamma_{{\mathrm{1}}}$, usage $\Gamma_{{\mathrm{2}}}$ doesn't appear in the conclusion of the rule. Many rules follow this pattern of checking types with
some usage-unconstrained context, including $\Gamma_{{\mathrm{2}}}$ in \rref{T-convert} and
\rref{T-lam}, and $\Gamma_{{\mathrm{3}}}$ in \rref{T-UnitElim}. % \rref{T-LetBox}, and \rref{CaseElim}.
This \rref{T-UnitElim} implements a form of dependent pattern
matching. Here, the type of the branch can observe 
that the eliminated term $\ottnt{a}$ is equal to the pattern $\ottkw{unit}$. 
To support this refinement, the result type $\ottnt{B}$ must type check with a
free variable $\ottmv{y}$ of $\ottkw{Unit}$ type. The other elimination rules, \rref{T-SigmaElim} and \rref{T-SumElim}, also follow this style of dependent pattern matching.
% For example, in \rref{T-LetBox}, the type of the branch can observe
% that the eliminated term $\ottnt{a}$ is equal to the pattern $ \ottkw{box} _ { \color{black}{q} } \, \ottmv{x} $). 
% Similarly, in case analysis for sums (\rref{T-CaseElim}), each branch
% can observe that the eliminated term is equal to $ \ottkw{inj}_1\,  \ottmv{x} $ or $ \ottkw{inj}_2\,  \ottmv{x} $.

\paragraph{Irrelevant Quantification}

Now consider the \rref{T-pi}. 
\ifextended
\[ \ottdruleTXXpi{} \]
\fi
In particular, note that the usage annotation on the type itself ($q$) is
different from $r$, which records how many times $x$ is used in $B$. The
annotation $q$ tracks the usage of the argument in the body of a function with
this type and this usage may have no relation to the usage of $x$ in
the body of the type itself.
This difference between $q$ and $r$ allows \Langname{} to represent parametric polymorphism by marking
type arguments with usage $0$. For example, the analogue of the System F type
$\forall\alpha.\alpha\rightarrow\alpha$, can be expressed in this system as
$ \Pi  \ottmv{x} \!:^  { \color{black}{0} }  \!  \textbf{Type}  .  {}^{  { \color{black}{1} }  } \ottmv{x} \rightarrow  \ottmv{x}  $.
This type is well-formed because, even though the annotation on the variable
$x$ is 0, that rule allows $x$ to be used any number of times in the body of
the type.

Some versions of irrelevant quantifiers in type theories constrain $r$ to be
equal to $q$ ~\cite{Abel12}.  By coupling the usage of variables in the body of the
lambda with the result type of the $\Pi$, these systems rule out the
representation of polymorphic types, such as the one shown
above. Here, we can model this more restricted form of quantifier with the
assistance of the box modality. If, instead of using the type
$ \Pi  \ottmv{x} \!:^  { \color{black}{0} }  \! \ottnt{A} . \ottnt{B} $, we use the type $ \Pi  \ottmv{x} \!:^  { \color{black}{1} }  \!  \Box^{  { \color{black}{0} }  }  \ottnt{A}  . \ottnt{B} $, we can force the
result type to also make no (relevant) use of the argument within $\ottnt{B}$. The
box $x$ can be unboxed as many times as desired, but each unboxing must be
used exactly $0$ times.

It is this distinction between the types $ \Pi  \ottmv{x} \!:^  { \color{black}{0} }  \! \ottnt{A} . \ottnt{B} $ and
$ \Pi  \ottmv{x} \!:^  { \color{black}{1} }  \! \ottsym{(}   \Box^{  { \color{black}{0} }  }  \ottnt{A}   \ottsym{)} . \ottnt{B} $ (and a similar distinction between $ \Sigma  \ottmv{x} \!\!:^  { \color{black}{0} }  \!\! \ottnt{A} . \ottnt{B} $ and $ \Sigma  \ottmv{x} \!\!:^  { \color{black}{1} }  \!\! \ottsym{(}   \Box^{  { \color{black}{0} }  }  \ottnt{A}   \ottsym{)} . \ottnt{B} $) that motivates our inclusion of the usage annotation
on the $\Pi$ and $\Sigma$ types directly. In the simple type
system, we can derive usage-annotated functions from linear functions and the
box modality: there is no need to annotate the arrow with any quantity other
than 1. But here, due to dependency, we cannot have parametrically polymorphic types
without this additional form. 
On the other hand, with the presence of usage-annotated $\Sigma$-types, we do
not need to include the box modality. Instead, we can encode the type
$ \Box^{ { \color{black}{q} } }  \ottnt{A} $ using the non-dependent tensor $ \Sigma  \ottmv{x} \!\!:^ { \color{black}{q} } \!\! \ottnt{A} . \ottkw{Unit} $. Thus, we
can eliminate this special form from the language.

\subsection{Metatheory}
\label{sec:dep-metatheory}

We have proven, in Coq, the following properties about the dependently-typed
system.

First, well-formed terms have well-formed types. However, the resources used by such types are unrelated to those of the terms.

\begin{lemma}[Regularity]
If $ \Delta \ ;\  \Gamma  \vdash \ottnt{a} : \ottnt{A} $ then there exists some $\Gamma'$ such that $ \Delta \ ;\  \Gamma'  \vdash \ottnt{A} :  \textbf{Type}  $.
\end{lemma}

Next, we generalize the substitution lemma for the simple version to this system, by propagating it
through the context and type.

\begin{lemma}[Substitution] If 
$ \Delta_{{\mathrm{1}}} \ ;\  \Gamma  \vdash \ottnt{a} : \ottnt{A} $ and 
$   \Delta_{{\mathrm{1}}} ,   \ottmv{x} \!\!:\!\! \ottnt{A}   ,  \Delta_{{\mathrm{2}}}   ;    \Gamma_{{\mathrm{1}}} ,   \ottmv{x} \! :^{ { \color{black}{q} } }\! \ottnt{A}   ,  \Gamma_{{\mathrm{2}}}   \vdash \ottnt{b} : \ottnt{B} $ then 
$  \Delta_{{\mathrm{1}}} ,  \Delta_{{\mathrm{2}}}   \ottsym{\{}  \ottnt{a}  \ottsym{/}  \ottmv{x}  \ottsym{\}}  ;   \ottsym{(}  \Gamma_{{\mathrm{1}}}  \ottsym{+}   { \color{black}{q} }   \cdot   \Gamma   \ottsym{)} ,  \Gamma_{{\mathrm{2}}}   \ottsym{\{}  \ottnt{a}  \ottsym{/}  \ottmv{x}  \ottsym{\}}  \vdash \ottnt{b}  \ottsym{\{}  \ottnt{a}  \ottsym{/}  \ottmv{x}  \ottsym{\}} : \ottnt{B}  \ottsym{\{}  \ottnt{a}  \ottsym{/}  \ottmv{x}  \ottsym{\}} $.
\end{lemma}

Furthermore, even though we have an explicit weakening rule in this system, we
also show that we can weaken with a zero-annotated fresh variable anywhere in
the judgement.

\begin{lemma}[Weakening]
\label{lem:weakening}
If $  \Delta_{{\mathrm{1}}} ,  \Delta_{{\mathrm{2}}}   ;   \Gamma_{{\mathrm{1}}} ,  \Gamma_{{\mathrm{2}}}   \vdash \ottnt{a} : \ottnt{A} $  and $ \Delta_{{\mathrm{1}}}  ;  \Gamma_{{\mathrm{3}}}  \vdash \ottnt{B} :  \textbf{Type}  $ 
then $   \Delta_{{\mathrm{1}}} ,   \ottmv{x} \!\!:\!\! \ottnt{B}   ,  \Delta_{{\mathrm{2}}}   ;    \Gamma_{{\mathrm{1}}} ,   \ottmv{x} \! :^{  { \color{black}{0} }  }\! \ottnt{B}   ,  \Gamma_{{\mathrm{2}}}   \vdash \ottnt{a} : \ottnt{A} $.
\end{lemma}

With a small-step relation that is identical to that of the simply
typed version, we have the following type soundness theorem. % shown in Figure~\ref{fig:step}. 

\begin{theorem}[Preservation]
If $ \Delta  ;  \Gamma  \vdash \ottnt{a} : \ottnt{A} $ and $ \ottnt{a}  \leadsto  \ottnt{a'} $ then $ \Delta  ;  \Gamma  \vdash \ottnt{a'} : \ottnt{A} $.
\end{theorem}

\begin{theorem}[Progress]
If $ \varnothing  ;  \varnothing  \vdash \ottnt{a} : \ottnt{A} $ then either $\ottnt{a}$ is a value or there exists some $\ottnt{a'}$ such
that $ \ottnt{a}  \leadsto  \ottnt{a'} $.
\end{theorem}

Now, akin to the simple version, we develop a heap semantics for the dependent version.

\section{Heap semantics for \Langname}
\label{sec:heap-dependent}

The presence of dependent types causes one issue with the heap semantics: because substitutions are delayed through the heap, the terms and their types can ``get out of sync''. 

For example, if we have the application of a polymorphic identity function $ \lambda \ottmv{x} \!:^  { \color{black}{0} }  \!  \textbf{Type}  .  \lambda \ottmv{y} \!:^  { \color{black}{1} }  \! \ottmv{x} . \ottmv{y}  $ to some type argument $\ottkw{Unit}$, then the result $\ottsym{(}   \lambda \ottmv{x} \!:^  { \color{black}{0} }  \!  \textbf{Type}  .  \lambda \ottmv{y} \!:^  { \color{black}{1} }  \! \ottmv{x} . \ottmv{y}    \ottsym{)} \, \ottkw{Unit}$ should have type 
$ \Pi  \ottmv{y} \!:^  { \color{black}{1} }  \! \ottkw{Unit} . \ottkw{Unit} $. By the \rref{Small-AppBeta}, $ [  \varnothing  ]\,  \ottsym{(}   \lambda \ottmv{x} \!:^  { \color{black}{0} }  \!  \textbf{Type}  .  \lambda \ottmv{y} \!:^  { \color{black}{1} }  \! \ottmv{x} . \ottmv{y}    \ottsym{)} \, \ottkw{Unit}  \Rightarrow [   \ottmv{x}  \overset{  { \color{black}{0} }  }{\mapsto}  \ottkw{Unit}   ] \,   \lambda \ottmv{y} \!:^  { \color{black}{1} }  \! \ottmv{x} . \ottmv{y}  $.  The term $ \lambda \ottmv{y} \!:^  { \color{black}{1} }  \! \ottmv{x} . \ottmv{y} $ has type $ \Pi  \ottmv{y} \!:^  { \color{black}{1} }  \! \ottmv{x} . \ottmv{x} $. But since $\ottmv{x}  \ottsym{=}  \ottkw{Unit}$, we see that $ \Pi  \ottmv{y} \!:^  { \color{black}{1} }  \! \ottmv{x} . \ottmv{x}   \ottsym{=}   \Pi  \ottmv{y} \!:^  { \color{black}{1} }  \! \ottkw{Unit} . \ottkw{Unit} $; as such, $ \lambda \ottmv{y} \!:^  { \color{black}{1} }  \! \ottmv{x} . \ottmv{y} $ can also be assigned the type $ \Pi  \ottmv{y} \!:^  { \color{black}{1} }  \! \ottkw{Unit} . \ottkw{Unit} $. So to align the types of the redex and the reduct, we need to know about the new assignments loaded into the heap. This issue did not exist in the simple setting since the types did not depend on term variables. Note that this is not a usage-related issue, any heap-based reduction that delays substitution will need to address it while proving soundness. The good news is that it can be resolved with a simple extension to the type system.

\subsection{A Dependently-Typed Language with Definitions}

We extend our contexts with definitions that mimic delayed substitutions. These definitions are used \emph{only} in deriving type equalities. From the type system perspective, they are essentially a bookkeeping device added to enable reasoning with respect to the heap semantics.
\[
\begin{array}{llcl}
\textit{usage contexts} & \Gamma & ::= &  \varnothing  \alt  \Gamma ,   \ottmv{x} \! :^{ { \color{black}{q} } }\! \ottnt{A}   \alt  \Gamma ,   \ottmv{x}  \! = \!  \ottnt{a}  \! :^{ { \color{black}{q} } } \!  \ottnt{A}  \\
\textit{contexts} & \Delta & ::= &  \varnothing  \alt  \Delta ,   \ottmv{x} \!\!:\!\! \ottnt{A}   \alt  \Delta ,   \ottmv{x} \! = \!  \ottnt{a}  \! : \!  \ottnt{A}  \\
\end{array}
\]

Along with this extension to the context, we modify the conversion rule and add two new
typing rules to the system, as shown below. (In \rref{T-conv}, $ \ottnt{A}  \{  \Delta  \} $ denotes the type
obtained by substituting in $A$, in reverse order, the definiens in place of
the variables for the definitions in $\Delta$.)

\drules[T]{$ \Delta  ;  \Gamma  \vdash \ottnt{a} : \ottnt{A} $}
{Typing rules for dependent system with definitions}
{conv,def,weak-def}

The definitions act like usual variable assumptions: \rref{T-def} and
\rref{T-weak-def} mirror \rref{T-var} and \rref{T-weak} respectively.
They are applied only during the conversion \rref{T-conv}
that substitutes out these definitions before comparing for
$\beta$-equivalence. This modified rule means that the term $ \lambda \ottmv{y} \!:^  { \color{black}{1} }  \! \ottmv{x} . \ottmv{y} $ can
be given the type $ \Pi  \ottmv{y} \!:^  { \color{black}{1} }  \! \ottkw{Unit} . \ottkw{Unit} $ in a context that defines $\ottmv{x}$ to
be $\ottkw{Unit}$.  

The extended type system too has the syntactic soundness
properties mentioned in Section~\ref{sec:dep-metatheory}.  Furthermore, because
definitions act only on types, definitions do not add extra resource demands
to the typing derivation. As a result, we can always convert a normal variable assumption to include
some definition as long as the definiens type checks. Furthermore, the
resources used by the definiens ($\Gamma$ below) are unimportant.
\begin{lemma}[InsertEq]\label{InsertEq}
If $    \Delta_{{\mathrm{1}}} ,   \ottmv{x} \!\!:\!\! \ottnt{A}    ,  \Delta_{{\mathrm{2}}}  \ ;\     \Gamma_{{\mathrm{1}}} ,   \ottmv{x} \! :^{ { \color{black}{q} } }\! \ottnt{A}    ,  \Gamma_{{\mathrm{2}}}   \vdash \ottnt{b} : \ottnt{B} $ and $ \Delta_{{\mathrm{1}}} \ ;\  \Gamma  \vdash \ottnt{a} : \ottnt{A} $, then $    \Delta_{{\mathrm{1}}} ,   \ottmv{x} \! = \!  \ottnt{a}  \! : \!  \ottnt{A}    ,  \Delta_{{\mathrm{2}}}  \ ;\     \Gamma_{{\mathrm{1}}} ,   \ottmv{x}  \! = \!  \ottnt{a}  \! :^{ { \color{black}{q} } } \!  \ottnt{A}    ,  \Gamma_{{\mathrm{2}}}   \vdash \ottnt{b} : \ottnt{B} $.
\end{lemma}
Contexts can also be weakened with new (unused) definitions, analogous to lemma ~\ref{lem:weakening}.
\begin{lemma}[Weakening with Definitions]
\label{weakdefn}
If $  \Delta_{{\mathrm{1}}} ,  \Delta_{{\mathrm{2}}}   ;   \Gamma_{{\mathrm{1}}} ,  \Gamma_{{\mathrm{2}}}   \vdash \ottnt{b} : \ottnt{B} $ and $ \Delta_{{\mathrm{1}}}  ;  \Gamma  \vdash \ottnt{a} : \ottnt{A} $ 
then $   \Delta_{{\mathrm{1}}} ,   \ottmv{x} \! = \!  \ottnt{a}  \! : \!  \ottnt{A}   ,  \Delta_{{\mathrm{2}}}   ;    \Gamma_{{\mathrm{1}}} ,   \ottmv{x}  \! = \!  \ottnt{a}  \! :^{  { \color{black}{0} }  } \!  \ottnt{A}   ,  \Gamma_{{\mathrm{2}}}   \vdash \ottnt{b} : \ottnt{B} $.
\end{lemma}

\scw{I want to claim some sort of multisubstitution lemma just for the type system here. But I have not proved it.}

Because we have modified the contexts to include definitions, we need to 
modify the heap reduction and compatibility relations. Only for the reduction rules that load new
assignments into the heap, we now need the added context of new variables ($\Gamma_{{\mathrm{4}}}$) to
remember their assignments. For example, the \rref{Small-AppBeta} is modified
as below.

\drules[Small]{$ [  \ottnt{H}  ]\,  \ottnt{a}  \Rightarrow_{ \ottnt{S} }^{ { \color{black}{q} } } [  \ottnt{H'} \, ;\,  \mathbf{u}' \, ;\,  \Gamma_{{\mathrm{4}}}  ]\,  \ottnt{a'} $}{SmallStep with definitions}{DAppBeta}

Similarly, \rref{Compat-Cons} needs to track more information in the context.

\drules[Compat]{$ \ottnt{H}  \vdash  \Delta ;  \Gamma $}{Compatibility with definitions}{ConsDef}
   
Note that with this modification, if $ \ottnt{H}  \vdash  \Delta ;  \Gamma $ then $ \ottnt{b}  \{  \ottnt{H}  \}   \ottsym{=}   \ottnt{b}  \{  \Delta  \} $ for any term $b$.

These are all the changes we need. Since the added context of new variables does not
play a major role in the step relation, all the previously stated lemmas regarding this relation
hold. But with dependency, the multi-substitution lemma \ref{multisub} needs to be modified
to also substitute into the type (in addition to the term).

\begin{lemma}[Multi-substitution]
If $ \ottnt{H}  \vdash  \Delta ;  \Gamma $ and $ \Delta \ ;\  \Gamma  \vdash \ottnt{a} : \ottnt{A} $, then $ \varnothing \ ;\  \varnothing  \vdash  \ottnt{a}  \{  \ottnt{H}  \}  :  \ottnt{A}  \{  \ottnt{H}  \}  $.
\end{lemma}

Another point worth noting here is that, if $ [  \ottnt{H}  ]\,  \ottnt{a}  \Rightarrow_{ \ottnt{S} }^{ { \color{black}{q} } } [  \ottnt{H'} \, ;\,  \mathbf{u}' \, ;\,  \Gamma'  ]\,  \ottnt{a'} $ then $ \ottnt{a}  \{  \ottnt{H}  \}   \equiv   \ottnt{a'}  \{  \ottnt{H'}  \} $ by lemma \ref{HOrd} and definition \ref{defeql}. Before moving further, let us reflect how the original typing and heap compatibility judgements relate to their extended counterparts. For the sake of distinction, let us denote the original relations by $\vdash_{o}$. Now, for $ \ottnt{H}  \vdash_{o}  \Delta ;  \Gamma $, let $ \Delta _ \ottnt{H} $ and $ \Gamma _ \ottnt{H} $ be $\Delta$ and $\Gamma$ respectively with their variables defined according to (assignments in) $H$. Also, let $ \ottnt{H} _{ \ottnt{H} } $ denote $H$ with the variables in the embedded contexts in $H$ defined according to $H$. Then, we have:

\begin{lemma}[Elaboration]
If $ \ottnt{H}  \vdash_{o}  \Delta ;  \Gamma $ and $ \Delta \ ;\  \Gamma  \vdash_{o}  \ottnt{a} : \ottnt{A} $, then $  \ottnt{H} _{ \ottnt{H} }   \vdash   \Delta _ \ottnt{H}  ;   \Gamma _ \ottnt{H}  $ and $  \Delta _ \ottnt{H}  \ ;\   \Gamma _ \ottnt{H}   \vdash \ottnt{a} : \ottnt{A} $.
\end{lemma}

By virtue of this elaboration, soundness for the extended system implies soundness for the original one.

\subsection{Proof of the Heap Soundness Theorem}

Now we prove the heap soundness theorem for \Langname{}.  However, we
first state some subordinate lemmas that are required in the proof. The
following lemma allows us to throw away resources from heaps.

\begin{lemma}[Sub-heaping]
\label{subheap}
If $ \ottnt{H}  \vdash  \Delta ;  \Gamma $ and $\Gamma' \leq \Gamma$, then there exists $H'$ such that $ \ottnt{H'}  \vdash  \Delta ;  \Gamma' $ and $H' \leq H$.
\end{lemma}

We can insert new definitions into the heap.
\begin{lemma}[SmallStep weakening]
\label{lem:smallstep-weakening}
If $ [   \ottnt{H_{{\mathrm{1}}}}  ,  \ottnt{H_{{\mathrm{2}}}}   ]\,  \ottnt{a}  \Rightarrow_{   \ottnt{S}  \, \cup   \{  \ottmv{x}  \}    }^{ { \color{black}{r} } } [   \ottnt{H'_{{\mathrm{1}}}}  ,  \ottnt{H'}  \, ;\,   \mathbf{u}_{{\mathrm{1}}}  \mathop{\diamond}  \mathbf{u}  \, ;\,  \Gamma_{{\mathrm{4}}}  ]\,  \ottnt{a'} $ 
  and $ | \ottnt{H'_{{\mathrm{1}}}} |  =  | \mathbf{u}_{{\mathrm{1}}} |  =  | \ottnt{H_{{\mathrm{1}}}} | $ 
  and $\ottmv{x} \, \not\in \, \mathsf{dom} \,   \ottnt{H_{{\mathrm{1}}}}  ,  \ottnt{H_{{\mathrm{2}}}}  $
  and $ \,\text{fv}\,  \ottnt{a_{{\mathrm{1}}}}  \cap  \mathsf{dom} \,  \ottnt{H_{{\mathrm{2}}}}  = \emptyset$, 
then \[
    [    \ottnt{H_{{\mathrm{1}}}}  ,   \ottmv{x}  \overset{ { \color{black}{q} } }{\mapsto} { \Gamma_{{\mathrm{1}}} \vdash  \ottnt{a_{{\mathrm{1}}}}  :  \ottnt{A_{{\mathrm{1}}}} }    ,  \ottnt{H_{{\mathrm{2}}}}   ]\,  \ottnt{a}  \Rightarrow_{ \ottnt{S} }^{ { \color{black}{r} } } [    \ottnt{H'_{{\mathrm{1}}}}  ,   \ottmv{x}  \overset{ { \color{black}{q} } }{\mapsto} { \Gamma_{{\mathrm{1}}} \vdash  \ottnt{a_{{\mathrm{1}}}}  :  \ottnt{A_{{\mathrm{1}}}} }    ,  \ottnt{H'}  \, ;\,    \mathbf{u}'_{{\mathrm{1}}}  \mathop{\diamond}    { \color{black}{0} }     \mathop{\diamond}  \mathbf{u}  \, ;\,  \Gamma_{{\mathrm{4}}}  ]\,  \ottnt{a'} 
\]
\end{lemma}

\begin{lemma}[Compatibility weakening]
\label{lem:compatibility-weakening}
Let
$  \ottnt{H_{{\mathrm{1}}}}  ,  \ottnt{H_{{\mathrm{2}}}}   \vdash   \Delta ,  \Delta'  ;   \ottsym{(}  \ottsym{(}   { \color{black}{r} }   \cdot   \Gamma_{{\mathrm{11}}}   \ottsym{)}  \ottsym{+}  \Gamma_{{\mathrm{0}}}  \ottsym{)} ,  \Gamma_{{\mathrm{1}}}   $
and $ \Delta \ ;\  \Gamma_{{\mathrm{11}}}  \vdash \ottnt{a} : \ottnt{A} $
and $ | \ottnt{H_{{\mathrm{1}}}} |  =  | \Gamma_{{\mathrm{0}}} |  =  | \Delta | $
and $\ottmv{x} \, \not\in \, \mathsf{dom} \,   \ottnt{H_{{\mathrm{1}}}}  ,  \ottnt{H_{{\mathrm{2}}}}  $. 
Let $\ottnt{H'_{{\mathrm{2}}}}$ be $\ottnt{H_{{\mathrm{2}}}}$ with the embedded contexts weakened by inserting $ \ottmv{x}  \! = \!  \ottnt{a}  \! :^{  { \color{black}{0} }  } \!  \ottnt{A} $ at the $ | \Delta | $ position. Then, 
\[    \ottnt{H_{{\mathrm{1}}}}  ,   \ottmv{x}  \overset{ { \color{black}{r} } }{\mapsto} { \Gamma_{{\mathrm{11}}} \vdash  \ottnt{a}  :  \ottnt{A} }    ,  \ottnt{H'_{{\mathrm{2}}}}   \vdash    \Delta ,   \ottmv{x} \! = \!  \ottnt{a}  \! : \!  \ottnt{A}   ,  \Delta'  ;    \Gamma_{{\mathrm{0}}} ,   \ottmv{x}  \! = \!  \ottnt{a}  \! :^{ { \color{black}{r} } } \!  \ottnt{A}   ,  \Gamma_{{\mathrm{1}}}    \]
\end{lemma}

The soundness theorem follows as an instance of the invariance lemma below. The lemma provides a strong enough hypothesis for the induction to go through.

\begin{lemma}[Invariance] \label{inv}
If $ \ottnt{H}  \vdash  \Delta ;  \Gamma_{{\mathrm{0}}}  \ottsym{+}   { \color{black}{q} }   \cdot   \Gamma  $ and $ \Delta \ ;\  \Gamma  \vdash \ottnt{a} : \ottnt{A} $ and $1 \leq q$ and $ \mathsf{dom} \,  \Delta  \subseteq \ottnt{S}$, then either $a$ is a value or there exists $\Gamma'$, $\ottnt{H'}$, $\mathbf{u}'$, $\Gamma_{{\mathrm{4}}}$ and $\ottnt{a'}$ such that:
\begin{itemize}
\item $ [  \ottnt{H}  ]\,  \ottnt{a}  \Rightarrow_{ \ottnt{S} }^{ { \color{black}{q} } } [  \ottnt{H'} \, ;\,  \mathbf{u}' \, ;\,  \Gamma_{{\mathrm{4}}}  ]\,  \ottnt{a'} $
\item $ \ottnt{H'}  \vdash   \Delta ,   \lfloor \Gamma_{{\mathrm{4}}} \rfloor   ;  \ottsym{(}   \Gamma_{{\mathrm{0}}} ,    { \color{black}{0} }    \cdot   \Gamma_{{\mathrm{4}}}    \ottsym{)}  \ottsym{+}   { \color{black}{q} }   \cdot   \Gamma'  $
\item $  \Delta ,   \lfloor \Gamma_{{\mathrm{4}}} \rfloor   \ ;\  \Gamma'  \vdash \ottnt{a'} : \ottnt{A} $
\item $  { \color{black}{q} }   \cdot    \ncoverline{ \Gamma' }    +  \mathbf{u}'  +   \mathbf{0}   \mathop{\diamond}   \ncoverline{ \Gamma_{{\mathrm{4}}} }   \times \langle H'\rangle \leq  { \color{black}{q} }   \cdot    (   \ncoverline{ \Gamma }   \mathop{\diamond}   \mathbf{0}   )   + \mathbf{u}' \times \langle H'\rangle +   \mathbf{0}   \mathop{\diamond}   \ncoverline{ \Gamma_{{\mathrm{4}}} }  $
\item $ \mathsf{dom} \,  \Gamma_{{\mathrm{4}}} $ is disjoint from $S$
\end{itemize}
\end{lemma}

\begin{proof}
Let  $ \ottnt{H}  \vdash  \Delta ;  \Gamma_{{\mathrm{0}}}  \ottsym{+}   { \color{black}{q} }   \cdot   \Gamma  $ and $ \Delta \ ;\  \Gamma  \vdash \ottnt{a} : \ottnt{A} $. We prove this lemma by induction on the typing judgement $ \Delta \ ;\  \Gamma  \vdash \ottnt{a} : \ottnt{A} $.
\ifextended\else For brevity, we show only the most interesting cases. \fi 

\begin{itemize}
\item \textbf{\rref{T-sub}}

Let $ \Delta \ ;\  \Gamma_{{\mathrm{2}}}  \vdash \ottnt{a} : \ottnt{A} $ where $ \Delta \ ;\  \Gamma_{{\mathrm{1}}}  \vdash \ottnt{a} : \ottnt{A} $ and $\Gamma_{{\mathrm{1}}} \leq \Gamma_{{\mathrm{2}}}$. Further, $ \ottnt{H}  \vdash  \Delta ;  \Gamma_{{\mathrm{0}}}  \ottsym{+}   { \color{black}{q} }   \cdot   \Gamma_{{\mathrm{2}}}  $.

Since $ \ottnt{H}  \vdash  \Delta ;  \Gamma_{{\mathrm{0}}}  \ottsym{+}   { \color{black}{q} }   \cdot   \Gamma_{{\mathrm{2}}}  $ and $\Gamma_{{\mathrm{1}}} \leq \Gamma_{{\mathrm{2}}}$, by lemma \ref{subheap}, there exists $H'$ such that $ \ottnt{H'}  \vdash  \Delta ;  \Gamma_{{\mathrm{0}}}  \ottsym{+}   { \color{black}{q} }   \cdot   \Gamma_{{\mathrm{1}}}  $ and $H' \leq H$. By inductive hypothesis, $ [  \ottnt{H'}  ]\,  \ottnt{a}  \Rightarrow_{ \ottnt{S} }^{ { \color{black}{q} } } [  \ottnt{H''} \, ;\,  \mathbf{u} \, ;\,  \Gamma_{{\mathrm{4}}}  ]\,  \ottnt{a''} $. Since $H' \leq H$, we have, $ [  \ottnt{H}  ]\,  \ottnt{a}  \Rightarrow_{ \ottnt{S} }^{ { \color{black}{q} } } [  \ottnt{H''} \, ;\,  \mathbf{u} \, ;\,  \Gamma_{{\mathrm{4}}}  ]\,  \ottnt{a''} $. The remaining clauses follow from the inductive hypothesis and the fact that $\Gamma_{{\mathrm{1}}} \leq \Gamma_{{\mathrm{2}}}$.

\item We don't need to consider \rref{T-var} and \rref{T-weak} since for $ \ottnt{H}  \vdash  \Delta ;  \Gamma $, any variable $x \in \text{ dom } \Delta$ should be a definition of the form $ \ottmv{x} \! = \!  \ottnt{a}  \! : \!  \ottnt{A} $.

\item \textbf{\rref{T-def}}

Let $  \Delta ,   \ottmv{x} \! = \!  \ottnt{a}  \! : \!  \ottnt{A}   \ ;\     { \color{black}{0} }    \cdot   \Gamma  ,   \ottmv{x}  \! = \!  \ottnt{a}  \! :^{  { \color{black}{1} }  } \!  \ottnt{A}    \vdash \ottmv{x} : \ottnt{A} $ where $ \Delta \ ;\  \Gamma  \vdash \ottnt{a} : \ottnt{A} $ and $\ottmv{x} \, \not\in \, \mathsf{dom} \, \Delta$. Further, $ \ottnt{H}  \vdash   \Delta ,   \ottmv{x} \! = \!  \ottnt{a}  \! : \!  \ottnt{A}   ;  \Gamma_{{\mathrm{0}}}  \ottsym{+}   { \color{black}{q} }   \cdot   \ottsym{(}     { \color{black}{0} }    \cdot   \Gamma  ,   \ottmv{x}  \! = \!  \ottnt{a}  \! :^{  { \color{black}{1} }  } \!  \ottnt{A}    \ottsym{)}  $. Let $\Gamma_{{\mathrm{0}}} =  \Gamma'_{{\mathrm{0}}} ,   \ottmv{x}  \! = \!  \ottnt{a}  \! :^{ { \color{black}{r} } } \!  \ottnt{A}  $. So $ \ottnt{H}  \vdash   \Delta ,   \ottmv{x} \! = \!  \ottnt{a}  \! : \!  \ottnt{A}   ;   \Gamma'_{{\mathrm{0}}} ,   \ottmv{x}  \! = \!  \ottnt{a}  \! :^{ \ottsym{(}  { \color{black}{r} }  \ottsym{+}  { \color{black}{q} }  \ottsym{)} } \!  \ottnt{A}   $. Therefore, $H =  \ottnt{H_{{\mathrm{1}}}}  ,   \ottmv{x}  \overset{ \ottsym{(}  { \color{black}{r} }  \ottsym{+}  { \color{black}{q} }  \ottsym{)} }{\mapsto} { \Gamma_{{\mathrm{11}}} \vdash  \ottnt{a}  :  \ottnt{A} }  $ where $ \Delta \ ;\  \Gamma_{{\mathrm{11}}}  \vdash \ottnt{a} : \ottnt{A} $.

Since $1 \leq q$, we have, $ [   \ottnt{H_{{\mathrm{1}}}}  ,   \ottmv{x}  \overset{ \ottsym{(}  { \color{black}{r} }  \ottsym{+}  { \color{black}{q} }  \ottsym{)} }{\mapsto} { \Gamma_{{\mathrm{11}}} \vdash  \ottnt{a}  :  \ottnt{A} }    ]\,  \ottmv{x}  \Rightarrow_{ \ottnt{S} }^{ { \color{black}{q} } } [   \ottnt{H_{{\mathrm{1}}}}  ,   \ottmv{x}  \overset{ { \color{black}{r} } }{\mapsto} { \Gamma_{{\mathrm{11}}} \vdash  \ottnt{a}  :  \ottnt{A} }   \, ;\,    \mathbf{0}^{| \ottnt{H_{{\mathrm{1}}}} |}   \mathop{\diamond}   { \color{black}{q} }   \, ;\,  \varnothing  ]\,  \ottnt{a} $. Also, $  \ottnt{H_{{\mathrm{1}}}}  ,   \ottmv{x}  \overset{ { \color{black}{r} } }{\mapsto} { \Gamma_{{\mathrm{11}}} \vdash  \ottnt{a}  :  \ottnt{A} }    \vdash   \Delta ,   \ottmv{x} \! = \!  \ottnt{a}  \! : \!  \ottnt{A}   ;  \ottsym{(}   \Gamma'_{{\mathrm{0}}} ,   \ottmv{x}  \! = \!  \ottnt{a}  \! :^{ { \color{black}{r} } } \!  \ottnt{A}    \ottsym{)}  \ottsym{+}   { \color{black}{q} }   \cdot   \ottsym{(}   \Gamma_{{\mathrm{11}}} ,   \ottmv{x}  \! = \!  \ottnt{a}  \! :^{  { \color{black}{0} }  } \!  \ottnt{A}    \ottsym{)}  $. By weakening, we have, $  \Delta ,   \ottmv{x} \! = \!  \ottnt{a}  \! : \!  \ottnt{A}   \ ;\   \Gamma_{{\mathrm{11}}} ,   \ottmv{x}  \! = \!  \ottnt{a}  \! :^{  { \color{black}{0} }  } \!  \ottnt{A}    \vdash \ottnt{a} : \ottnt{A} $. The fourth clause: $  { \color{black}{q} }   \cdot    (   \ncoverline{ \Gamma_{{\mathrm{11}}} }   \mathop{\diamond}    { \color{black}{0} }    )    +   (   \mathbf{0}   \mathop{\diamond}   { \color{black}{q} }   )   \leq     \mathbf{0}   \mathop{\diamond}   { \color{black}{q} }     +   (   \mathbf{0}   \mathop{\diamond}   { \color{black}{q} }   )   \times \big( \begin{smallmatrix} \langle H_1 \rangle &  \mathbf{0} ^\intercal \\  \ncoverline{ \Gamma_{{\mathrm{11}}} }  & 0 \end{smallmatrix} \big)$ follows by reflexivity.

\item \textbf{\rref{T-weak-def}}

Let $  \Delta ,   \ottmv{x} \! = \!  \ottnt{a}  \! : \!  \ottnt{A}   \ ;\   \Gamma ,   \ottmv{x}  \! = \!  \ottnt{a}  \! :^{  { \color{black}{0} }  } \!  \ottnt{A}    \vdash \ottnt{b} : \ottnt{B} $ where $ \Delta \ ;\  \Gamma  \vdash \ottnt{b} : \ottnt{B} $ and $ \Delta \ ;\  \Gamma_{{\mathrm{9}}}  \vdash \ottnt{a} : \ottnt{A} $  and $\ottmv{x} \, \not\in \, \mathsf{dom} \, \Delta$. Further, $ \ottnt{H}  \vdash   \Delta ,   \ottmv{x} \! = \!  \ottnt{a}  \! : \!  \ottnt{A}   ;  \Gamma_{{\mathrm{0}}}  \ottsym{+}   { \color{black}{q} }   \cdot   \ottsym{(}   \Gamma ,   \ottmv{x}  \! = \!  \ottnt{a}  \! :^{  { \color{black}{0} }  } \!  \ottnt{A}    \ottsym{)}  $. Let $\Gamma_{{\mathrm{0}}} =  \Gamma'_{{\mathrm{0}}} ,   \ottmv{x}  \! = \!  \ottnt{a}  \! :^{ { \color{black}{r} } } \!  \ottnt{A}  $. So $ \ottnt{H}  \vdash   \Delta ,   \ottmv{x} \! = \!  \ottnt{a}  \! : \!  \ottnt{A}   ;   \Gamma'_{{\mathrm{0}}}  \ottsym{+}    { \color{black}{q} }   \cdot   \Gamma   ,   \ottmv{x}  \! = \!  \ottnt{a}  \! :^{ { \color{black}{r} } } \!  \ottnt{A}   $. Therefore, $H =  \ottnt{H_{{\mathrm{1}}}}  ,   \ottmv{x}  \overset{ { \color{black}{r} } }{\mapsto} { \Gamma_{{\mathrm{11}}} \vdash  \ottnt{a}  :  \ottnt{A} }  $ where $ \Delta \ ;\  \Gamma_{{\mathrm{11}}}  \vdash \ottnt{a} : \ottnt{A} $. Also, $ \ottnt{H_{{\mathrm{1}}}}  \vdash  \Delta ;   \Gamma'_{{\mathrm{0}}}  \ottsym{+}   { \color{black}{q} }   \cdot   \Gamma    \ottsym{+}   { \color{black}{r} }   \cdot   \Gamma_{{\mathrm{11}}}  $.

Applying the inductive hypothesis, we get, $ [  \ottnt{H_{{\mathrm{1}}}}  ]\,  \ottnt{b}  \Rightarrow_{   \ottnt{S}  \, \cup   \{  \ottmv{x}  \}    }^{ { \color{black}{q} } } [   \ottnt{H'_{{\mathrm{1}}}}  ,  \ottnt{H_{{\mathrm{4}}}}  \, ;\,   \mathbf{u}_{{\mathrm{1}}}  \mathop{\diamond}  \mathbf{u}_{{\mathrm{4}}}  \, ;\,  \Gamma_{{\mathrm{4}}}  ]\,  \ottnt{b'} $ and $  \ottnt{H'_{{\mathrm{1}}}}  ,  \ottnt{H_{{\mathrm{4}}}}   \vdash   \Delta ,   \lfloor \Gamma_{{\mathrm{4}}} \rfloor   ;  \ottsym{(}    \Gamma'_{{\mathrm{0}}}  \ottsym{+}   { \color{black}{r} }   \cdot   \Gamma_{{\mathrm{11}}}   ,    { \color{black}{0} }    \cdot   \Gamma_{{\mathrm{4}}}    \ottsym{)}  \ottsym{+}   { \color{black}{q} }   \cdot   \ottsym{(}   \Gamma' ,  \Gamma''   \ottsym{)}  $ and $  \Delta ,   \lfloor \Gamma_{{\mathrm{4}}} \rfloor   \ ;\   \Gamma' ,  \Gamma''   \vdash \ottnt{b'} : \ottnt{B} $. Here, $ | \ottnt{H'_{{\mathrm{1}}}} |  =  | \mathbf{u}_{{\mathrm{1}}} |  =  | \ottnt{H_{{\mathrm{1}}}} |  =  | \Gamma' |  =  | \Delta | $. Now, note that $x$ does not appear in $\ottnt{H_{{\mathrm{4}}}}$. Let $\ottnt{H'_{{\mathrm{4}}}}$ be $\ottnt{H_{{\mathrm{4}}}}$ with the embedded contexts weakened by inserting $ \ottmv{x}  \! = \!  \ottnt{a}  \! :^{  { \color{black}{0} }  } \!  \ottnt{A} $ at the $ | \Delta | $ position. Then, by \ref{lem:compatibility-weakening}, $    \ottnt{H'_{{\mathrm{1}}}}  ,   \ottmv{x}  \overset{ { \color{black}{r} } }{\mapsto} { \Gamma_{{\mathrm{11}}} \vdash  \ottnt{a}  :  \ottnt{A} }     ,  \ottnt{H'_{{\mathrm{4}}}}   \vdash     \Delta ,   \ottmv{x} \! = \!  \ottnt{a}  \! : \!  \ottnt{A}    ,   \lfloor \Gamma_{{\mathrm{4}}} \rfloor   ;  \ottsym{(}     \Gamma'_{{\mathrm{0}}} ,   \ottmv{x}  \! = \!  \ottnt{a}  \! :^{ { \color{black}{r} } } \!  \ottnt{A}    ,    { \color{black}{0} }    \cdot   \Gamma_{{\mathrm{4}}}    \ottsym{)}  \ottsym{+}   { \color{black}{q} }   \cdot   \ottsym{(}     \Gamma' ,   \ottmv{x}  \! = \!  \ottnt{a}  \! :^{  { \color{black}{0} }  } \!  \ottnt{A}    ,  \Gamma''   \ottsym{)}  $. Because extra assignments do not impact the evaluation, we have, 
by \ref{lem:smallstep-weakening}, $ [   \ottnt{H_{{\mathrm{1}}}}  ,   \ottmv{x}  \overset{ { \color{black}{r} } }{\mapsto} { \Gamma_{{\mathrm{11}}} \vdash  \ottnt{a}  :  \ottnt{A} }    ]\,  \ottnt{b}  \Rightarrow_{ \ottnt{S} }^{ { \color{black}{q} } } [     \ottnt{H'_{{\mathrm{1}}}}  ,   \ottmv{x}  \overset{ { \color{black}{r} } }{\mapsto} { \Gamma_{{\mathrm{11}}} \vdash  \ottnt{a}  :  \ottnt{A} }     ,  \ottnt{H'_{{\mathrm{4}}}}  \, ;\,     \mathbf{u}_{{\mathrm{1}}}  \mathop{\diamond}    { \color{black}{0} }      \mathop{\diamond}  \mathbf{u}_{{\mathrm{4}}}  \, ;\,  \Gamma_{{\mathrm{4}}}  ]\,  \ottnt{b'} $. Also, by weakening lemma \ref{weakdefn}, $    \Delta ,   \ottmv{x} \! = \!  \ottnt{a}  \! : \!  \ottnt{A}    ,   \lfloor \Gamma_{{\mathrm{4}}} \rfloor   \ ;\     \Gamma' ,   \ottmv{x}  \! = \!  \ottnt{a}  \! :^{  { \color{black}{0} }  } \!  \ottnt{A}    ,  \Gamma''   \vdash \ottnt{b'} : \ottnt{B} $.  The fourth clause follows from the inductive hypothesis after inserting $0$ at the $ | \Delta | $-position on both sides.

\item \textbf{\rref{T-app}}

Let $ \Delta \ ;\  \Gamma_{{\mathrm{1}}}  \ottsym{+}   { \color{black}{r} }   \cdot   \Gamma_{{\mathrm{2}}}   \vdash \ottnt{b} \, \ottnt{a} : \ottnt{B}  \ottsym{\{}  \ottnt{a}  \ottsym{/}  \ottmv{x}  \ottsym{\}} $ where $ \Delta \ ;\  \Gamma_{{\mathrm{1}}}  \vdash \ottnt{b} :  \Pi  \ottmv{x} \!:^ { \color{black}{r} } \! \ottnt{A} . \ottnt{B}  $ and $ \Delta \ ;\  \Gamma_{{\mathrm{2}}}  \vdash \ottnt{a} : \ottnt{A} $. Further, $ \ottnt{H}  \vdash  \Delta ;  \Gamma_{{\mathrm{0}}}  \ottsym{+}   { \color{black}{q} }   \cdot   \ottsym{(}  \Gamma_{{\mathrm{1}}}  \ottsym{+}   { \color{black}{r} }   \cdot   \Gamma_{{\mathrm{2}}}   \ottsym{)}  $. Now, there are two cases to consider depending on whether $b$ is a value or not.

\begin{itemize}
 \item $b$ is not a value.
 
 In this case, we get from the inductive hypothesis, $ [  \ottnt{H}  ]\,  \ottnt{b}  \Rightarrow_{   \ottnt{S}  \, \cup   \,\text{fv}\,  \ottnt{a}    }^{ { \color{black}{q} } } [  \ottnt{H'} \, ;\,  \mathbf{u}' \, ;\,  \Gamma_{{\mathrm{4}}}  ]\,  \ottnt{b'} $ and $  \Delta ,   \lfloor \Gamma_{{\mathrm{4}}} \rfloor   \ ;\  \Gamma'_{{\mathrm{1}}}  \vdash \ottnt{b'} :  \Pi  \ottmv{x} \!:^ { \color{black}{r} } \! \ottnt{A} . \ottnt{B}  $ and $ \ottnt{H'}  \vdash   \Delta ,   \lfloor \Gamma_{{\mathrm{4}}} \rfloor   ;  \ottsym{(}    \Gamma_{{\mathrm{0}}}  \ottsym{+}   \ottsym{(}  { \color{black}{q} }  \cdot  { \color{black}{r} }  \ottsym{)}   \cdot   \Gamma_{{\mathrm{2}}}   ,    { \color{black}{0} }    \cdot   \Gamma_{{\mathrm{4}}}    \ottsym{)}  \ottsym{+}   { \color{black}{q} }   \cdot   \Gamma'_{{\mathrm{1}}}  $. So $ [  \ottnt{H}  ]\,  \ottnt{b} \, \ottnt{a}  \Rightarrow_{ \ottnt{S} }^{ { \color{black}{q} } } [  \ottnt{H'} \, ;\,  \mathbf{u}' \, ;\,  \Gamma_{{\mathrm{4}}}  ]\,  \ottnt{b'} \, \ottnt{a} $ By weakening, we get, $  \Delta ,   \lfloor \Gamma_{{\mathrm{4}}} \rfloor   \ ;\   \Gamma_{{\mathrm{2}}} ,    { \color{black}{0} }    \cdot   \Gamma_{{\mathrm{4}}}    \vdash \ottnt{a} : \ottnt{A} $. Therefore, by App, we have, $  \Delta ,   \lfloor \Gamma_{{\mathrm{4}}} \rfloor   \ ;\  \Gamma'_{{\mathrm{1}}}  \ottsym{+}   { \color{black}{r} }   \cdot   \ottsym{(}   \Gamma_{{\mathrm{2}}} ,    { \color{black}{0} }    \cdot   \Gamma_{{\mathrm{4}}}    \ottsym{)}   \vdash \ottnt{b'} \, \ottnt{a} : \ottnt{B}  \ottsym{\{}  \ottnt{a}  \ottsym{/}  \ottmv{x}  \ottsym{\}} $. Also, by rearranging, we get $ \ottnt{H'}  \vdash   \Delta ,   \lfloor \Gamma_{{\mathrm{4}}} \rfloor   ;  \ottsym{(}   \Gamma_{{\mathrm{0}}} ,    { \color{black}{0} }    \cdot   \Gamma_{{\mathrm{4}}}    \ottsym{)}  \ottsym{+}   { \color{black}{q} }   \cdot   \ottsym{(}  \Gamma'_{{\mathrm{1}}}  \ottsym{+}   { \color{black}{r} }   \cdot   \ottsym{(}   \Gamma_{{\mathrm{2}}} ,    { \color{black}{0} }    \cdot   \Gamma_{{\mathrm{4}}}    \ottsym{)}   \ottsym{)}  $. The fourth clause follows from the corresponding clause of inductive hypothesis.
 
 \item $b$ is a value.
 
 Since $b$ has a $\Pi$-type, it must be headed by a $\lambda$. Let $\ottnt{b}  \ottsym{=}   \lambda \ottmv{y} \!:^ { \color{black}{s} } \! \ottnt{A_{{\mathrm{1}}}} . \ottnt{b_{{\mathrm{1}}}} $ for some sufficiently fresh variable $y$. Then, we have, $ \Delta \ ;\  \Gamma_{{\mathrm{1}}}  \vdash  \lambda \ottmv{y} \!:^ { \color{black}{s} } \! \ottnt{A_{{\mathrm{1}}}} . \ottnt{b_{{\mathrm{1}}}}  :  \Pi  \ottmv{x} \!:^ { \color{black}{r} } \! \ottnt{A} . \ottnt{B}  $. By inversion, there exists $\ottnt{B_{{\mathrm{1}}}}$ such that $  \Delta ,   \ottmv{y} \!\!:\!\! \ottnt{A_{{\mathrm{1}}}}   \ ;\   \Gamma_{{\mathrm{1}}} ,   \ottmv{y} \! :^{ { \color{black}{s} } }\! \ottnt{A_{{\mathrm{1}}}}    \vdash \ottnt{b_{{\mathrm{1}}}} : \ottnt{B_{{\mathrm{1}}}} $ and $  \Delta \ ;\  \Gamma_{{\mathrm{7}}}  \vdash  \Pi  \ottmv{y} \!:^ { \color{black}{s} } \! \ottnt{A_{{\mathrm{1}}}} . \ottnt{B_{{\mathrm{1}}}}  :  \textbf{Type}  $ and $ \ottsym{(}   \Pi  \ottmv{y} \!:^ { \color{black}{s} } \! \ottnt{A_{{\mathrm{1}}}} . \ottnt{B_{{\mathrm{1}}}}   \ottsym{)}  \{  \Delta  \}   \equiv   \ottsym{(}   \Pi  \ottmv{x} \!:^ { \color{black}{r} } \! \ottnt{A} . \ottnt{B}   \ottsym{)}  \{  \Delta  \} $. By definition \ref{defeql}, $s = r$ and $ \ottnt{A_{{\mathrm{1}}}}  \{  \Delta  \}   \equiv   \ottnt{A}  \{  \Delta  \} $ and $ \ottnt{B_{{\mathrm{1}}}}  \{  \Delta  \}   \equiv   \ottnt{B}  \ottsym{\{}  \ottmv{y}  \ottsym{/}  \ottmv{x}  \ottsym{\}}  \{  \Delta  \} $. Now, by \rref{T-Conv}, $ \Delta \ ;\  \Gamma_{{\mathrm{2}}}  \vdash \ottnt{a} : \ottnt{A_{{\mathrm{1}}}} $. Therefore, by lemma \ref{InsertEq}, $  \Delta ,   \ottmv{y} \! = \!  \ottnt{a}  \! : \!  \ottnt{A_{{\mathrm{1}}}}   \ ;\   \Gamma_{{\mathrm{1}}} ,   \ottmv{y}  \! = \!  \ottnt{a}  \! :^{ { \color{black}{r} } } \!  \ottnt{A_{{\mathrm{1}}}}    \vdash \ottnt{b_{{\mathrm{1}}}} : \ottnt{B_{{\mathrm{1}}}} $.
 
We have, $ [  \ottnt{H}  ]\,  \ottsym{(}   \lambda \ottmv{y} \!:^ { \color{black}{r} } \! \ottnt{A_{{\mathrm{1}}}} . \ottnt{b_{{\mathrm{1}}}}   \ottsym{)} \, \ottnt{a}  \Rightarrow_{ \ottnt{S} }^{ { \color{black}{q} } } [   \ottnt{H}  ,   \ottmv{y}  \overset{ \ottsym{(}  { \color{black}{q} }  \cdot  { \color{black}{r} }  \ottsym{)} }{\mapsto} { \Gamma_{{\mathrm{2}}} \vdash  \ottnt{a}  :  \ottnt{A_{{\mathrm{1}}}} }   \, ;\,   \mathbf{0}  \, ;\,   \ottmv{y}  \! = \!  \ottnt{a}  \! :^{ \ottsym{(}  { \color{black}{q} }  \cdot  { \color{black}{r} }  \ottsym{)} } \!  \ottnt{A_{{\mathrm{1}}}}   ]\,  \ottnt{b_{{\mathrm{1}}}} $. Again, since $ \ottnt{H}  \vdash  \Delta ;   \Gamma_{{\mathrm{0}}}  \ottsym{+}   { \color{black}{q} }   \cdot   \Gamma_{{\mathrm{1}}}    \ottsym{+}   \ottsym{(}  { \color{black}{q} }  \cdot  { \color{black}{r} }  \ottsym{)}   \cdot   \Gamma_{{\mathrm{2}}}  $, we get, $  \ottnt{H}  ,   \ottmv{y}  \overset{ \ottsym{(}  { \color{black}{q} }  \cdot  { \color{black}{r} }  \ottsym{)} }{\mapsto} { \Gamma_{{\mathrm{2}}} \vdash  \ottnt{a}  :  \ottnt{A_{{\mathrm{1}}}} }    \vdash   \Delta ,   \ottmv{y} \! = \!  \ottnt{a}  \! : \!  \ottnt{A_{{\mathrm{1}}}}   ;  \ottsym{(}   \Gamma_{{\mathrm{0}}} ,   \ottmv{y}  \! = \!  \ottnt{a}  \! :^{  { \color{black}{0} }  } \!  \ottnt{A_{{\mathrm{1}}}}    \ottsym{)}  \ottsym{+}   { \color{black}{q} }   \cdot   \ottsym{(}   \Gamma_{{\mathrm{1}}} ,   \ottmv{y}  \! = \!  \ottnt{a}  \! :^{ { \color{black}{r} } } \!  \ottnt{A_{{\mathrm{1}}}}    \ottsym{)}  $.
 
By regularity, we know that $ \Delta \ ;\  \Gamma'  \vdash \ottnt{B}  \ottsym{\{}  \ottnt{a}  \ottsym{/}  \ottmv{x}  \ottsym{\}} :  \textbf{Type}  $. By weakening, $  \Delta ,   \ottmv{y} \! = \!  \ottnt{a}  \! : \!  \ottnt{A_{{\mathrm{1}}}}   \ ;\   \Gamma' ,   \ottmv{y}  \! = \!  \ottnt{a}  \! :^{  { \color{black}{0} }  } \!  \ottnt{A_{{\mathrm{1}}}}    \vdash \ottnt{B}  \ottsym{\{}  \ottnt{a}  \ottsym{/}  \ottmv{x}  \ottsym{\}} :  \textbf{Type}  $. But $ \ottnt{B}  \ottsym{\{}  \ottnt{a}  \ottsym{/}  \ottmv{x}  \ottsym{\}}  \{   \Delta ,   \ottmv{y} \! = \!  \ottnt{a}  \! : \!  \ottnt{A_{{\mathrm{1}}}}    \}   \ottsym{=}   \ottnt{B}  \ottsym{\{}  \ottnt{a}  \ottsym{/}  \ottmv{x}  \ottsym{\}}  \{  \Delta  \}  =  \ottnt{B}  \ottsym{\{}  \ottmv{y}  \ottsym{/}  \ottmv{x}  \ottsym{\}}  \ottsym{\{}  \ottnt{a}  \ottsym{/}  \ottmv{y}  \ottsym{\}}  \{  \Delta  \}  =  \ottnt{B}  \ottsym{\{}  \ottmv{y}  \ottsym{/}  \ottmv{x}  \ottsym{\}}  \{  \Delta  \}   \ottsym{\{}   \ottnt{a}  \{  \Delta  \}   \ottsym{/}  \ottmv{y}  \ottsym{\}}  \equiv   \ottnt{B_{{\mathrm{1}}}}  \{  \Delta  \}   \ottsym{\{}   \ottnt{a}  \{  \Delta  \}   \ottsym{/}  \ottmv{y}  \ottsym{\}} =  \ottnt{B_{{\mathrm{1}}}}  \ottsym{\{}  \ottnt{a}  \ottsym{/}  \ottmv{y}  \ottsym{\}}  \{  \Delta  \}   \ottsym{=}   \ottnt{B_{{\mathrm{1}}}}  \{   \Delta ,   \ottmv{y} \! = \!  \ottnt{a}  \! : \!  \ottnt{A_{{\mathrm{1}}}}    \} $. Hence, by \rref{T-Conv}, $  \Delta ,   \ottmv{y} \! = \!  \ottnt{a}  \! : \!  \ottnt{A_{{\mathrm{1}}}}   \ ;\   \Gamma_{{\mathrm{1}}} ,   \ottmv{y}  \! = \!  \ottnt{a}  \! :^{ { \color{black}{r} } } \!  \ottnt{A_{{\mathrm{1}}}}    \vdash \ottnt{b_{{\mathrm{1}}}} : \ottnt{B}  \ottsym{\{}  \ottnt{a}  \ottsym{/}  \ottmv{x}  \ottsym{\}} $. The fourth clause: $   { \color{black}{q} }   \cdot    (   \ncoverline{ \Gamma_{{\mathrm{1}}} }   \mathop{\diamond}   { \color{black}{r} }   )    +   \mathbf{0}    +   (   \mathbf{0}   \mathop{\diamond}   \ottsym{(}  { \color{black}{q} }  \cdot  { \color{black}{r} }  \ottsym{)}   )   \times \big( \begin{smallmatrix} \langle H \rangle &  \mathbf{0} ^\intercal \\  \ncoverline{ \Gamma_{{\mathrm{2}}} }  & 0 \end{smallmatrix} \big) \leq    { \color{black}{q} }   \cdot    (   (   \ncoverline{ \Gamma_{{\mathrm{1}}} }   +   { \color{black}{r} }   \cdot    \ncoverline{ \Gamma_{{\mathrm{2}}} }    )   \mathop{\diamond}    { \color{black}{0} }    )    +   \mathbf{0}    +   (   \mathbf{0}   \mathop{\diamond}   \ottsym{(}  { \color{black}{q} }  \cdot  { \color{black}{r} }  \ottsym{)}   )  $ follows by reflexivity.  
 
\end{itemize}

\item \textbf{\rref{T-conv}}

Let $ \Delta \ ;\  \Gamma  \vdash \ottnt{a} : \ottnt{B} $ where $ \Delta \ ;\  \Gamma  \vdash \ottnt{a} : \ottnt{A} $ and $ \Delta \ ;\  \Gamma_{{\mathrm{1}}}  \vdash \ottnt{B} :  \textbf{Type}  $ and $ \ottnt{A}  \{  \Delta  \}   \equiv   \ottnt{B}  \{  \Delta  \} $. Further, $ \ottnt{H}  \vdash  \Delta ;  \Gamma_{{\mathrm{0}}}  \ottsym{+}   { \color{black}{q} }   \cdot   \Gamma  $. Therefore, by inductive hypothesis, $ [  \ottnt{H}  ]\,  \ottnt{a}  \Rightarrow_{ \ottnt{S} }^{ { \color{black}{q} } } [  \ottnt{H'} \, ;\,  \mathbf{u}' \, ;\,  \Gamma_{{\mathrm{4}}}  ]\,  \ottnt{a'} $ and $  \Delta ,   \lfloor \Gamma_{{\mathrm{4}}} \rfloor   \ ;\  \Gamma'  \vdash \ottnt{a'} : \ottnt{A} $ and $ \ottnt{H'}  \vdash   \Delta ,   \lfloor \Gamma_{{\mathrm{4}}} \rfloor   ;  \ottsym{(}   \Gamma_{{\mathrm{0}}} ,    { \color{black}{0} }    \cdot   \Gamma_{{\mathrm{4}}}    \ottsym{)}  \ottsym{+}   { \color{black}{q} }   \cdot   \Gamma'  $.

Now, since $ \,\text{fv}\,  \ottnt{A}  \subseteq  \mathsf{dom} \,  \Delta $ and $ \,\text{fv}\,  \ottnt{B}  \subseteq  \mathsf{dom} \,  \Delta $; we have, $ \ottnt{A}  \{   \Delta ,   \lfloor \Gamma_{{\mathrm{4}}} \rfloor    \}   \ottsym{=}   \ottnt{A}  \{  \Delta  \} $ and $ \ottnt{B}  \{   \Delta ,   \lfloor \Gamma_{{\mathrm{4}}} \rfloor    \}   \ottsym{=}   \ottnt{B}  \{  \Delta  \} $. Therefore, $  \Delta ,   \lfloor \Gamma_{{\mathrm{4}}} \rfloor   \ ;\  \Gamma'  \vdash \ottnt{a'} : \ottnt{B} $. The other clauses follow from the inductive hypothesis.

\ifextended
\item \textbf{\rref{T-UnitElim}}

Let $ \Delta \ ;\  \Gamma_{{\mathrm{1}}}  \ottsym{+}  \Gamma_{{\mathrm{2}}}  \vdash  \ottkw{let}\,  \ottkw{unit} \,=\, \ottnt{a} \ \ottkw{in}\  \ottnt{b}  : \ottnt{B}  \ottsym{\{}  \ottnt{a}  \ottsym{/}  \ottmv{y}  \ottsym{\}} $ where $ \Delta \ ;\  \Gamma_{{\mathrm{1}}}  \vdash \ottnt{a} : \ottkw{Unit} $ and $ \Delta \ ;\  \Gamma_{{\mathrm{2}}}  \vdash \ottnt{b} : \ottnt{B}  \ottsym{\{}  \ottkw{unit}  \ottsym{/}  \ottmv{y}  \ottsym{\}} $ and $  \Delta ,   \ottmv{y} \!\!:\!\! \ottkw{Unit}   \ ;\   \Gamma ,   \ottmv{y} \! :^{ { \color{black}{r} } }\! \ottkw{Unit}    \vdash \ottnt{B} :  \textbf{Type}  $. Further, $ \ottnt{H}  \vdash  \Delta ;  \Gamma_{{\mathrm{0}}}  \ottsym{+}   { \color{black}{q} }   \cdot   \ottsym{(}  \Gamma_{{\mathrm{1}}}  \ottsym{+}  \Gamma_{{\mathrm{2}}}  \ottsym{)}  $. Now, there are two cases to consider depending on whether $a$ is a value or not.

\begin{itemize}
\item $a$ is not a value.

In this case, we get from the inductive hypothesis, $ [  \ottnt{H}  ]\,  \ottnt{a}  \Rightarrow_{   \ottnt{S}  \, \cup     \,\text{fv}\,  \ottnt{b}   \, \cup   \{  \ottmv{y}  \}      }^{ { \color{black}{q} } } [  \ottnt{H'} \, ;\,  \mathbf{u} \, ;\,  \Gamma_{{\mathrm{4}}}  ]\,  \ottnt{a'} $ and $  \Delta ,   \lfloor \Gamma_{{\mathrm{4}}} \rfloor   \ ;\  \Gamma'_{{\mathrm{1}}}  \vdash \ottnt{a'} : \ottkw{Unit} $ and $ \ottnt{H'}  \vdash   \Delta ,   \lfloor \Gamma_{{\mathrm{4}}} \rfloor   ;  \ottsym{(}   \ottsym{(}  \Gamma_{{\mathrm{0}}}  \ottsym{+}   { \color{black}{q} }   \cdot   \Gamma_{{\mathrm{2}}}   \ottsym{)} ,    { \color{black}{0} }    \cdot   \Gamma_{{\mathrm{4}}}    \ottsym{)}  \ottsym{+}   { \color{black}{q} }   \cdot   \Gamma'_{{\mathrm{1}}}  $. Therefore, $ [  \ottnt{H}  ]\,   \ottkw{let}\,  \ottkw{unit} \,=\, \ottnt{a} \ \ottkw{in}\  \ottnt{b}   \Rightarrow_{ \ottnt{S} }^{ { \color{black}{q} } } [  \ottnt{H'} \, ;\,  \mathbf{u} \, ;\,  \Gamma_{{\mathrm{4}}}  ]\,   \ottkw{let}\,  \ottkw{unit} \,=\, \ottnt{a'} \ \ottkw{in}\  \ottnt{b}  $. And, by weakening and \rref{UnitElim}, we have, $  \Delta ,   \lfloor \Gamma_{{\mathrm{4}}} \rfloor   \ ;\  \Gamma'_{{\mathrm{1}}}  \ottsym{+}  \ottsym{(}   \Gamma_{{\mathrm{2}}} ,    { \color{black}{0} }    \cdot   \Gamma_{{\mathrm{4}}}    \ottsym{)}  \vdash  \ottkw{let}\,  \ottkw{unit} \,=\, \ottnt{a'} \ \ottkw{in}\  \ottnt{b}  : \ottnt{B}  \ottsym{\{}  \ottnt{a'}  \ottsym{/}  \ottmv{y}  \ottsym{\}} $.\\

Now, by regularity, $ \Delta \ ;\  \Gamma'  \vdash \ottnt{B}  \ottsym{\{}  \ottnt{a}  \ottsym{/}  \ottmv{y}  \ottsym{\}} :  \textbf{Type}  $. By weakening, $  \Delta ,   \lfloor \Gamma_{{\mathrm{4}}} \rfloor   \ ;\   \Gamma' ,    { \color{black}{0} }    \cdot   \Gamma_{{\mathrm{4}}}    \vdash \ottnt{B}  \ottsym{\{}  \ottnt{a}  \ottsym{/}  \ottmv{y}  \ottsym{\}} :  \textbf{Type}  $. But then, $ \ottnt{B}  \ottsym{\{}  \ottnt{a}  \ottsym{/}  \ottmv{y}  \ottsym{\}}  \{   \Delta ,   \lfloor \Gamma_{{\mathrm{4}}} \rfloor    \}   \ottsym{=}   \ottnt{B}  \ottsym{\{}  \ottnt{a}  \ottsym{/}  \ottmv{y}  \ottsym{\}}  \{  \Delta  \}  =  \ottnt{B}  \{  \Delta  \}   \ottsym{\{}   \ottnt{a}  \{  \Delta  \}   \ottsym{/}  \ottmv{y}  \ottsym{\}}$. Also, $ \ottnt{B}  \ottsym{\{}  \ottnt{a'}  \ottsym{/}  \ottmv{y}  \ottsym{\}}  \{   \Delta ,   \lfloor \Gamma_{{\mathrm{4}}} \rfloor    \}   \ottsym{=}   \ottnt{B}  \{   \Delta ,   \lfloor \Gamma_{{\mathrm{4}}} \rfloor    \}   \ottsym{\{}   \ottnt{a'}  \{   \Delta ,   \lfloor \Gamma_{{\mathrm{4}}} \rfloor    \}   \ottsym{/}  \ottmv{y}  \ottsym{\}} =  \ottnt{B}  \{  \Delta  \}   \ottsym{\{}   \ottnt{a'}  \{   \Delta ,   \lfloor \Gamma_{{\mathrm{4}}} \rfloor    \}   \ottsym{/}  \ottmv{y}  \ottsym{\}}$. Again, $ \ottnt{a}  \{  \Delta  \}  =  \ottnt{a}  \{  \ottnt{H}  \}   \equiv   \ottnt{a'}  \{  \ottnt{H'}  \}  =  \ottnt{a'}  \{   \Delta ,   \lfloor \Gamma_{{\mathrm{4}}} \rfloor    \} $. Hence, by \rref{T-Conv}, $  \Delta ,   \lfloor \Gamma_{{\mathrm{4}}} \rfloor   \ ;\  \Gamma'_{{\mathrm{1}}}  \ottsym{+}  \ottsym{(}   \Gamma_{{\mathrm{2}}} ,    { \color{black}{0} }    \cdot   \Gamma_{{\mathrm{4}}}    \ottsym{)}  \vdash  \ottkw{let}\,  \ottkw{unit} \,=\, \ottnt{a'} \ \ottkw{in}\  \ottnt{b}  : \ottnt{B}  \ottsym{\{}  \ottnt{a}  \ottsym{/}  \ottmv{y}  \ottsym{\}} $. The fourth clause: $   { \color{black}{q} }   \cdot    (   \ncoverline{ \Gamma'_{{\mathrm{1}}} }   +   (   \ncoverline{ \Gamma_{{\mathrm{2}}} }   \mathop{\diamond}   \mathbf{0}   )   )    +  \mathbf{u}   +   (   \mathbf{0}   \mathop{\diamond}   \ncoverline{ \Gamma_{{\mathrm{4}}} }   )   \times \langle H' \rangle \leq   { \color{black}{q} }   \cdot    (   (   \ncoverline{ \Gamma_{{\mathrm{1}}} }   +   \ncoverline{ \Gamma_{{\mathrm{2}}} }   )   \mathop{\diamond}   \mathbf{0}   )    +  \mathbf{u}  \times \langle H' \rangle +   \mathbf{0}   \mathop{\diamond}   \ncoverline{ \Gamma_{{\mathrm{4}}} }  $ follows from the inductive hypothesis.

\item $a$ is a value.

Since $a$ has type $\ottkw{Unit}$, we have, $\ottnt{a}  \ottsym{=}  \ottkw{unit}$ and $  { \color{black}{0} }    \cdot   \Gamma   \leq  \Gamma_{{\mathrm{1}}}$, for some $\Gamma$. Now, $ [  \ottnt{H}  ]\,   \ottkw{let}\,  \ottkw{unit} \,=\, \ottkw{unit} \ \ottkw{in}\  \ottnt{b}   \Rightarrow_{ \ottnt{S} }^{ { \color{black}{q} } } [  \ottnt{H} \, ;\,   \mathbf{0}  \, ;\,  \varnothing  ]\,  \ottnt{b} $. By sub-usaging, $ \Delta \ ;\  \Gamma_{{\mathrm{1}}}  \ottsym{+}  \Gamma_{{\mathrm{2}}}  \vdash \ottnt{b} : \ottnt{B}  \ottsym{\{}  \ottkw{unit}  \ottsym{/}  \ottmv{y}  \ottsym{\}} $. The fourth clause: $ { \color{black}{q} }   \cdot    (   \ncoverline{ \Gamma_{{\mathrm{1}}} }   +   \ncoverline{ \Gamma_{{\mathrm{2}}} }   )   \leq  { \color{black}{q} }   \cdot    (   \ncoverline{ \Gamma_{{\mathrm{1}}} }   +   \ncoverline{ \Gamma_{{\mathrm{2}}} }   )  $ follows by reflexivity.
\end{itemize}

\item \textbf{\rref{T-SumElim}}

Let $ \Delta \ ;\    { \color{black}{r} }   \cdot   \Gamma_{{\mathrm{1}}}    \ottsym{+}  \Gamma_{{\mathrm{2}}}  \vdash  \ottkw{case}_ { \color{black}{r} } \,  \ottnt{a} \, \ottkw{of}\,  \ottnt{b_{{\mathrm{1}}}}  ;   \ottnt{b_{{\mathrm{2}}}}  : \ottnt{B}  \ottsym{\{}  \ottnt{a}  \ottsym{/}  \ottmv{y}  \ottsym{\}} $ where $ \Delta \ ;\  \Gamma_{{\mathrm{1}}}  \vdash \ottnt{a} :  \ottnt{A_{{\mathrm{1}}}}  \oplus  \ottnt{A_{{\mathrm{2}}}}  $ and $ \Delta \ ;\  \Gamma_{{\mathrm{2}}}  \vdash \ottnt{b_{{\mathrm{1}}}} :  \Pi  \ottmv{x} \!:^ { \color{black}{r} } \! \ottnt{A_{{\mathrm{1}}}} . \ottnt{B}   \ottsym{\{}   \ottkw{inj}_1\,  \ottmv{x}   \ottsym{/}  \ottmv{y}  \ottsym{\}} $ and $ \Delta \ ;\  \Gamma_{{\mathrm{2}}}  \vdash \ottnt{b_{{\mathrm{2}}}} :  \Pi  \ottmv{x} \!:^ { \color{black}{r} } \! \ottnt{A_{{\mathrm{2}}}} . \ottnt{B}   \ottsym{\{}   \ottkw{inj}_2\,  \ottmv{x}   \ottsym{/}  \ottmv{y}  \ottsym{\}} $ and $  \Delta ,   \ottmv{y} \!\!:\!\!  \ottnt{A_{{\mathrm{1}}}}  \oplus  \ottnt{A_{{\mathrm{2}}}}    \ ;\   \Gamma ,   \ottmv{y} \! :^{ { \color{black}{s} } }\!  \ottnt{A_{{\mathrm{1}}}}  \oplus  \ottnt{A_{{\mathrm{2}}}}     \vdash \ottnt{B} :  \textbf{Type}  $ and $ { \color{black}{1} }   \leq  { \color{black}{r} }$. Further, $ \ottnt{H}  \vdash  \Delta ;  \Gamma_{{\mathrm{0}}}  \ottsym{+}   { \color{black}{q} }   \cdot   \ottsym{(}    { \color{black}{r} }   \cdot   \Gamma_{{\mathrm{1}}}    \ottsym{+}  \Gamma_{{\mathrm{2}}}  \ottsym{)}  $. Now, there are two cases to consider depending on whether $a$ is a value or not.

\begin{itemize}
\item $a$ is not a value.

By inductive hypothesis, we get $ [  \ottnt{H}  ]\,  \ottnt{a}  \Rightarrow_{   \ottnt{S}  \, \cup     \,\text{fv}\,  \ottnt{b_{{\mathrm{1}}}}   \, \cup     \,\text{fv}\,  \ottnt{b_{{\mathrm{2}}}}   \, \cup   \{  \ottmv{y}  \}        }^{  { \color{black}{q} }  \cdot  { \color{black}{r} }  } [  \ottnt{H'} \, ;\,  \mathbf{u} \, ;\,  \Gamma_{{\mathrm{4}}}  ]\,  \ottnt{a'}  $ and $  \Delta ,   \lfloor \Gamma_{{\mathrm{4}}} \rfloor   \ ;\  \Gamma'_{{\mathrm{1}}}  \vdash \ottnt{a'} :  \ottnt{A_{{\mathrm{1}}}}  \oplus  \ottnt{A_{{\mathrm{2}}}}  $ and $ \ottnt{H'}  \vdash   \Delta ,   \lfloor \Gamma_{{\mathrm{4}}} \rfloor   ;  \ottsym{(}   \ottsym{(}  \Gamma_{{\mathrm{0}}}  \ottsym{+}   { \color{black}{q} }   \cdot   \Gamma_{{\mathrm{2}}}   \ottsym{)} ,    { \color{black}{0} }    \cdot   \Gamma_{{\mathrm{4}}}    \ottsym{)}  \ottsym{+}   \ottsym{(}  { \color{black}{q} }  \cdot  { \color{black}{r} }  \ottsym{)}   \cdot   \Gamma'_{{\mathrm{1}}}  $. This gives us, $ [  \ottnt{H}  ]\,   \ottkw{case}_ { \color{black}{r} } \,  \ottnt{a} \, \ottkw{of}\,  \ottnt{b_{{\mathrm{1}}}}  ;   \ottnt{b_{{\mathrm{2}}}}   \Rightarrow_{ \ottnt{S} }^{ { \color{black}{q} } } [  \ottnt{H'} \, ;\,  \mathbf{u} \, ;\,  \Gamma_{{\mathrm{4}}}  ]\,   \ottkw{case}_ { \color{black}{r} } \,  \ottnt{a'} \, \ottkw{of}\,  \ottnt{b_{{\mathrm{1}}}}  ;   \ottnt{b_{{\mathrm{2}}}}  $. By weakening and using \rref{T-SumElim} thereafter, we get, $  \Delta ,   \lfloor \Gamma_{{\mathrm{4}}} \rfloor   \ ;\   { \color{black}{r} }   \cdot   \Gamma'_{{\mathrm{1}}}   \ottsym{+}  \ottsym{(}   \Gamma_{{\mathrm{2}}} ,    { \color{black}{0} }    \cdot   \Gamma_{{\mathrm{4}}}    \ottsym{)}  \vdash  \ottkw{case}_ { \color{black}{r} } \,  \ottnt{a'} \, \ottkw{of}\,  \ottnt{b_{{\mathrm{1}}}}  ;   \ottnt{b_{{\mathrm{2}}}}  : \ottnt{B}  \ottsym{\{}  \ottnt{a'}  \ottsym{/}  \ottmv{y}  \ottsym{\}} $. By following the argument presented earlier, $ \ottnt{B}  \ottsym{\{}  \ottnt{a}  \ottsym{/}  \ottmv{y}  \ottsym{\}}  \{  \Delta  \}   \equiv   \ottnt{B}  \ottsym{\{}  \ottnt{a'}  \ottsym{/}  \ottmv{y}  \ottsym{\}}  \{   \Delta ,   \lfloor \Gamma_{{\mathrm{4}}} \rfloor    \} $; as such, by \rref{T-Conv}, $  \Delta ,   \lfloor \Gamma_{{\mathrm{4}}} \rfloor   \ ;\   { \color{black}{r} }   \cdot   \Gamma'_{{\mathrm{1}}}   \ottsym{+}  \ottsym{(}   \Gamma_{{\mathrm{2}}} ,    { \color{black}{0} }    \cdot   \Gamma_{{\mathrm{4}}}    \ottsym{)}  \vdash  \ottkw{case}_ { \color{black}{r} } \,  \ottnt{a'} \, \ottkw{of}\,  \ottnt{b_{{\mathrm{1}}}}  ;   \ottnt{b_{{\mathrm{2}}}}  : \ottnt{B}  \ottsym{\{}  \ottnt{a}  \ottsym{/}  \ottmv{y}  \ottsym{\}} $.\\

By inductive hypothesis, $  (  \ottsym{(}  { \color{black}{q} }  \cdot  { \color{black}{r} }  \ottsym{)}   \cdot     \ncoverline{ \Gamma'_{{\mathrm{1}}} }   +  \mathbf{u}   )   +   \ncoverline{ \Gamma_{{\mathrm{4}}} }   \times \langle H' \rangle \leq   \ottsym{(}  { \color{black}{q} }  \cdot  { \color{black}{r} }  \ottsym{)}   \cdot    (   \ncoverline{ \Gamma_{{\mathrm{1}}} }   \mathop{\diamond}   \mathbf{0}   )    +  \mathbf{u}  \times \langle H' \rangle +   \mathbf{0}   \mathop{\diamond}   \ncoverline{ \Gamma_{{\mathrm{4}}} }  $. From this, we have, $  (  { \color{black}{q} }   \cdot     \ncoverline{ \ottsym{(}   { \color{black}{r} }   \cdot   \Gamma'_{{\mathrm{1}}}   \ottsym{+}  \ottsym{(}   \Gamma_{{\mathrm{2}}} ,    { \color{black}{0} }    \cdot   \Gamma_{{\mathrm{4}}}    \ottsym{)}  \ottsym{)} }   +  \mathbf{u}   )   +   \ncoverline{ \Gamma_{{\mathrm{4}}} }   \times \langle H' \rangle \leq   { \color{black}{q} }   \cdot    (   \ncoverline{ \ottsym{(}   { \color{black}{r} }   \cdot   \Gamma_{{\mathrm{1}}}   \ottsym{+}  \Gamma_{{\mathrm{2}}}  \ottsym{)} }   \mathop{\diamond}   \mathbf{0}   )    +  \mathbf{u}  \times \langle H' \rangle +   \mathbf{0}   \mathop{\diamond}   \ncoverline{ \Gamma_{{\mathrm{4}}} }  $.  
  
\item $a$ is a value.

Since $a$ has type $ \ottnt{A_{{\mathrm{1}}}}  \oplus  \ottnt{A_{{\mathrm{2}}}} $, so $\ottnt{a}  \ottsym{=}   \ottkw{inj}_1\,  \ottnt{a_{{\mathrm{1}}}} $ or $\ottnt{a}  \ottsym{=}   \ottkw{inj}_2\,  \ottnt{a_{{\mathrm{2}}}} $.\\ Let $\ottnt{a}  \ottsym{=}   \ottkw{inj}_1\,  \ottnt{a_{{\mathrm{1}}}} $. Now, $ [  \ottnt{H}  ]\,   \ottkw{case}_ { \color{black}{r} } \,  \ottsym{(}   \ottkw{inj}_1\,  \ottnt{a_{{\mathrm{1}}}}   \ottsym{)} \, \ottkw{of}\,  \ottnt{b_{{\mathrm{1}}}}  ;   \ottnt{b_{{\mathrm{2}}}}   \Rightarrow_{ \ottnt{S} }^{ { \color{black}{q} } } [  \ottnt{H} \, ;\,   \mathbf{0}  \, ;\,  \varnothing  ]\,  \ottnt{b_{{\mathrm{1}}}} \, \ottnt{a_{{\mathrm{1}}}} $. By inverting the typing judgement, we have $ \Delta \ ;\  \Gamma_{{\mathrm{1}}}  \vdash \ottnt{a_{{\mathrm{1}}}} : \ottnt{A_{{\mathrm{1}}}} $. By \rref{T-app}, $ \Delta \ ;\  \Gamma_{{\mathrm{2}}}  \ottsym{+}   { \color{black}{r} }   \cdot   \Gamma_{{\mathrm{1}}}   \vdash \ottnt{b_{{\mathrm{1}}}} \, \ottnt{a_{{\mathrm{1}}}} : \ottnt{B}  \ottsym{\{}   \ottkw{inj}_1\,  \ottnt{a_{{\mathrm{1}}}}   \ottsym{/}  \ottmv{y}  \ottsym{\}} $. The fourth clause follows by reflexivity. The case where $\ottnt{a}  \ottsym{=}   \ottkw{inj}_2\,  \ottnt{a_{{\mathrm{2}}}} $ follows similarly.
 
\end{itemize}
  
\item \textbf{\rref{T-SigmaElim}}

Let $\ottnt{A}  \ottsym{=}   \Sigma  \ottmv{x} \!\!:^ { \color{black}{r} } \!\! \ottnt{A_{{\mathrm{1}}}} . \ottnt{A_{{\mathrm{2}}}} $ and $ \Delta \ ;\  \Gamma_{{\mathrm{1}}}  \ottsym{+}  \Gamma_{{\mathrm{2}}}  \vdash \ottkw{let} \, \ottsym{(}  \ottmv{x}  \ottsym{,}  \ottmv{y}  \ottsym{)}  \ottsym{=}  \ottnt{a} \, \mathsf{in} \, \ottnt{b} : \ottnt{B}  \ottsym{\{}  \ottnt{a}  \ottsym{/}  \ottmv{z}  \ottsym{\}} $ where $ \Delta \ ;\  \Gamma_{{\mathrm{1}}}  \vdash \ottnt{a} : \ottnt{A} $ and $   \Delta ,   \ottmv{x} \!\!:\!\! \ottnt{A_{{\mathrm{1}}}}   ,   \ottmv{y} \!\!:\!\! \ottnt{A_{{\mathrm{2}}}}   \ ;\    \Gamma_{{\mathrm{2}}} ,   \ottmv{x} \! :^{ { \color{black}{r} } }\! \ottnt{A_{{\mathrm{1}}}}   ,   \ottmv{y} \! :^{  { \color{black}{1} }  }\! \ottnt{A_{{\mathrm{2}}}}    \vdash \ottnt{b} : \ottnt{B}  \ottsym{\{}  \ottsym{(}  \ottmv{x}  \ottsym{,}  \ottmv{y}  \ottsym{)}  \ottsym{/}  \ottmv{z}  \ottsym{\}} $ and $  \Delta ,   \ottmv{z} \!\!:\!\! \ottnt{A}   \ ;\   \Gamma ,   \ottmv{z} \! :^{ { \color{black}{s} } }\! \ottnt{A}    \vdash \ottnt{B} :  \textbf{Type}  $. Further, $ \ottnt{H}  \vdash  \Delta ;  \Gamma_{{\mathrm{0}}}  \ottsym{+}   { \color{black}{q} }   \cdot   \ottsym{(}  \Gamma_{{\mathrm{1}}}  \ottsym{+}  \Gamma_{{\mathrm{2}}}  \ottsym{)}  $. Now, there are two cases to consider depending on whether $a$ is a value or not.

\begin{itemize}
\item $a$ is not a value.

By inductive hypothesis, $ [  \ottnt{H}  ]\,  \ottnt{a}  \Rightarrow_{   \ottnt{S}  \, \cup     \,\text{fv}\,  \ottnt{b}   \, \cup       \{  \ottmv{x}  \}   \, \cup   \{  \ottmv{y}  \}     \, \cup   \{  \ottmv{z}  \}        }^{ { \color{black}{q} } } [  \ottnt{H'} \, ;\,  \mathbf{u} \, ;\,  \Gamma_{{\mathrm{4}}}  ]\,  \ottnt{a'} $ and $  \Delta ,   \lfloor \Gamma_{{\mathrm{4}}} \rfloor   \ ;\  \Gamma'_{{\mathrm{1}}}  \vdash \ottnt{a'} : \ottnt{A} $. Therefore, $ [  \ottnt{H}  ]\,  \ottkw{let} \, \ottsym{(}  \ottmv{x}  \ottsym{,}  \ottmv{y}  \ottsym{)}  \ottsym{=}  \ottnt{a} \, \mathsf{in} \, \ottnt{b}  \Rightarrow_{ \ottnt{S} }^{ { \color{black}{q} } } [  \ottnt{H'} \, ;\,  \mathbf{u} \, ;\,  \Gamma_{{\mathrm{4}}}  ]\,  \ottkw{let} \, \ottsym{(}  \ottmv{x}  \ottsym{,}  \ottmv{y}  \ottsym{)}  \ottsym{=}  \ottnt{a'} \, \mathsf{in} \, \ottnt{b} $. By weakening and \rref{SigmaElim}, we have, $  \Delta ,   \lfloor \Gamma_{{\mathrm{4}}} \rfloor   \ ;\  \Gamma'_{{\mathrm{1}}}  \ottsym{+}  \ottsym{(}  \Gamma_{{\mathrm{2}}}  \ottsym{+}    { \color{black}{0} }    \cdot   \Gamma_{{\mathrm{4}}}   \ottsym{)}  \vdash \ottkw{let} \, \ottsym{(}  \ottmv{x}  \ottsym{,}  \ottmv{y}  \ottsym{)}  \ottsym{=}  \ottnt{a'} \, \mathsf{in} \, \ottnt{b} : \ottnt{B}  \ottsym{\{}  \ottnt{a'}  \ottsym{/}  \ottmv{z}  \ottsym{\}} $. Using an argument presented before, we get $ \ottnt{B}  \ottsym{\{}  \ottnt{a}  \ottsym{/}  \ottmv{z}  \ottsym{\}}  \{  \Delta  \}   \equiv   \ottnt{B}  \ottsym{\{}  \ottnt{a'}  \ottsym{/}  \ottmv{z}  \ottsym{\}}  \{   \Delta ,   \lfloor \Gamma_{{\mathrm{4}}} \rfloor    \} $; as such, by \rref{T-conv}, $  \Delta ,   \lfloor \Gamma_{{\mathrm{4}}} \rfloor   \ ;\  \Gamma'_{{\mathrm{1}}}  \ottsym{+}  \ottsym{(}  \Gamma_{{\mathrm{2}}}  \ottsym{+}    { \color{black}{0} }    \cdot   \Gamma_{{\mathrm{4}}}   \ottsym{)}  \vdash \ottkw{let} \, \ottsym{(}  \ottmv{x}  \ottsym{,}  \ottmv{y}  \ottsym{)}  \ottsym{=}  \ottnt{a'} \, \mathsf{in} \, \ottnt{b} : \ottnt{B}  \ottsym{\{}  \ottnt{a}  \ottsym{/}  \ottmv{z}  \ottsym{\}} $. The other clauses follow from the inductive hypothesis.

\item $a$ is a value.

Since $a$ has $\Sigma$-type, $\ottnt{a}  \ottsym{=}  \ottsym{(}  \ottnt{a_{{\mathrm{1}}}}  \ottsym{,}  \ottnt{a_{{\mathrm{2}}}}  \ottsym{)}$ where $ \Delta \ ;\  \Gamma_{{\mathrm{11}}}  \vdash \ottnt{a_{{\mathrm{1}}}} : \ottnt{A_{{\mathrm{1}}}} $ and $ \Delta \ ;\  \Gamma_{{\mathrm{12}}}  \vdash \ottnt{a_{{\mathrm{2}}}} : \ottnt{A_{{\mathrm{2}}}}  \ottsym{\{}  \ottnt{a_{{\mathrm{1}}}}  \ottsym{/}  \ottmv{x}  \ottsym{\}} $ and $\Gamma_{{\mathrm{1}}}  \ottsym{=}   { \color{black}{r} }   \cdot   \Gamma_{{\mathrm{11}}}   \ottsym{+}  \Gamma_{{\mathrm{12}}}$. Assuming $x'$ and $y'$ are fresh enough, $ [  \ottnt{H}  ]\,  \ottkw{let} \, \ottsym{(}  \ottmv{x}  \ottsym{,}  \ottmv{y}  \ottsym{)}  \ottsym{=}  \ottsym{(}  \ottnt{a_{{\mathrm{1}}}}  \ottsym{,}  \ottnt{a_{{\mathrm{2}}}}  \ottsym{)} \, \mathsf{in} \, \ottnt{b}  \Rightarrow_{ \ottnt{S} }^{ { \color{black}{q} } } [   \ottnt{H}  ,     \ottmv{x'}  \overset{ \ottsym{(}  { \color{black}{q} }  \cdot  { \color{black}{r} }  \ottsym{)} }{\mapsto}  \ottnt{a_{{\mathrm{1}}}}  :  \ottnt{A_{{\mathrm{1}}}}   ,   \ottmv{y'}  \overset{ { \color{black}{q} } }{\mapsto}  \ottnt{a_{{\mathrm{2}}}}  :  \ottnt{A_{{\mathrm{2}}}}  \ottsym{\{}  \ottmv{x'}  \ottsym{/}  \ottmv{x}  \ottsym{\}}     \, ;\,   \mathbf{0}  \, ;\,    \ottmv{x'}  \! = \!  \ottnt{a_{{\mathrm{1}}}}  \! :^{ \ottsym{(}  { \color{black}{q} }  \cdot  { \color{black}{r} }  \ottsym{)} } \!  \ottnt{A_{{\mathrm{1}}}}  ,   \ottmv{y'}  \! = \!  \ottnt{a_{{\mathrm{2}}}}  \! :^{ { \color{black}{q} } } \!  \ottnt{A_{{\mathrm{2}}}}    \ottsym{\{}  \ottmv{x'}  \ottsym{/}  \ottmv{x}  \ottsym{\}}  ]\,  \ottnt{b}  \ottsym{\{}  \ottmv{x'}  \ottsym{/}  \ottmv{x}  \ottsym{\}}  \ottsym{\{}  \ottmv{y'}  \ottsym{/}  \ottmv{y}  \ottsym{\}} $.\footnote{We omit some extraneous components to keep things clean.} \\

Since $ \ottnt{H}  \vdash  \Delta ;  \Gamma_{{\mathrm{0}}}  \ottsym{+}   { \color{black}{q} }   \cdot   \ottsym{(}     { \color{black}{r} }   \cdot   \Gamma_{{\mathrm{11}}}    \ottsym{+}  \Gamma_{{\mathrm{12}}}   \ottsym{+}  \Gamma_{{\mathrm{2}}}  \ottsym{)}  $, so\\ $ \ottnt{H}  ,     \ottmv{x'}  \overset{ \ottsym{(}  { \color{black}{q} }  \cdot  { \color{black}{r} }  \ottsym{)} }{\mapsto} { \Gamma_{{\mathrm{11}}} \vdash  \ottnt{a_{{\mathrm{1}}}}  :  \ottnt{A_{{\mathrm{1}}}} }   ,   \ottmv{y'}  \overset{ { \color{black}{q} } }{\mapsto} { \ottsym{(}   \Gamma_{{\mathrm{12}}} ,   \ottmv{x'}  \! = \!  \ottnt{a_{{\mathrm{1}}}}  \! :^{  { \color{black}{0} }  } \!  \ottnt{A_{{\mathrm{1}}}}    \ottsym{)} \vdash  \ottnt{a_{{\mathrm{2}}}}  :  \ottnt{A_{{\mathrm{2}}}}  \ottsym{\{}  \ottmv{x'}  \ottsym{/}  \ottmv{x}  \ottsym{\}} }      \vdash  \\ \ottsym{(}    \Gamma_{{\mathrm{0}}} ,   \ottmv{x'}  \! = \!  \ottnt{a_{{\mathrm{1}}}}  \! :^{  { \color{black}{0} }  } \!  \ottnt{A_{{\mathrm{1}}}}   ,   \ottmv{y'}  \! = \!  \ottnt{a_{{\mathrm{2}}}}  \! :^{  { \color{black}{0} }  } \!  \ottnt{A_{{\mathrm{2}}}}    \ottsym{\{}  \ottmv{x'}  \ottsym{/}  \ottmv{x}  \ottsym{\}}  \ottsym{)}  \ottsym{+}   { \color{black}{q} }   \cdot   \ottsym{(}   \Gamma_{{\mathrm{2}}} ,     \ottmv{x'}  \! = \!  \ottnt{a_{{\mathrm{1}}}}  \! :^{ { \color{black}{r} } } \!  \ottnt{A_{{\mathrm{1}}}}  ,   \ottmv{y'}  \! = \!  \ottnt{a_{{\mathrm{2}}}}  \! :^{  { \color{black}{1} }  } \!  \ottnt{A_{{\mathrm{2}}}}  \ottsym{\{}  \ottmv{x'}  \ottsym{/}  \ottmv{x}  \ottsym{\}}      \ottsym{)} $. \\

By \ref{InsertEq}, $  \Gamma_{{\mathrm{2}}} ,     \ottmv{x'}  \! = \!  \ottnt{a_{{\mathrm{1}}}}  \! :^{ { \color{black}{r} } } \!  \ottnt{A_{{\mathrm{1}}}}  ,   \ottmv{y'}  \! = \!  \ottnt{a_{{\mathrm{2}}}}  \! :^{  { \color{black}{1} }  } \!  \ottnt{A_{{\mathrm{2}}}}  \ottsym{\{}  \ottmv{x'}  \ottsym{/}  \ottmv{x}  \ottsym{\}}      \vdash \ottnt{b}  \ottsym{\{}  \ottmv{x'}  \ottsym{/}  \ottmv{x}  \ottsym{\}}  \ottsym{\{}  \ottmv{y'}  \ottsym{/}  \ottmv{y}  \ottsym{\}} : \ottnt{B}  \ottsym{\{}  \ottsym{(}  \ottmv{x'}  \ottsym{,}  \ottmv{y'}  \ottsym{)}  \ottsym{/}  \ottmv{z}  \ottsym{\}} $. But then, by conversion, $  \Gamma_{{\mathrm{2}}} ,     \ottmv{x'}  \! = \!  \ottnt{a_{{\mathrm{1}}}}  \! :^{ { \color{black}{r} } } \!  \ottnt{A_{{\mathrm{1}}}}  ,   \ottmv{y'}  \! = \!  \ottnt{a_{{\mathrm{2}}}}  \! :^{  { \color{black}{1} }  } \!  \ottnt{A_{{\mathrm{2}}}}  \ottsym{\{}  \ottmv{x'}  \ottsym{/}  \ottmv{x}  \ottsym{\}}      \vdash \ottnt{b}  \ottsym{\{}  \ottmv{x'}  \ottsym{/}  \ottmv{x}  \ottsym{\}}  \ottsym{\{}  \ottmv{y'}  \ottsym{/}  \ottmv{y}  \ottsym{\}} : \ottnt{B}  \ottsym{\{}  \ottsym{(}  \ottnt{a_{{\mathrm{1}}}}  \ottsym{,}  \ottnt{a_{{\mathrm{2}}}}  \ottsym{)}  \ottsym{/}  \ottmv{z}  \ottsym{\}} $.

The fourth clause: $  { \color{black}{q} }   \cdot    (   \ncoverline{ \Gamma_{{\mathrm{2}}} }   \mathop{\diamond}   (   { \color{black}{r} }   \mathop{\diamond}    { \color{black}{1} }    )   )    +     \mathbf{0}   \mathop{\diamond}     \ottsym{(}  { \color{black}{q} }  \cdot  { \color{black}{r} }  \ottsym{)}   \mathop{\diamond}   { \color{black}{q} }       \times \big( \begin{smallmatrix} \langle H \rangle &  \mathbf{0} ^\intercal &  \mathbf{0} ^\intercal \\  \ncoverline{ \Gamma_{{\mathrm{11}}} }  & 0 & 0 \\  \ncoverline{ \Gamma_{{\mathrm{12}}} }  & 0 & 0\end{smallmatrix} \big) \leq   { \color{black}{q} }   \cdot    (     { \color{black}{r} }   \cdot    \ncoverline{ \Gamma_{{\mathrm{11}}} }    +     \ncoverline{ \Gamma_{{\mathrm{12}}} }   +   \ncoverline{ \Gamma_{{\mathrm{2}}} }       \mathop{\diamond}      { \color{black}{0} }    \mathop{\diamond}    { \color{black}{0} }      )    +   (   \mathbf{0}   \mathop{\diamond}   (   \ottsym{(}  { \color{black}{q} }  \cdot  { \color{black}{r} }  \ottsym{)}   \mathop{\diamond}   { \color{black}{q} }   )   )  $ follows by reflexivity.

\end{itemize}

In all the other cases, $a$ is a value.
\else

\fi

\end{itemize}
\end{proof}

\begin{theorem}[Soundness]
If $ \ottnt{H}  \vdash  \Delta ;  \Gamma $ and $ \Delta \ ;\  \Gamma  \vdash \ottnt{a} : \ottnt{A} $ and $S \supseteq  \mathsf{dom} \,  \Delta $, then either $a$ is a value or there exists $\Gamma'$, $\ottnt{H'}$, $\mathbf{u}'$, $\Gamma_{{\mathrm{4}}}, \ottnt{A'}$ such that:
\begin{itemize}
\item $ [  \ottnt{H}  ]\,  \ottnt{a}  \! :  \ottnt{A}  \Rightarrow_{ \ottnt{S} } [  \ottnt{H'} \, ;\,  \mathbf{u}' \, ;\,  \Gamma_{{\mathrm{4}}}  ]\,  \ottnt{a'}  \! :  \ottnt{A'} $
\item $ \ottnt{H'}  \vdash   \Delta ,   \lfloor \Gamma_{{\mathrm{4}}} \rfloor   ;  \Gamma' $
\item $  \Delta ,   \lfloor \Gamma_{{\mathrm{4}}} \rfloor   \ ;\  \Gamma'  \vdash \ottnt{a'} : \ottnt{A'} $
\item $  \ncoverline{ \Gamma' }   +  \mathbf{u}'  +   \mathbf{0}   \mathop{\diamond}   \ncoverline{ \Gamma_{{\mathrm{4}}} }   \times \langle H'\rangle \leq   \ncoverline{ \Gamma }   \mathop{\diamond}   \mathbf{0}   + \mathbf{u}' \times \langle H' \rangle +   \mathbf{0}   \mathop{\diamond}   \ncoverline{ \Gamma_{{\mathrm{4}}} }  $
\end{itemize}
\end{theorem}

\begin{proof}
Follows from \ref{inv} with $q = 1$ and $\Gamma_{{\mathrm{0}}} =   { \color{black}{0} }    \cdot   \Gamma $.
\end{proof}

The soundness theorem is similar in spirit to theorems showing correctness of usage in graded type systems  via operational methods. But this theorem can be proved by simple induction on the typing derivation; it does not require much extra machinery over and above the reduction relation, unlike the proof of soundness in \citet{Brunel:2014} which requires a realizability model on top of the reduction relation. In this regard, our soundness theorem is more in line with the modality preservation theorem in \citet{abel:icfp2020}.  

We can use the soundness theorem to prove the usual preservation and progress lemmas. The proofs are similar to the corresponding ones for the simple version (\ref{spreservation} and \ref{sprogress}). This means that the ordinary semantics is sound with respect to a resource-aware semantics.

Below, we show an example application of this theorem. Other similar examples can be worked out using the lemmas from Section \ref{sec:applications}.  
\begin{example}
\label{ex:depex}
Consider any security lattice $Q$, as described in Section \ref{sec:semiring-examples}. Let ${ \color{black}{s} }$ be any element such that $1 \nleq s$. Also, let $ \varnothing  \vdash \ottnt{A} :  {}^{  { \color{black}{1} }  }  \textbf{Type}  \rightarrow   \textbf{Type}   $ such that $\ottnt{A} \, \ottkw{Unit}  \ottsym{=}  \ottkw{Int}$. Now, let $   \ottmv{x} \! :^{  { \color{black}{1} }  }\!  \textbf{Type}   ,   \ottmv{y} \! :^{ { \color{black}{q} } }\! \ottnt{A} \, \ottmv{x}    \vdash \ottnt{B} :  \textbf{Type}  $ for some $q \in Q$.

In empty context, consider a term $f$ of type $ \Pi  \ottmv{x} \!:^ { \color{black}{r} } \!  \textbf{Type}  .  \Pi  \ottmv{y} \!:^ { \color{black}{s} } \! \ottnt{A} \, \ottmv{x} . \ottnt{B}  $ for some $r \in Q$. (Note that neither $r$ and $1$ nor $s$ and $q$ need to be equal, as explained in irrelevant quantification.) Then, we can show that both $\ottmv{f} \, \ottkw{Unit} \,  0 $ and $\ottmv{f} \, \ottkw{Unit} \,  1 $ either diverge or produce equal values. If $\ottmv{f} \, \ottkw{Unit}$ diverges, then both diverge. Otherwise, $ [  \varnothing  ]\,  \ottmv{f} \, \ottkw{Unit}  \Rightarrow\!\!\!\!\!\Rightarrow  [  \ottnt{H}  ] \,   \lambda \ottmv{y} \!:^ { \color{black}{s} } \! \ottnt{C} . \ottnt{b}  $ where $ \ottnt{C}  \{  \ottnt{H}  \}   \equiv  \ottkw{Int}$. Now, $ [  \ottnt{H}  ]\,  \ottsym{(}   \lambda \ottmv{y} \!:^ { \color{black}{s} } \! \ottnt{C} . \ottnt{b}   \ottsym{)} \,  0   \Rightarrow [   \ottnt{H}  ,   \ottmv{y}  \overset{ { \color{black}{s} } }{\mapsto}   0     ] \,  \ottnt{b} $ and $  [  \ottnt{H}  ]\,  \ottsym{(}   \lambda \ottmv{y} \!:^ { \color{black}{s} } \! \ottnt{C} . \ottnt{b}   \ottsym{)} \,  1   \Rightarrow [   \ottnt{H}  ,   \ottmv{y}  \overset{ { \color{black}{s} } }{\mapsto}   1     ] \,  \ottnt{b} $. By \ref{lemma:s-nonint}, we know that both $ ((  \ottnt{H}  ,   \ottmv{y}  \overset{ { \color{black}{s} } }{\mapsto}   0    ), \ottnt{b} ) $ and $ ((  \ottnt{H}  ,   \ottmv{y}  \overset{ { \color{black}{s} } }{\mapsto}   1    ), \ottnt{b} ) $ either diverge or reduce to the same value.
\end{example}  

\section{Discussion}

\subsection{Definitional-Equivalence  and Irrelevance}
\label{defeq}

The terms ``irrelevance'' and ``irrelevant quantification'' have multiple
meanings in the literature. Our primary focus is on erasability, the ability 
to quantify over arguments that need not be present at
runtime. However, this terminology often includes compile-time irrelevance, or
the blindness of type equality to such erasable parts of terms.  These terms
are also related to, but not the same as, ``parametricity'' or ``parametric
quantification'', which characterizes functions that map equivalent arguments
to equivalent results.

One difference between our formulation and a more traditional dependently-typed
calculus is that the conversion rule (\rref{T-conv}) is specified in terms of
an abstract equivalence relation on terms, written $\ottnt{A}  \equiv  \ottnt{B}$.  Our proofs
about this system work for any relation that satisfies the following
properties.  \scw{add a citation to POPL10 paper?}

\begin{definition}
\label{defeql}
We say that the relation $\ottnt{A}  \equiv  \ottnt{B}$ is \emph{sound} if it:
\begin{enumerate}
\item \emph{is equivalence relation},
\item \emph{contains the small step relation}, in other words, if $ \ottnt{a}  \leadsto  \ottnt{a'} $ then $\ottnt{a}  \equiv  \ottnt{a'}$,
\item \emph{is closed under substitution}, in other words, if $\ottnt{a_{{\mathrm{1}}}}  \equiv  \ottnt{a_{{\mathrm{2}}}}$ then $\ottnt{b}  \ottsym{\{}  \ottnt{a_{{\mathrm{1}}}}  \ottsym{/}  \ottmv{x}  \ottsym{\}}  \equiv  \ottnt{b}  \ottsym{\{}  \ottnt{a_{{\mathrm{2}}}}  \ottsym{/}  \ottmv{x}  \ottsym{\}}$ and $\ottnt{a_{{\mathrm{1}}}}  \ottsym{\{}  \ottnt{b}  \ottsym{/}  \ottmv{x}  \ottsym{\}}  \equiv  \ottnt{a_{{\mathrm{2}}}}  \ottsym{\{}  \ottnt{b}  \ottsym{/}  \ottmv{x}  \ottsym{\}}$, 
\item \emph{is injective for type constructors}, for example, if
  $ \Pi  \ottmv{x} \!:^ { \color{black}{q} }_{{\mathrm{1}}} \! \ottnt{A_{{\mathrm{1}}}} . \ottnt{B_{{\mathrm{1}}}}   \equiv   \Pi  \ottmv{x} \!:^ { \color{black}{q} }_{{\mathrm{2}}} \! \ottnt{A_{{\mathrm{2}}}} . \ottnt{B_{{\mathrm{2}}}} $ then ${ \color{black}{q} }_{{\mathrm{1}}}  \ottsym{=}  { \color{black}{q} }_{{\mathrm{2}}}$ and $\ottnt{A_{{\mathrm{1}}}}  \equiv  \ottnt{A_{{\mathrm{2}}}}$
  and $\ottnt{B_{{\mathrm{1}}}}  \equiv  \ottnt{B_{{\mathrm{2}}}}$ (and similar for $ \Sigma  \ottmv{x} \!\!:^ { \color{black}{r} } \!\! \ottnt{A} . \ottnt{B} $ and $ \ottnt{A}  \oplus  \ottnt{B} $), 
\item and \emph{is consistent}, in other words, if $\ottnt{A}  \equiv  \ottnt{B}$ and both are
  values, then they have the same head form.
\end{enumerate}
\end{definition}

The standard $\beta$-conversion relation, defined as the reflexive, symmetric,
transitive and congruent closure of the step relation, is a sound relation.

However, $\beta$-conversion is not the only relation that would
work. Dependent type systems with irrelevance sometimes erase irrelevant parts
of terms before comparing them up to
$\beta$-equivalence \hspace*{6pt} ~\cite{barras:icc-star}. Alternatively, a typed definition
of equivalence, might use the total relation when equating irrelevant
components ~\cite{pfenning:2001}.  In future work, we hope to show that any
sound definition of equivalence can be coarsened by ignoring irrelevant
components in terms during comparison. We conjecture that such a relation
would also satisfy the properties above. In particular, our results from
Section~\ref{sec:applications} tell us that such coarsening of the equivalence relation is
consistent with evaluation, and therefore contains the step relation.

\subsection{Connection to Haskell}
\label{sec:haskell}

%\paragraph{Relationship to today's Haskell}
The current design of linear types in GHC/Haskell
is essentially an instance of the type system described in this paper,
one that uses the linearity semiring.
Haskell users can mark arguments with grades $1$ or $\omega$, but a
grade of $0$ is sometimes needed internally. Haskell's kind system supports irrelevance, but
not linearity, so the two features do not yet interact. It is with dependent types that
we need a uniform treatment which we can achieve through this graded type system. The current Haskell structure will be able to migrate
to a graded type system with little, if any, backward compatibility trouble for users.

One feature of Haskell's linear types does cause a small wrinkle, though: Haskell supports
\emph{multiplicity polymorphism}. An easy example is the type of \texttt{map}, which is
\texttt{forall m a b. (a \%m-> b) -> [a] \%m-> [b]}. We see that the function argument to
\texttt{map} can be either linear or unrestricted, and that this choice affects whether
the input list is restricted. We cannot support quantity polymorphism in our type theory,
as quantifying over whether or not an argument is relevant would mean that we could no longer
compile a quantity-polymorphic function: would the compiled function take the argument in a
register or not? The solution is to tweak the meaning of quantity polymorphism slightly: instead
of quantifying over \emph{all} possible quantities, we would be polymorphic only over quantities
$q$ such that $ { \color{black}{1} }   \leq  { \color{black}{q} }$. That is, we would quantify over only relevant quantities. This
reinterpretation of multiplicity polymorphism avoids the mentioned trouble with static compilation.
Furthermore, we see no difficulty in extending our graded type system with this kind
of quantity polymorphism; in the linear Haskell work, multiplicity polymorphism is nicely
straightforward, and we expect the same to be true here, too.

Commentary on the practicalities of type checking Haskell based on \Langname{} appears
in \auxref{app:type-checking}.

\subsection{Comparison with Quantitative Type Theory}
\label{sec:qtt-comparison}

Quantitative Type Theory QTT~\cite{McBride:2016,atkey} uses
elements of a resource semiring to track the usage of variables in a dependent type system. 
This system has a typing judgement of the form:
$x_1 :^{\rho_1} A_1, x_2 :^{\rho_2} A_2, \ldots, x_n :^{\rho_n} A_n \vdash a
:^{\sigma} A$, where $\rho_i$s and $\sigma$ are elements of a semiring. Roughly
speaking, this judgement means that using $\rho_i$ copies of $x_i$ of type
$A_i$, with $i$ varying from $1$ to $n$, we get $\sigma$ copies of $a$ of type
$A$. 

In QTT, $\sigma$ can be either 0 or 1. When $\sigma$ is 1, the system is similar
to \Langname{}. \Langname{} does not have the $0$-fragment of QTT but this is not a limitation per se: to express the requirement of $0$ copy of $a$, one need only multiply the context by $0$. This approach implies that our system treats types the same as any other irrelevant component of terms.

In contrast, QTT disables resource checking for the $0$-fragment. In other words, in the $0$-fragment, the resource annotations are not meaningful. This difference has both
positive and negative effects on the design of the language.

On the positive side, because linear tensor types are turned into normal
(non-linear) products, QTT can support \emph{strong}-$\Sigma$ types, allowing
projections that violate the usage requirements of their
construction. In contrast, \Langname supports \emph{weak}-$\Sigma$ types only,
with resource-checked pattern matching as the only elimination form.

On the negative side, however, QTT is restricted to semirings that are
\emph{zerosumfree} ($q_1 + q_2 = 0 \Rightarrow q_1 = q_2 = 0$) and
\emph{entire} ($q_1 \cdot q_2 = 0 \Rightarrow q_1 = 0 \vee q_2 = 0$).  (These
properties are necessary to prove substitution.) This limits QTT's
applicability. For example, QTT can not be applied to the class of semirings described in Section~\ref{sec:semiring-examples} that are not entire. On the other hand, our soundness theorem places no constraint on the semiring allowing us to work with such semirings, as lemma \ref{lemma:s-nonint} and example \ref{ex:depex} show.

Furthermore, because QTT ignores usages in the $0$-fragment, its internal logic is limited in its
reasoning about the resource usage of programs. For example, the following
proposition is not provable in QTT:
\[ \forall f:  {}^{  { \color{black}{0} }  } \ottkw{Bool} \rightarrow  \ottkw{Bool} . f \mathsf{True} = f \mathsf{False} \]

The above proposition says that for any constant boolean function, the result of
applying it to $\mathsf{True}$ is the same as the result of applying it to
$\mathsf{False}$. This proposition is not provable in QTT because $f$ ranges
over many functions, including those that examine the argument. In the
0-fragment, the type system cannot prevent a function that uses its argument
to be given a type that says that it does not.
\[ \vdash \lambda x:^0 A. x\ :^0\ \Pi x:^0 A. A \]

\citet{abel:2018} also lists additional ramifications of eliminating
resource checking in types. In particular, he notes that in QTT, it is not
possible to use resource usage to optimize the computation of types during
type checking. In particular, erasing irrelevant terms not only optimizes the
output of a compiler for a dependently-typed language, it is also an
optimization that is useful during compilation, when types are normalized
for comparison.

Finally, we note that \Langname includes case expressions and sub-usaging, while QTT does
not. 
% Furthermore, Atkey~\citet{atkey} proved the soundness of QTT with respect
% to a denotational model, whereas our proofs use a heap-based operational
% model. As a result, our approach extends to non-normalizing languages (such as
% \Langname, which includes $ \textbf{Type} : \textbf{Type} $).

\subsection{Abstract Algebraic Generalization}

Our type system with graded contexts has operations for
addition ($\Gamma_{{\mathrm{1}}}  \ottsym{+}  \Gamma_{{\mathrm{2}}}$) and scalar multiplication ($ { \color{black}{q} }   \cdot   \Gamma $) defined
over an arbitrary partially-ordered semiring. Further, the partial ordering
from the semiring was lifted to contexts. However, we can provide
reasonable alternative definitions for these operations and relations and all our proofs
would still work the same. Here, we lay out what constitutes a reasonable definition. 

Our contexts are an example of a general algebraic structure, called a partially-ordered left semimodule.  Additionally, vectors and matrices of quantities also can also be seen through this abstract mathematical lens. This may help in future extensions and applications of the work presented in this paper.

We follow \citet{golan} in our terminology and definitions here.

%% Macro for notation zero component of module.
\newcommand{\mzero}{\ensuremath{\ncoverline{0}}}

\begin{definition}[Left $Q$-semimodule] Given a semiring $(Q,+,\cdot,0,1)$, a left $Q$-semimodule is a commutative monoid
  $(M,\oplus,\mzero)$ along with a left multiplication function
  $\_\odot\_ : Q \times M \to M$ such that the following properties hold.
\begin{itemize}
\item for $q_1, q_2 \in Q$ and $m \in M$, we have, $(q_1 + q_2) \odot m = q_1 \odot m \oplus q_2 \odot m$
\item for $q \in Q$ and $m_1, m_2 \in M$, we have, $q \odot (m_1 \oplus m_2) = q \odot m_1 \oplus q \odot m_2$
\item for $q_1, q_2 \in Q$ and $m \in M$, we have, $(q_1 \cdot q_2) \odot m = q_1 \odot (q_2 \odot m)$
\item for $m \in M$, we have, $1 \odot m = m$
\item for $q \in Q$ and $m \in M$, we have, $0 \odot m = q \odot \mzero = \mzero$. 
\end{itemize} 
\end{definition}

Graded contexts $\Gamma$ (with the same $\lfloor \Gamma \rfloor$) satisfy this definition, with the operations as defined before. Another example of a semimodule is $Q$ itself, with $\oplus := +$ and $\odot := \cdot$.

%% In fact, any Cartesian power of $Q$ is also a left semimodule. Given a semiring $Q$, let $Q^n$ be the set of $n$-length vectors of elements of $Q$. Then $Q^n$ with $\oplus$ and $\odot$ defined componentwise forms a left $Q$-semimodule. 

Next, let us consider the partial ordering of our contexts. The ordering is basically a lifting of the partial ordering in the semiring. But in general, a partial order on a left semimodule needs to satisfy only the following properties. 

\begin{definition}[Partially-ordered left $Q$-semimodule]
Given a partially-ordered semiring $(Q,\leq)$, a left $Q$-semimodule $M$ is said to be partially-ordered iff there exists a partial order $\leq_M$ on $M$ such that the following properties hold.
\begin{itemize}
\item for $m_1, m_2, m \in M$, if $m_1 \leq_M m_2$, then $m \oplus m_1 \leq_M m \oplus m_2$ 
\item for $q \in Q$ and $m_1, m_2 \in M$, if $m_1 \leq_M m_2$, then $q \odot m_1 \leq_M q \odot m_2$
\item for $q_1, q_2 \in M$ and $m \in M$, if $q_1 \leq q_2$, then $ q_1 \odot m \leq_M q_2 \odot m$.
\end{itemize}
\end{definition}

Note that our ordering of contexts $\Gamma$ satisfy these properties.

We use matrices on several occasions. Matrices can be seen as homomorphisms between semimodules. Given a semiring $Q$, an $m \times n$ matrix with elements drawn from $Q$ is basically a $Q$-homomorphism from $Q^m$ to $Q^n$. 

For $Q$-semimodules $M, N$, a function $\_ \alpha : M \to N$ is said to be a $Q$-homomorphism iff:
\begin{itemize}
\item for $m_1, m_2 \in M$, we have, $(m_1 \oplus m_2)  \alpha = m_1  \alpha \oplus m_2  \alpha$
\item for $q \in Q$ and $m \in M$, we have, $(q \odot m)  \alpha = q \odot (m  \alpha)$.
\end{itemize}

So the matrix $\langle H \rangle$ for a heap $H$ is an endomorphism from $Q^n$ to $Q^n$ where $n = | H |$. Also, an identity matrix is an identity homomorphism.

Next, for natural numbers $i , j , k$ and $Q$-homomorphisms $\_  \alpha : Q^i \to Q^j$ and $\_ \beta : Q^j  \to Q^k$, the composition $\_(\alpha \circ \beta) : Q^i \to Q^k$ can be given by matrix multiplication, $\ncoverline{\alpha} \times \ncoverline{\beta}$. The composition is associative. And it obeys the identity laws.

This makes the set $V_Q =\{ Q^n | n \in \mathbb{N} \}$ with $\text{Hom}(Q^m,Q^n) = \mathcal{M}_{m,n}(Q)$, the set of $m \times n$ matrices over $Q$, a category. We worked in this category. There may be other such categories worth exploring.
\rae{I feel like I'm missing the payoff of this section. It reads like a collection
of observations that we came across while doing this work. But there is no punch line.
Maybe move to the appendix? And leave notes saying that there are connections to abstract
math, viewing graded contexts as semimodules, but with the details deferred to the
appendix.}

\section{Other Related work}
\label{sec:related-work}

\subsection{Heap Semantics for Linear Logic}

Computational and operational interpretations of linear logic have been explored in several works, especially in ~\citet{chirimar}, ~\citet{turner}. In ~\citet{turner}, the authors provide a heap-based operational interpretation of linear logic. They show that a call-by-name calculus enjoys the single pointer property, meaning a linear resource has exactly one reference while a call-by-need calculus satisfies a weaker version of this property, guaranteeing only the maintenance of a single pointer. This system considers only linear and unrestricted resources. We generalize this operational interpretation of linear logic to a graded type system by allowing resources to be drawn from an arbitrary semiring. We derive a quantitative version of the single pointer property in Section \ref{sec:applications}. We can develop a quantitative version of the weak single pointer property for call-by-need reduction but for this, we need to modify the typing rules to allow sharing of resources.

\subsection{Combining Dependent and Linear Types}
Perhaps the earliest work studying the combination of linear and
dependent types was proposed in the form of a categorical model by
\citet{Bonfante:2001} who were interested in characterizing how a
linear dependent type system should be designed.  A year later,
\citet{Cervesato:2002} proposed the Linear Logical Framework (LLF) that
combined non-dependent linear types with dependent types.  This paper
spurred a number of publications, but most relevant are in the line of
work which extend dependent types with \citet{Girard:1992}'s and
\citet{DalLago:2009}'s bounded linear types. For example,
\citet{DalLago:2011}'s $\mathsf{d}l\mathsf{PCF}$, a sound and complete
system for reasoning about evaluation bounds of PCF programs.
\citet{DalLago:2012} also show that $\mathsf{d}l\mathsf{PCF}$ can also
be used to reason about call-by-value execution. 
\citet{Gaboardi:2013} develop a similar system called DFuzz for
analyzing differential privacy of queries involving sensitive information.
In the same vein, \citet{Krishnaswami:2015} show how to combine
non-dependent linear types with dependent types by generalizing
\citet{Benton:1995}'s linear/non-linear logic. But all of these work had some separation between the  linear and non-linear parts of their languages. Quantitative type theory~\cite{McBride:2016,atkey} provided a fresh way to look at this problem by combining the linear and non-linear parts using a resource semiring.

\ifextended
\subsection{Graded Type Systems}

\citet{orchard:2019} introduced a system with a notion of graded
necessity modalities over an arbitrary semiring---here called usage modalities---in a practical
programming language with usage polymorphism and indexed types. However, their system does not have full
dependent types.  They show that usage modalities can be used to
encode a large number of graded coeffects in the style of~\citet{Gaboardi:2016} and \citet{Brunel:2014}. 

Abel and Bernardy~\citep{abel:icfp2020} use a graded type system to
provide an abstract view of modalities. Their type system
is similar in structure to ours, but its features and
requirements differ. It includes usage polymorphism and parametric
polymorphism, but lacks dependent types. Their system is also
strongly normalizing. Furthermore, Abel and Bernardy define a relational
interpretation for their system and use it to derive parametricity
theorems. Due to our
inclusion of the $ \textbf{Type} : \textbf{Type} $ axiom, this parametricity proof technique is 
unavailable to us, so we must use more syntactic methods to reason about our programs. But 
this axiom does not play a major role in our proofs. We conjecture that our
approach to graded dependent types would work equally well in normalizing type
theories.
\fi

\subsection{Irrelevance and Dependent Types}

There are several approaches to adding irrelevant quantification to
dependently-typed languages.  \citet{Miquel:ICC} first added
``implicit'' quantification to a Curry-style version of the extended Calculus
of Constructions. Implicit arguments are those that do not appear free in the body of their abstractions.  In Miquel's system, only the relevant parts of the computation
may be explicit in terms, everything else must be implicit. 
\citet{barras:icc-star} showed how to support decidable type checking by
allowing type annotations and other irrelevant subcomponents to appear in
terms. In this setting, irrelevant arguments must not be free in the erasure
of the body of their abstractions. 
\citet{erasure-pure-type-systems} extended this approach to pure type
systems. More recently, \citet{weirich:icfp17} used these ideas as part of a
proposal for a core language for Dependent Haskell.
We have followed their design in making the usage of irrelevant
variables in the co-domain of $\Pi$-types unrestricted. Specifically, the irrelevance ($-$) tag 
in their language corresponds to the $0$ grade in our language.

\section{Future Work and Conclusions}

Graded type systems are a generic framework for expressing the flow and
usage of resources in programs.  This work provides a new way of
incorporating this framework into dependently-typed languages, with the goal
of supporting both erasure and linearity in the same system.

%This gives us a new way to look at irrelevance and linearity in dependent type theory.
%

We designed a graded dependent type system \Langname and presented a standard substitution-based semantics and a usage-aware heap-based semantics. The standard
semantics does not have the ability to model use of resources.  But the heap-based semantics can track
usage during evaluation of terms. Further, the heap-based reduction relation enforces fair usage of resources. We show that the type system is sound
with respect to this heap semantics. This implies that the type system does a proper static accounting of resource usage.
%% Since our heap semantics is
%% bisimilar to the ordinary semantics, the usage-agnostic way
%% of evaluation is sound with respect to a usage aware semantics.
%%  This means that a graded type
%% system is essentially a tool for static program analysis. 

As always, there is more to explore:
What additional reasoning principles can we get from our heap semantics?
What happens when we add imperative features---like arrays---to our language?
What would a general form of equality up to erasure look like?
What happens when we add multiple modalities, all of them graded, to our language?

The answers to these questions may have theoretical as well as practical implications. Currently,
languages such as Haskell, Rust, Idris, and Agda are experimenting with dependent and
linear types, as well as the more general applications of graded type theories. We hope that this work will provide guidance in these language designs and extensions.

%% Acknowledgments
\begin{acks}
  %% acks environment is optional
  %% contents suppressed with 'anonymous'
  %% Commands \grantsponsor{<sponsorID>}{<name>}{<url>} and
  %% \grantnum[<url>]{<sponsorID>}{<number>} should be used to
  %% acknowledge financial support and will be used by metadata
  %% extraction tools.
  This material is based upon work supported by the
  \grantsponsor{GS100000001}{National Science
    Foundation}{http://dx.doi.org/10.13039/100000001} under Grant
  No.~\grantnum{GS100000001}{1521539}, and Grant
  No.~\grantnum{GS100000001}{1704041}.  Any opinions, findings, and
  conclusions or recommendations expressed in this material are those of the
  author and do not necessarily reflect the views of the National Science
  Foundation.
\end{acks}

\newpage
\bibliography{weirich,eades,qtt}

\newpage

\ifextended
\appendix

\section{Full Judgements}

\subsection{Simple Graded Type System}
\label{app:simple-types}
\rae{Check that I've included the right rules here.}\scw{Checked}

\drules[ST]{$ \Delta \ ;\  \Gamma  \vdash \ottnt{a} : \ottnt{A} $}{Simple graded type system}{Sub,Var,Weak,Unit,UnitE,Lam,App,Box,LetBox,Pair,Spread,InjOne,InjTwo,Case}

\subsection{Operational Semantics for the Simple Graded Type System}
\label{app:simple-opsem}
\rae{Check that I've included the right rules here.}\scw{Checked}

The operational semantics depend on a notion of values:
\[
\begin{array}{llcl}
\textit{values}          & \ottnt{v} & ::= & \ottkw{unit} \alt 
 \lambda \ottmv{x} \!:^ { \color{black}{q} } \! \ottnt{A} . \ottnt{a}  \alt  \ottkw{box} _ { \color{black}{q} } \, \ottnt{a}  \alt \ottsym{(}  \ottnt{a}  \ottsym{,}  \ottnt{b}  \ottsym{)} \alt  \ottkw{inj}_1\,  \ottnt{a}  \alt  \ottkw{inj}_2\,  \ottnt{a} \\
\end{array}
\]

\drules[S]{$ \ottnt{a}  \leadsto  \ottnt{a'} $}{Small-step operational semantics}{AppCong,Beta,UnitCong,UnitBeta,BoxCong,BoxBeta,SpreadCong,SpreadBeta,CaseCong,CaseOneBeta,CaseTwoBeta}

\begin{theorem}[Preservation]
If $ \Delta \ ;\  \Gamma  \vdash \ottnt{a} : \ottnt{A} $ and $ \ottnt{a}  \leadsto  \ottnt{a'} $ then $ \Delta \ ;\  \Gamma  \vdash \ottnt{a'} : \ottnt{A} $.
\end{theorem}

\begin{theorem}[Progress]
  If $ \varnothing \ ;\  \varnothing  \vdash \ottnt{a} : \ottnt{A} $ then either $a$ is a value or there exists
  some $a'$ such that $ \ottnt{a}  \leadsto  \ottnt{a'} $.
\end{theorem}

\subsection{Heap Semantics}

\drules[Small]{$ [  \ottnt{H}  ]\,  \ottnt{a}  \Rightarrow_{ \ottnt{S} }^{ { \color{black}{r} } } [  \ottnt{H'} \, ;\,  \mathbf{u}' \, ;\,  \Gamma'  ]\,  \ottnt{a'} $}
{Small-step reduction relation (part 1)}
{Var,AppL,AppBeta,UnitL,UnitBeta,CaseL,CaseOne,CaseTwo,Sub}
\drules[Small]{$ [  \ottnt{H}  ]\,  \ottnt{a}  \Rightarrow_{ \ottnt{S} }^{ { \color{black}{r} } } [  \ottnt{H'} \, ;\,  \mathbf{u}' \, ;\,  \Gamma'  ]\,  \ottnt{a'} $}
{Small-step reduction relation (part 2)}
{LetBoxL,LetBoxBeta,ProjL,ProjBeta}

\subsection{Dependent Graded Type System}
\label{app:dependent-types}

The full typing rules for \Langname are below. This system uses the same definition of values and operational semantics as the simple system.

\drules[T]{$ \Delta  ;  \Gamma  \vdash \ottnt{a} : \ottnt{A} $}
{Typing rules for dependent system}
{sub,weak,conv1,type,var,Unit,unit,UnitElim,pi,lam,app,
Sigma,Tensor,SigmaElim,sum,injOne,injTwo,case}

\section{Type-checking a graded, dependent Haskell}
\label{app:type-checking}

This paper concerns itself with an implicit, internal language. Yet, if we are to
integrate with GHC, we must make these ideas practical. There are two type-checking
challenges that will arise:

\begin{description}
\item[Producing \Langname{} via elaboration] A real-world compiler must support
taking a surface language, performing type inference, and then producing well-typed
\Langname{} programs via an elaboration step. The key question here: is \Langname{}
a suitable target for elaboration? We claim that it is. %% Our claim is informed by
%% two facts: \Langname{} is a generalization of GHC's currently implemented
%%  internal language, based on \citet{linear-haskell}, and the careful study of
%% inference and elaboration for a dependently-typed Haskell by \citet{eisenberg-thesis}.
One author of the current paper, Eisenberg, has been involved in the
day-to-day implementation concerns of both linear and dependent types in GHC.
While challenges surely remain in any task this substantial, Eisenberg believes the
type inference concerns of linear types and of dependent types to be largely
orthogonal. The former have been worked out during the implementation of
today's linear types~\cite{linear-haskell}, and the latter have been carefully
studied in the context of Haskell previously~\cite{eisenberg-thesis}.

\item[Checking \Langname{} itself] GHC uses a typed intermediate language.
Type-checking this language serves only as a check on the compiler itself---but
a vital check it is. With the right compiler flags, GHC will repeat the check
after every optimization pass, frequently discovering bugs that might have
otherwise gone unnoticed. If we are to use \Langname{} as GHC's intermediate
language, it, too, must support reasonably efficient type-checking. Yet, \Langname{}
as presented here does not. The solution is not to encode \Langname{} into
GHC directly, but instead use an encoding of \Langname{}'s typing judgements
as the internal language within GHC. The relationship between the implicit
nature and the explicit, implementable nature of a more detailed
encoding is one focus of our previous work~\cite{weirich:icfp17}. A particular
challenge is how to encode the context splitting in, say, the application rule.
The solution is not to encode this at all, but to have grades be an \emph{output}
of the checking algorithm, not an input. The algorithm then checks that the grades
line up with expectations at the binding sites of restricted variables---just as
is done in the implementation today.
\end{description}

\fi

\end{document}

%% file: abstract.tex
Graded Type Theory provides a mechanism to track and reason about
resource usage in type systems. In this paper, we develop \textsc{GraD}\xspace, a novel
version of such a graded dependent type system that includes functions, tensor products,
additive sums, and a unit type. Since standard operational semantics is
resource-agnostic, we develop a heap-based operational semantics and
prove a soundness theorem that shows correct accounting of resource usage. Several useful properties,
including the standard type soundness theorem, non-interference of irrelevant resources in computation and single pointer property for linear resources, can be derived from this theorem. We hope that our work
will provide a base for integrating linearity, irrelevance and
dependent types in practical programming languages like Haskell.

%% file: qtt-rules.tex
% generated by Ott 0.31 from: ../simple/qtt.ott extra.ott
\newcommand{\ottdrule}[4][]{{\displaystyle\frac{\begin{array}{l}#2\end{array}}{#3}\quad\ottdrulename{#4}}}
\newcommand{\ottusedrule}[1]{\[#1\]}
\newcommand{\ottpremise}[1]{ #1 \\}
\newenvironment{ottdefnblock}[3][]{ \framebox{\mbox{#2}} \quad #3 \\[0pt]}{}
\newenvironment{ottfundefnblock}[3][]{ \framebox{\mbox{#2}} \quad #3 \\[0pt]\begin{displaymath}\begin{array}{l}}{\end{array}\end{displaymath}}
\newcommand{\ottfunclause}[2]{ #1 \equiv #2 \\}
\newcommand{\ottnt}[1]{\mathit{#1}}
\newcommand{\ottmv}[1]{\mathit{#1}}
\newcommand{\ottkw}[1]{\mathbf{#1}}
\newcommand{\ottsym}[1]{#1}
\newcommand{\ottcom}[1]{\text{#1}}
\newcommand{\ottdrulename}[1]{\textsc{#1}}
\newcommand{\ottcomplu}[5]{\overline{#1}^{\,#2\in #3 #4 #5}}
\newcommand{\ottcompu}[3]{\overline{#1}^{\,#2<#3}}
\newcommand{\ottcomp}[2]{\overline{#1}^{\,#2}}
\newcommand{\ottgrammartabular}[1]{\begin{supertabular}{llcllllll}#1\end{supertabular}}
\newcommand{\ottmetavartabular}[1]{\begin{supertabular}{ll}#1\end{supertabular}}
\newcommand{\ottrulehead}[3]{$#1$ & & $#2$ & & & \multicolumn{2}{l}{#3}}
\newcommand{\ottprodline}[6]{& & $#1$ & $#2$ & $#3 #4$ & $#5$ & $#6$}
\newcommand{\ottfirstprodline}[6]{\ottprodline{#1}{#2}{#3}{#4}{#5}{#6}}
\newcommand{\ottlongprodline}[2]{& & $#1$ & \multicolumn{4}{l}{$#2$}}
\newcommand{\ottfirstlongprodline}[2]{\ottlongprodline{#1}{#2}}
\newcommand{\ottbindspecprodline}[6]{\ottprodline{#1}{#2}{#3}{#4}{#5}{#6}}
\newcommand{\ottprodnewline}{\\}
\newcommand{\ottinterrule}{\\[5.0mm]}
\newcommand{\ottafterlastrule}{\\}
\newcommand{\ottmetavars}{
\ottmetavartabular{
 $ \ottmv{tmvar} ,\, \ottmv{x} ,\, \ottmv{y} ,\, \ottmv{z} ,\, \ottmv{f} ,\, \ottmv{g} $ & \ottcom{variables} \\
 $ \ottmv{covar} ,\, \ottmv{d} $ & \ottcom{coercion variables} \\
 $ \ottmv{semiringvar} ,\, \ottmv{m} $ & \ottcom{semiring variables} \\
 $ \ottmv{datacon} ,\, \ottmv{K} $ &  \\
 $ \ottmv{const} ,\, \ottmv{T} $ &  \\
 $ \mathit{index} ,\, \mathit{i} ,\, \mathit{j} $ & \ottcom{indices} \\
}}

\newcommand{\ottusage}{
\ottrulehead{\ottnt{usage}  ,\ { \color{black}{q} }  ,\ { \color{black}{r} }  ,\ { \color{black}{s} }}{::=}{}}

\newcommand{\otttm}{
\ottrulehead{\ottnt{tm}  ,\ \ottnt{a}  ,\ \ottnt{b}  ,\ \ottnt{c}  ,\ \ottnt{A}  ,\ \ottnt{B}  ,\ \ottnt{C}  ,\ \ottnt{u}  ,\ \ottnt{v}  ,\ \ottnt{t}}{::=}{\ottcom{terms and types}}\ottprodnewline
\ottfirstprodline{|}{\ottkw{Unit}}{}{}{}{}\ottprodnewline
\ottprodline{|}{\ottkw{unit}}{}{}{}{}\ottprodnewline
\ottprodline{|}{ \ottkw{let}\,  \ottkw{unit} \,=\, \ottnt{a} \ \ottkw{in}\  \ottnt{b} }{}{}{}{}\ottprodnewline
\ottprodline{|}{ \Pi  \ottmv{x} \!:^ { \color{black}{q} } \! \ottnt{A} . \ottnt{B} }{}{\textsf{bind}\; \ottmv{x}\; \textsf{in}\; \ottnt{B}}{}{}\ottprodnewline
\ottprodline{|}{ \lambda \ottmv{x} \!:^ { \color{black}{q} } \! \ottnt{A} . \ottnt{a} }{}{\textsf{bind}\; \ottmv{x}\; \textsf{in}\; \ottnt{a}}{}{}\ottprodnewline
\ottprodline{|}{\ottnt{a} \, \ottnt{b}}{}{}{}{}\ottprodnewline
\ottprodline{|}{ \Box^{ { \color{black}{q} } }  \ottnt{A} }{}{}{}{}\ottprodnewline
\ottprodline{|}{ \ottkw{let}\, \ottkw{box} \, \ottmv{x} \,=\, \ottnt{a} \ \ottkw{in}\  \ottnt{b} }{}{\textsf{bind}\; \ottmv{x}\; \textsf{in}\; \ottnt{b}}{}{}\ottprodnewline
\ottprodline{|}{ \textbf{Type} }{}{}{}{}\ottprodnewline
\ottprodline{|}{\ottmv{x}}{}{}{}{}\ottprodnewline
\ottprodline{|}{\ottnt{a}  \ottsym{\{}  \ottnt{b}  \ottsym{/}  \ottmv{x}  \ottsym{\}}} {\textsf{S}}{}{}{}\ottprodnewline
\ottprodline{|}{ \ottnt{a} } {\textsf{S}}{}{}{\ottcom{parsing precedence is hard}}\ottprodnewline
\ottprodline{|}{ \ottkw{box} _ { \color{black}{q} } \, \ottnt{a} }{}{}{}{}\ottprodnewline
\ottprodline{|}{\ottkw{let} \, \ottmv{x}  \ottsym{=}  \ottnt{a} \, \mathsf{in} \, \ottnt{b}}{}{\textsf{bind}\; \ottmv{x}\; \textsf{in}\; \ottnt{b}}{}{\ottcom{eliminator for box}}\ottprodnewline
\ottprodline{|}{ \ottnt{A_{{\mathrm{1}}}}  \oplus  \ottnt{A_{{\mathrm{2}}}} }{}{}{}{}\ottprodnewline
\ottprodline{|}{ \ottkw{inj}_1\,  \ottnt{a} }{}{}{}{}\ottprodnewline
\ottprodline{|}{ \ottkw{inj}_2\,  \ottnt{a} }{}{}{}{}\ottprodnewline
\ottprodline{|}{ \ottkw{case}_ { \color{black}{q} } \,  \ottnt{a} \, \ottkw{of}\,  \ottnt{b_{{\mathrm{1}}}}  ;  \ottnt{b_{{\mathrm{2}}}} }{}{}{}{}\ottprodnewline
\ottprodline{|}{ \Sigma  \ottmv{x} \!\!:^ { \color{black}{q} } \!\! \ottnt{A} . \ottnt{B} }{}{\textsf{bind}\; \ottmv{x}\; \textsf{in}\; \ottnt{B}}{}{}\ottprodnewline
\ottprodline{|}{\ottsym{(}  \ottnt{a}  \ottsym{,}  \ottnt{b}  \ottsym{)}}{}{}{}{}\ottprodnewline
\ottprodline{|}{ \ottkw{spread}\,  \ottnt{a} \, \ottkw{to}\,  \ottmv{x} \, \ottkw{in}\,  \ottnt{b} }{}{\textsf{bind}\; \ottmv{x}\; \textsf{in}\; \ottnt{b}}{}{}\ottprodnewline
\ottprodline{|}{\ottkw{let} \, \ottsym{(}  \ottmv{x}  \ottsym{,}  \ottmv{y}  \ottsym{)}  \ottsym{=}  \ottnt{a} \, \mathsf{in} \, \ottnt{b}}{}{\textsf{bind}\; \ottmv{x}\; \textsf{in}\; \ottnt{b}}{}{}\ottprodnewline
\ottprodline{|}{\ottkw{Maybe} \, \ottnt{a}}{}{}{}{}\ottprodnewline
\ottprodline{|}{\ottkw{Just} \, \ottnt{a}}{}{}{}{}\ottprodnewline
\ottprodline{|}{\ottkw{Nothing}}{}{}{}{}\ottprodnewline
\ottprodline{|}{\ottkw{case} \, \ottnt{a} \, \ottkw{of} \, \ottkw{Just} \, \ottmv{x}  \to  \ottnt{b_{{\mathrm{1}}}}  \mathsf{;} \, \ottkw{Nothing} \, \to  \ottnt{b_{{\mathrm{2}}}}}{}{\textsf{bind}\; \ottmv{x}\; \textsf{in}\; \ottnt{b_{{\mathrm{1}}}}}{}{}}

\newcommand{\ottsort}{
\ottrulehead{\ottnt{sort}}{::=}{\ottcom{binding classifier}}\ottprodnewline
\ottfirstprodline{|}{\ottkw{Tm} \, \ottnt{A}}{}{}{}{}\ottprodnewline
\ottprodline{|}{\ottkw{Def} \, \ottnt{a} \, \ottnt{A}}{}{}{}{}}

\newcommand{\ottcontext}{
\ottrulehead{\ottnt{context}  ,\ \Gamma}{::=}{\ottcom{contexts}}\ottprodnewline
\ottfirstprodline{|}{\varnothing}{}{}{}{}\ottprodnewline
\ottprodline{|}{ \ottmv{x} \! :^{ { \color{black}{q} } }\! \ottnt{A} }{}{}{}{}\ottprodnewline
\ottprodline{|}{ \ottmv{x}  \! = \!  \ottnt{a}  \! :^{ { \color{black}{q} } } \!  \ottnt{A} }{}{}{}{}}

\newcommand{\ottD}{
\ottrulehead{\Delta}{::=}{\ottcom{contexts}}\ottprodnewline
\ottfirstprodline{|}{\varnothing}{}{}{}{}\ottprodnewline
\ottprodline{|}{ \ottmv{x} \!\!:\!\! \ottnt{A} }{}{}{}{}\ottprodnewline
\ottprodline{|}{ \ottmv{x} \! = \!  \ottnt{a}  \! : \!  \ottnt{A} }{}{}{}{}}

\newcommand{\ottheap}{
\ottrulehead{\ottnt{heap}  ,\ \ottnt{H}}{::=}{\ottcom{heap}}\ottprodnewline
\ottfirstprodline{|}{\varnothing}{}{}{}{}\ottprodnewline
\ottprodline{|}{ \ottmv{x}  \overset{ { \color{black}{q} } }{\mapsto} { \Gamma \vdash  \ottnt{a}  :  \ottnt{A} } }{}{}{}{}\ottprodnewline
\ottprodline{|}{ \ottmv{x}  \stackrel{ { \color{black}{q} } }{\mapsto} \! \ottnt{a}  \in  \Gamma }{}{}{}{}\ottprodnewline
\ottprodline{|}{ \ottmv{x}  \overset{ { \color{black}{q} } }{\mapsto}  \ottnt{a} }{}{}{}{}}

\newcommand{\ottqlist}{
\ottrulehead{\ottnt{qlist}  ,\ \mathbf{u}  ,\ \mathbf{v}}{::=}{\ottcom{qlist}}\ottprodnewline
\ottfirstprodline{|}{ \_ }{}{}{}{}\ottprodnewline
\ottprodline{|}{ { \color{black}{q} } }{}{}{}{}}

\newcommand{\ottsupp}{
\ottrulehead{\ottnt{supp}  ,\ \ottnt{S}}{::=}{\ottcom{support}}}

\newcommand{\ottn}{
\ottrulehead{\ottnt{n}}{::=}{\ottcom{natural number}}}

\newcommand{\ottgrammar}{\ottgrammartabular{
\ottusage\ottinterrule
\otttm\ottinterrule
\ottsort\ottinterrule
\ottcontext\ottinterrule
\ottD\ottinterrule
\ottheap\ottinterrule
\ottqlist\ottinterrule
\ottsupp\ottinterrule
\ottn\ottafterlastrule
}}

% defnss
% defns JOp
%% defn Step
\newcommand{\ottdruleSXXAppCong}[1]{\ottdrule[#1]{%
\ottpremise{ \ottnt{a}  \leadsto  \ottnt{a'} }%
}{
 \ottnt{a} \, \ottnt{b}  \leadsto  \ottnt{a'} \, \ottnt{b} }{%
{\ottdrulename{S\_AppCong}}{}%
}}

\newcommand{\ottdruleSXXBeta}[1]{\ottdrule[#1]{%
}{
 \ottsym{(}   \lambda \ottmv{x} \!:^ { \color{black}{q} } \! \ottnt{A} . \ottnt{a}   \ottsym{)} \, \ottnt{b}  \leadsto  \ottnt{a}  \ottsym{\{}  \ottnt{b}  \ottsym{/}  \ottmv{x}  \ottsym{\}} }{%
{\ottdrulename{S\_Beta}}{}%
}}

\newcommand{\ottdruleSXXUnitCong}[1]{\ottdrule[#1]{%
\ottpremise{ \ottnt{a}  \leadsto  \ottnt{a'} }%
}{
  \ottkw{let}\,  \ottkw{unit} \,=\, \ottnt{a} \ \ottkw{in}\  \ottnt{b}   \leadsto   \ottkw{let}\,  \ottkw{unit} \,=\, \ottnt{a'} \ \ottkw{in}\  \ottnt{b}  }{%
{\ottdrulename{S\_UnitCong}}{}%
}}

\newcommand{\ottdruleSXXUnitBeta}[1]{\ottdrule[#1]{%
}{
  \ottkw{let}\,  \ottkw{unit} \,=\, \ottkw{unit} \ \ottkw{in}\  \ottnt{b}   \leadsto  \ottnt{b} }{%
{\ottdrulename{S\_UnitBeta}}{}%
}}

\newcommand{\ottdruleSXXBoxCong}[1]{\ottdrule[#1]{%
\ottpremise{ \ottnt{a}  \leadsto  \ottnt{a'} }%
}{
  \ottkw{let}\, \ottkw{box} \, \ottmv{x} \,=\, \ottnt{a} \ \ottkw{in}\  \ottnt{b}   \leadsto   \ottkw{let}\, \ottkw{box} \, \ottmv{x} \,=\, \ottnt{a'} \ \ottkw{in}\  \ottnt{b}  }{%
{\ottdrulename{S\_BoxCong}}{}%
}}

\newcommand{\ottdruleSXXBoxBeta}[1]{\ottdrule[#1]{%
}{
  \ottkw{let}\, \ottkw{box} \, \ottmv{x} \,=\,  \ottkw{box} _ { \color{black}{q} } \, \ottnt{a}  \ \ottkw{in}\  \ottnt{b}   \leadsto  \ottnt{b}  \ottsym{\{}  \ottnt{a}  \ottsym{/}  \ottmv{x}  \ottsym{\}} }{%
{\ottdrulename{S\_BoxBeta}}{}%
}}

\newcommand{\ottdruleSXXCaseCong}[1]{\ottdrule[#1]{%
\ottpremise{ \ottnt{a}  \leadsto  \ottnt{a'} }%
}{
  \ottkw{case}_ { \color{black}{q} } \,  \ottnt{a} \, \ottkw{of}\,  \ottnt{b_{{\mathrm{1}}}}  ;  \ottnt{b_{{\mathrm{2}}}}   \leadsto   \ottkw{case}_ { \color{black}{q} } \,  \ottnt{a'} \, \ottkw{of}\,  \ottnt{b_{{\mathrm{1}}}}  ;  \ottnt{b_{{\mathrm{2}}}}  }{%
{\ottdrulename{S\_CaseCong}}{}%
}}

\newcommand{\ottdruleSXXCaseOneBeta}[1]{\ottdrule[#1]{%
}{
  \ottkw{case}_ { \color{black}{q} } \,  \ottsym{(}   \ottkw{inj}_1\,  \ottnt{a}   \ottsym{)} \, \ottkw{of}\,  \ottnt{b_{{\mathrm{1}}}}  ;  \ottnt{b_{{\mathrm{2}}}}   \leadsto  \ottnt{b_{{\mathrm{1}}}} \, \ottnt{a} }{%
{\ottdrulename{S\_Case1Beta}}{}%
}}

\newcommand{\ottdruleSXXCaseTwoBeta}[1]{\ottdrule[#1]{%
}{
  \ottkw{case}_ { \color{black}{q} } \,  \ottsym{(}   \ottkw{inj}_2\,  \ottnt{a}   \ottsym{)} \, \ottkw{of}\,  \ottnt{b_{{\mathrm{1}}}}  ;  \ottnt{b_{{\mathrm{2}}}}   \leadsto  \ottnt{b_{{\mathrm{2}}}} \, \ottnt{a} }{%
{\ottdrulename{S\_Case2Beta}}{}%
}}

\newcommand{\ottdruleSXXSpreadCong}[1]{\ottdrule[#1]{%
\ottpremise{ \ottnt{a}  \leadsto  \ottnt{a'} }%
}{
 \ottkw{let} \, \ottsym{(}  \ottmv{x}  \ottsym{,}  \ottmv{y}  \ottsym{)}  \ottsym{=}  \ottnt{a} \, \mathsf{in} \, \ottnt{b}  \leadsto  \ottkw{let} \, \ottsym{(}  \ottmv{x}  \ottsym{,}  \ottmv{y}  \ottsym{)}  \ottsym{=}  \ottnt{a'} \, \mathsf{in} \, \ottnt{b} }{%
{\ottdrulename{S\_SpreadCong}}{}%
}}

\newcommand{\ottdruleSXXSpreadBeta}[1]{\ottdrule[#1]{%
}{
 \ottkw{let} \, \ottsym{(}  \ottmv{x}  \ottsym{,}  \ottmv{y}  \ottsym{)}  \ottsym{=}  \ottsym{(}  \ottnt{a}  \ottsym{,}  \ottnt{b}  \ottsym{)} \, \mathsf{in} \, \ottnt{b}  \leadsto  \ottnt{b}  \ottsym{\{}  \ottnt{a}  \ottsym{/}  \ottmv{x}  \ottsym{\}}  \ottsym{\{}  \ottnt{b}  \ottsym{/}  \ottmv{y}  \ottsym{\}} }{%
{\ottdrulename{S\_SpreadBeta}}{}%
}}

\newcommand{\ottdefnStep}[1]{\begin{ottdefnblock}[#1]{$ \ottnt{a}  \leadsto  \ottnt{a'} $}{\ottcom{small-step}}
\ottusedrule{\ottdruleSXXAppCong{}}
\ottusedrule{\ottdruleSXXBeta{}}
\ottusedrule{\ottdruleSXXUnitCong{}}
\ottusedrule{\ottdruleSXXUnitBeta{}}
\ottusedrule{\ottdruleSXXBoxCong{}}
\ottusedrule{\ottdruleSXXBoxBeta{}}
\ottusedrule{\ottdruleSXXCaseCong{}}
\ottusedrule{\ottdruleSXXCaseOneBeta{}}
\ottusedrule{\ottdruleSXXCaseTwoBeta{}}
\ottusedrule{\ottdruleSXXSpreadCong{}}
\ottusedrule{\ottdruleSXXSpreadBeta{}}
\end{ottdefnblock}}

\newcommand{\ottdefnsJOp}{
\ottdefnStep{}}

% defns JSimpleTyping
%% defn SimpleTyping
\newcommand{\ottdruleSTXXSub}[1]{\ottdrule[#1]{%
\ottpremise{ \Delta \ ;\  \Gamma_{{\mathrm{1}}}  \vdash \ottnt{a} : \ottnt{A} }%
\ottpremise{ \Gamma_{{\mathrm{1}}}   \leq   \Gamma_{{\mathrm{2}}} }%
}{
 \Delta \ ;\  \Gamma_{{\mathrm{2}}}  \vdash \ottnt{a} : \ottnt{A} }{%
{\ottdrulename{ST\_Sub}}{}%
}}

\newcommand{\ottdruleSTXXVar}[1]{\ottdrule[#1]{%
\ottpremise{\ottmv{x} \, \not\in \, \mathsf{dom} \, \Delta}%
\ottpremise{ \Delta \vdash \Gamma }%
}{
 \ottsym{(}   \Delta ,   \ottmv{x} \!\!:\!\! \ottnt{A}    \ottsym{)} \ ;\  \ottsym{(}      { \color{black}{0} }    \cdot   \Gamma   ,   \ottmv{x} \! :^{  { \color{black}{1} }  }\! \ottnt{A}    \ottsym{)}  \vdash \ottmv{x} : \ottnt{A} }{%
{\ottdrulename{ST\_Var}}{}%
}}

\newcommand{\ottdruleSTXXWeak}[1]{\ottdrule[#1]{%
\ottpremise{\ottmv{x} \, \not\in \, \mathsf{dom} \, \Delta}%
\ottpremise{ \Delta \ ;\  \Gamma  \vdash \ottnt{a} : \ottnt{B} }%
}{
  \Delta ,   \ottmv{x} \!\!:\!\! \ottnt{A}   \ ;\   \Gamma ,   \ottmv{x} \! :^{  { \color{black}{0} }  }\! \ottnt{A}    \vdash \ottnt{a} : \ottnt{B} }{%
{\ottdrulename{ST\_Weak}}{}%
}}

\newcommand{\ottdruleSTXXUnit}[1]{\ottdrule[#1]{%
}{
 \varnothing \ ;\  \varnothing  \vdash \ottkw{unit} : \ottkw{Unit} }{%
{\ottdrulename{ST\_Unit}}{}%
}}

\newcommand{\ottdruleSTXXInjOne}[1]{\ottdrule[#1]{%
\ottpremise{ \Delta \ ;\  \Gamma  \vdash \ottnt{a} : \ottnt{A_{{\mathrm{1}}}} }%
}{
 \Delta \ ;\  \Gamma  \vdash  \ottkw{inj}_1\,  \ottnt{a}  :  \ottnt{A_{{\mathrm{1}}}}  \oplus  \ottnt{A_{{\mathrm{2}}}}  }{%
{\ottdrulename{ST\_Inj1}}{}%
}}

\newcommand{\ottdruleSTXXInjTwo}[1]{\ottdrule[#1]{%
\ottpremise{ \Delta \ ;\  \Gamma  \vdash \ottnt{a} : \ottnt{A_{{\mathrm{2}}}} }%
}{
 \Delta \ ;\  \Gamma  \vdash  \ottkw{inj}_2\,  \ottnt{a}  :  \ottnt{A_{{\mathrm{1}}}}  \oplus  \ottnt{A_{{\mathrm{2}}}}  }{%
{\ottdrulename{ST\_Inj2}}{}%
}}

\newcommand{\ottdruleSTXXCase}[1]{\ottdrule[#1]{%
\ottpremise{ { \color{black}{1} }   \leq  { \color{black}{q} }}%
\ottpremise{ \Delta \ ;\  \Gamma_{{\mathrm{1}}}  \vdash \ottnt{a} :  \ottnt{A_{{\mathrm{1}}}}  \oplus  \ottnt{A_{{\mathrm{2}}}}  }%
\ottpremise{ \Delta \ ;\  \Gamma_{{\mathrm{2}}}  \vdash \ottnt{b_{{\mathrm{1}}}} :  {}^{ { \color{black}{q} } } \ottnt{A_{{\mathrm{1}}}} \rightarrow  \ottnt{B}  }%
\ottpremise{ \Delta \ ;\  \Gamma_{{\mathrm{2}}}  \vdash \ottnt{b_{{\mathrm{2}}}} :  {}^{ { \color{black}{q} } } \ottnt{A_{{\mathrm{2}}}} \rightarrow  \ottnt{B}  }%
}{
 \Delta \ ;\    { \color{black}{q} }   \cdot   \Gamma_{{\mathrm{1}}}    \ottsym{+}  \Gamma_{{\mathrm{2}}}  \vdash  \ottkw{case}_ { \color{black}{q} } \,  \ottnt{a} \, \ottkw{of}\,  \ottnt{b_{{\mathrm{1}}}}  ;  \ottnt{b_{{\mathrm{2}}}}  : \ottnt{B} }{%
{\ottdrulename{ST\_Case}}{}%
}}

\newcommand{\ottdruleSTXXUnitE}[1]{\ottdrule[#1]{%
\ottpremise{ \Delta \ ;\  \Gamma_{{\mathrm{1}}}  \vdash \ottnt{a} : \ottkw{Unit} }%
\ottpremise{ \Delta \ ;\  \Gamma_{{\mathrm{2}}}  \vdash \ottnt{b} : \ottnt{B} }%
}{
 \Delta \ ;\  \Gamma_{{\mathrm{1}}}  \ottsym{+}  \Gamma_{{\mathrm{2}}}  \vdash  \ottkw{let}\,  \ottkw{unit} \,=\, \ottnt{a} \ \ottkw{in}\  \ottnt{b}  : \ottnt{B} }{%
{\ottdrulename{ST\_UnitE}}{}%
}}

\newcommand{\ottdruleSTXXLam}[1]{\ottdrule[#1]{%
\ottpremise{  \Delta ,   \ottmv{x} \!\!:\!\! \ottnt{A}   \ ;\   \Gamma ,   \ottmv{x} \! :^{ { \color{black}{q} } }\! \ottnt{A}    \vdash \ottnt{a} : \ottnt{B} }%
}{
 \Delta \ ;\  \Gamma  \vdash  \lambda \ottmv{x} \!:^ { \color{black}{q} } \! \ottnt{A} . \ottnt{a}  : \ottsym{(}   {}^{ { \color{black}{q} } } \ottnt{A} \rightarrow  \ottnt{B}   \ottsym{)} }{%
{\ottdrulename{ST\_Lam}}{}%
}}

\newcommand{\ottdruleSTXXApp}[1]{\ottdrule[#1]{%
\ottpremise{ \Delta \ ;\  \Gamma_{{\mathrm{1}}}  \vdash \ottnt{a} : \ottsym{(}   {}^{ { \color{black}{q} } } \ottnt{A} \rightarrow  \ottnt{B}   \ottsym{)} }%
\ottpremise{ \Delta \ ;\  \Gamma_{{\mathrm{2}}}  \vdash \ottnt{b} : \ottnt{A} }%
}{
 \Delta \ ;\  \Gamma_{{\mathrm{1}}}  \ottsym{+}   { \color{black}{q} }   \cdot   \Gamma_{{\mathrm{2}}}   \vdash \ottnt{a} \, \ottnt{b} : \ottnt{B} }{%
{\ottdrulename{ST\_App}}{}%
}}

\newcommand{\ottdruleSTXXBox}[1]{\ottdrule[#1]{%
\ottpremise{ \Delta \ ;\  \Gamma  \vdash \ottnt{a} : \ottnt{A} }%
}{
 \Delta \ ;\   { \color{black}{q} }   \cdot   \Gamma   \vdash  \ottkw{box} _ { \color{black}{q} } \, \ottnt{a}  :  \Box^{ { \color{black}{q} } }  \ottnt{A}  }{%
{\ottdrulename{ST\_Box}}{}%
}}

\newcommand{\ottdruleSTXXLetBox}[1]{\ottdrule[#1]{%
\ottpremise{ \Delta \ ;\  \Gamma_{{\mathrm{1}}}  \vdash \ottnt{a} :  \Box^{ { \color{black}{q} } }  \ottnt{A}  }%
\ottpremise{  \Delta ,   \ottmv{x} \!\!:\!\! \ottnt{A}   \ ;\   \Gamma_{{\mathrm{2}}} ,   \ottmv{x} \! :^{ { \color{black}{q} } }\! \ottnt{A}    \vdash \ottnt{b} : \ottnt{B} }%
}{
 \Delta \ ;\  \Gamma_{{\mathrm{1}}}  \ottsym{+}  \Gamma_{{\mathrm{2}}}  \vdash  \ottkw{let}\, \ottkw{box} \, \ottmv{x} \,=\, \ottnt{a} \ \ottkw{in}\  \ottnt{b}  : \ottnt{B} }{%
{\ottdrulename{ST\_LetBox}}{}%
}}

\newcommand{\ottdruleSTXXPair}[1]{\ottdrule[#1]{%
\ottpremise{ \Delta \ ;\  \Gamma_{{\mathrm{1}}}  \vdash \ottnt{a} : \ottnt{A} }%
\ottpremise{ \Delta \ ;\  \Gamma_{{\mathrm{2}}}  \vdash \ottnt{b} : \ottnt{B} }%
}{
 \Delta \ ;\  \Gamma_{{\mathrm{1}}}  \ottsym{+}  \Gamma_{{\mathrm{2}}}  \vdash \ottsym{(}  \ottnt{a}  \ottsym{,}  \ottnt{b}  \ottsym{)} :  \ottnt{A}  \otimes  \ottnt{B}  }{%
{\ottdrulename{ST\_Pair}}{}%
}}

\newcommand{\ottdruleSTXXSpread}[1]{\ottdrule[#1]{%
\ottpremise{ \Delta \ ;\  \Gamma_{{\mathrm{1}}}  \vdash \ottnt{a} :  \ottnt{A_{{\mathrm{1}}}}  \otimes  \ottnt{A_{{\mathrm{2}}}}  }%
\ottpremise{ \Delta \ ;\   \Gamma_{{\mathrm{2}}} ,     \ottmv{x} \! :^{  { \color{black}{1} }  }\! \ottnt{A_{{\mathrm{1}}}}  ,   \ottmv{y} \! :^{  { \color{black}{1} }  }\! \ottnt{A_{{\mathrm{2}}}}      \vdash \ottnt{b} : \ottnt{B} }%
}{
 \Delta \ ;\  \Gamma_{{\mathrm{1}}}  \ottsym{+}  \Gamma_{{\mathrm{2}}}  \vdash \ottkw{let} \, \ottsym{(}  \ottmv{x}  \ottsym{,}  \ottmv{y}  \ottsym{)}  \ottsym{=}  \ottnt{a} \, \mathsf{in} \, \ottnt{b} : \ottnt{B} }{%
{\ottdrulename{ST\_Spread}}{}%
}}

\newcommand{\ottdefnSimpleTyping}[1]{\begin{ottdefnblock}[#1]{$ \Delta \ ;\  \Gamma  \vdash \ottnt{a} : \ottnt{A} $}{\ottcom{Simple Typing}}
\ottusedrule{\ottdruleSTXXSub{}}
\ottusedrule{\ottdruleSTXXVar{}}
\ottusedrule{\ottdruleSTXXWeak{}}
\ottusedrule{\ottdruleSTXXUnit{}}
\ottusedrule{\ottdruleSTXXInjOne{}}
\ottusedrule{\ottdruleSTXXInjTwo{}}
\ottusedrule{\ottdruleSTXXCase{}}
\ottusedrule{\ottdruleSTXXUnitE{}}
\ottusedrule{\ottdruleSTXXLam{}}
\ottusedrule{\ottdruleSTXXApp{}}
\ottusedrule{\ottdruleSTXXBox{}}
\ottusedrule{\ottdruleSTXXLetBox{}}
\ottusedrule{\ottdruleSTXXPair{}}
\ottusedrule{\ottdruleSTXXSpread{}}
\end{ottdefnblock}}

\newcommand{\ottdefnsJSimpleTyping}{
\ottdefnSimpleTyping{}}

% defns JPath
%% defn Path
\newcommand{\ottdrulePXXVar}[1]{\ottdrule[#1]{%
}{
 \text{Path }  \ottmv{x} }{%
{\ottdrulename{P\_Var}}{}%
}}

\newcommand{\ottdrulePXXApp}[1]{\ottdrule[#1]{%
\ottpremise{ \text{Path }  \ottnt{a} }%
}{
 \text{Path }  \ottnt{a} \, \ottnt{b} }{%
{\ottdrulename{P\_App}}{}%
}}

\newcommand{\ottdrulePXXPair}[1]{\ottdrule[#1]{%
\ottpremise{ \text{Path }  \ottnt{a} }%
}{
 \text{Path }  \ottsym{(}  \ottkw{let} \, \ottsym{(}  \ottmv{x}  \ottsym{,}  \ottmv{y}  \ottsym{)}  \ottsym{=}  \ottnt{a} \, \mathsf{in} \, \ottnt{b}  \ottsym{)} }{%
{\ottdrulename{P\_Pair}}{}%
}}

\newcommand{\ottdrulePXXUnit}[1]{\ottdrule[#1]{%
\ottpremise{ \text{Path }  \ottnt{a} }%
}{
 \text{Path }  \ottsym{(}   \ottkw{let}\,  \ottkw{unit} \,=\, \ottnt{a} \ \ottkw{in}\  \ottnt{b}   \ottsym{)} }{%
{\ottdrulename{P\_Unit}}{}%
}}

\newcommand{\ottdrulePXXCase}[1]{\ottdrule[#1]{%
\ottpremise{ \text{Path }  \ottnt{a} }%
}{
 \text{Path }  \ottsym{(}   \ottkw{case}_ { \color{black}{q} } \,  \ottnt{a} \, \ottkw{of}\,  \ottnt{b_{{\mathrm{1}}}}  ;   \ottnt{b_{{\mathrm{2}}}}   \ottsym{)} }{%
{\ottdrulename{P\_Case}}{}%
}}

\newcommand{\ottdefnPath}[1]{\begin{ottdefnblock}[#1]{$ \text{Path }  \ottnt{a} $}{\ottcom{Path}}
\ottusedrule{\ottdrulePXXVar{}}
\ottusedrule{\ottdrulePXXApp{}}
\ottusedrule{\ottdrulePXXPair{}}
\ottusedrule{\ottdrulePXXUnit{}}
\ottusedrule{\ottdrulePXXCase{}}
\end{ottdefnblock}}

\newcommand{\ottdefnsJPath}{
\ottdefnPath{}}

% defns JTyping
%% defn Typing
\newcommand{\ottdruleTXXsub}[1]{\ottdrule[#1]{%
\ottpremise{ \Delta  ;  \Gamma_{{\mathrm{1}}}  \vdash \ottnt{a} : \ottnt{A} }%
\ottpremise{ \Gamma_{{\mathrm{1}}}   \leq   \Gamma_{{\mathrm{2}}} }%
}{
 \Delta  ;  \Gamma_{{\mathrm{2}}}  \vdash \ottnt{a} : \ottnt{A} }{%
{\ottdrulename{T\_sub}}{}%
}}

\newcommand{\ottdruleTXXtype}[1]{\ottdrule[#1]{%
}{
 \varnothing  ;  \varnothing  \vdash  \textbf{Type}  :  \textbf{Type}  }{%
{\ottdrulename{T\_type}}{}%
}}

\newcommand{\ottdruleTXXvar}[1]{\ottdrule[#1]{%
\ottpremise{\ottmv{x} \, \not\in \, \mathsf{dom} \, \Delta}%
\ottpremise{ \Delta  ;  \Gamma  \vdash \ottnt{A} :  \textbf{Type}  }%
}{
  \Delta ,   \ottmv{x} \!\!:\!\! \ottnt{A}    ;      { \color{black}{0} }    \cdot   \Gamma   ,   \ottmv{x} \! :^{  { \color{black}{1} }  }\! \ottnt{A}    \vdash \ottmv{x} : \ottnt{A} }{%
{\ottdrulename{T\_var}}{}%
}}

\newcommand{\ottdruleTXXdef}[1]{\ottdrule[#1]{%
\ottpremise{\ottmv{x} \, \not\in \, \mathsf{dom} \, \Delta}%
\ottpremise{ \Delta  ;  \Gamma  \vdash \ottnt{a} : \ottnt{A} }%
}{
  \Delta ,   \ottmv{x} \! = \!  \ottnt{a}  \! : \!  \ottnt{A}    ;      { \color{black}{0} }    \cdot   \Gamma   ,   \ottmv{x}  \! = \!  \ottnt{a}  \! :^{  { \color{black}{1} }  } \!  \ottnt{A}    \vdash \ottmv{x} : \ottnt{A} }{%
{\ottdrulename{T\_def}}{}%
}}

\newcommand{\ottdruleTXXweak}[1]{\ottdrule[#1]{%
\ottpremise{\ottmv{x} \, \not\in \, \mathsf{dom} \, \Delta}%
\ottpremise{ \Delta  ;  \Gamma_{{\mathrm{1}}}  \vdash \ottnt{a} : \ottnt{B} }%
\ottpremise{ \Delta  ;  \Gamma_{{\mathrm{2}}}  \vdash \ottnt{A} :  \textbf{Type}  }%
}{
  \Delta ,   \ottmv{x} \!\!:\!\! \ottnt{A}    ;   \Gamma_{{\mathrm{1}}} ,   \ottmv{x} \! :^{  { \color{black}{0} }  }\! \ottnt{A}    \vdash \ottnt{a} : \ottnt{B} }{%
{\ottdrulename{T\_weak}}{}%
}}

\newcommand{\ottdruleTXXweakXXdef}[1]{\ottdrule[#1]{%
\ottpremise{\ottmv{x} \, \not\in \, \mathsf{dom} \, \Delta}%
\ottpremise{ \Delta  ;  \Gamma_{{\mathrm{1}}}  \vdash \ottnt{b} : \ottnt{B} }%
\ottpremise{ \Delta  ;  \Gamma_{{\mathrm{2}}}  \vdash \ottnt{a} : \ottnt{A} }%
}{
  \Delta ,   \ottmv{x} \! = \!  \ottnt{a}  \! : \!  \ottnt{A}    ;   \Gamma_{{\mathrm{1}}} ,   \ottmv{x}  \! = \!  \ottnt{a}  \! :^{  { \color{black}{0} }  } \!  \ottnt{A}    \vdash \ottnt{b} : \ottnt{B} }{%
{\ottdrulename{T\_weak\_def}}{}%
}}

\newcommand{\ottdruleTXXpi}[1]{\ottdrule[#1]{%
\ottpremise{ \Delta  ;  \Gamma_{{\mathrm{1}}}  \vdash \ottnt{A} :  \textbf{Type}  }%
\ottpremise{  \Delta ,   \ottmv{x} \!\!:\!\! \ottnt{A}    ;   \Gamma_{{\mathrm{2}}} ,   \ottmv{x} \! :^{ { \color{black}{r} } }\! \ottnt{A}    \vdash \ottnt{B} :  \textbf{Type}  }%
}{
 \Delta  ;  \Gamma_{{\mathrm{1}}}  \ottsym{+}  \Gamma_{{\mathrm{2}}}  \vdash  \Pi  \ottmv{x} \!:^ { \color{black}{q} } \! \ottnt{A} . \ottnt{B}  :  \textbf{Type}  }{%
{\ottdrulename{T\_pi}}{}%
}}

\newcommand{\ottdruleTXXlam}[1]{\ottdrule[#1]{%
\ottpremise{  \Delta ,   \ottmv{x} \!\!:\!\! \ottnt{A}    ;   \Gamma_{{\mathrm{1}}} ,   \ottmv{x} \! :^{ { \color{black}{q} } }\! \ottnt{A}    \vdash \ottnt{a} : \ottnt{B} }%
\ottpremise{ \Delta  ;  \Gamma_{{\mathrm{2}}}  \vdash \ottnt{A} :  \textbf{Type}  }%
}{
 \Delta  ;  \Gamma_{{\mathrm{1}}}  \vdash  \lambda \ottmv{x} \!:^ { \color{black}{q} } \! \ottnt{A} . \ottnt{a}  :  \Pi  \ottmv{x} \!:^ { \color{black}{q} } \! \ottnt{A} . \ottnt{B}  }{%
{\ottdrulename{T\_lam}}{}%
}}

\newcommand{\ottdruleTXXapp}[1]{\ottdrule[#1]{%
\ottpremise{ \Delta  ;  \Gamma_{{\mathrm{1}}}  \vdash \ottnt{a} :  \Pi  \ottmv{x} \!:^ { \color{black}{q} } \! \ottnt{A} . \ottnt{B}  }%
\ottpremise{ \Delta  ;  \Gamma_{{\mathrm{2}}}  \vdash \ottnt{b} : \ottnt{A} }%
}{
 \Delta  ;  \Gamma_{{\mathrm{1}}}  \ottsym{+}    { \color{black}{q} }   \cdot   \Gamma_{{\mathrm{2}}}    \vdash \ottnt{a} \, \ottnt{b} : \ottnt{B}  \ottsym{\{}  \ottnt{b}  \ottsym{/}  \ottmv{x}  \ottsym{\}} }{%
{\ottdrulename{T\_app}}{}%
}}

\newcommand{\ottdruleTXXconv}[1]{\ottdrule[#1]{%
\ottpremise{ \Delta  ;  \Gamma_{{\mathrm{1}}}  \vdash \ottnt{a} : \ottnt{A} }%
\ottpremise{ \Delta  ;  \Gamma_{{\mathrm{2}}}  \vdash \ottnt{B} :  \textbf{Type}  }%
\ottpremise{ \ottnt{A}  \{  \Delta  \}   \equiv   \ottnt{B}  \{  \Delta  \} }%
}{
 \Delta  ;  \Gamma_{{\mathrm{1}}}  \vdash \ottnt{a} : \ottnt{B} }{%
{\ottdrulename{T\_conv}}{}%
}}

\newcommand{\ottdruleTXXunit}[1]{\ottdrule[#1]{%
}{
 \varnothing  ;  \varnothing  \vdash \ottkw{unit} : \ottkw{Unit} }{%
{\ottdrulename{T\_unit}}{}%
}}

\newcommand{\ottdruleTXXUnit}[1]{\ottdrule[#1]{%
}{
 \varnothing  ;  \varnothing  \vdash \ottkw{Unit} :  \textbf{Type}  }{%
{\ottdrulename{T\_Unit}}{}%
}}

\newcommand{\ottdruleTXXUnitE}[1]{\ottdrule[#1]{%
\ottpremise{ \Delta  ;  \Gamma_{{\mathrm{1}}}  \vdash \ottnt{a} : \ottkw{Unit} }%
\ottpremise{\ottnt{B_{{\mathrm{1}}}}  \ottsym{=}  \ottnt{B}  \ottsym{\{}  \ottkw{unit}  \ottsym{/}  \ottmv{y}  \ottsym{\}}}%
\ottpremise{ \Delta  ;  \Gamma_{{\mathrm{2}}}  \vdash \ottnt{b} : \ottnt{B_{{\mathrm{1}}}} }%
\ottpremise{  \Delta ,   \ottmv{y} \!\!:\!\! \ottkw{Unit}    ;   \Gamma_{{\mathrm{3}}} ,   \ottmv{y} \! :^{ { \color{black}{r} } }\! \ottkw{Unit}    \vdash \ottnt{B} :  \textbf{Type}  }%
}{
 \Delta  ;  \Gamma_{{\mathrm{1}}}  \ottsym{+}  \Gamma_{{\mathrm{2}}}  \vdash  \ottkw{let}\,  \ottkw{unit} \,=\, \ottnt{a} \ \ottkw{in}\  \ottnt{b}  : \ottnt{B}  \ottsym{\{}  \ottnt{a}  \ottsym{/}  \ottmv{y}  \ottsym{\}} }{%
{\ottdrulename{T\_UnitE}}{}%
}}

\newcommand{\ottdruleTXXBox}[1]{\ottdrule[#1]{%
\ottpremise{ \Delta  ;  \Gamma  \vdash \ottnt{A} :  \textbf{Type}  }%
}{
 \Delta  ;  \Gamma  \vdash  \Box^{ { \color{black}{q} } }  \ottnt{A}  :  \textbf{Type}  }{%
{\ottdrulename{T\_Box}}{}%
}}

\newcommand{\ottdruleTXXbox}[1]{\ottdrule[#1]{%
\ottpremise{ \Delta  ;  \Gamma  \vdash \ottnt{a} : \ottnt{A} }%
}{
 \Delta  ;   { \color{black}{q} }   \cdot   \Gamma   \vdash  \ottkw{box} _ { \color{black}{q} } \, \ottnt{a}  :  \Box^{ { \color{black}{q} } }  \ottnt{A}  }{%
{\ottdrulename{T\_box}}{}%
}}

\newcommand{\ottdruleTXXletbox}[1]{\ottdrule[#1]{%
\ottpremise{ \Delta  ;  \Gamma_{{\mathrm{1}}}  \vdash \ottnt{a} :  \Box^{ { \color{black}{q} } }  \ottnt{A}  }%
\ottpremise{  \Delta ,   \ottmv{x} \!\!:\!\! \ottnt{A}    ;   \Gamma_{{\mathrm{2}}} ,   \ottmv{x} \! :^{ { \color{black}{q} } }\! \ottnt{A}    \vdash \ottnt{b} : \ottnt{B}  \ottsym{\{}   \ottkw{box} _ { \color{black}{q} } \, \ottmv{x}   \ottsym{/}  \ottmv{y}  \ottsym{\}} }%
\ottpremise{  \Delta ,   \ottmv{y} \!\!:\!\!  \Box^{ { \color{black}{q} } }  \ottnt{A}     ;   \Gamma_{{\mathrm{3}}} ,   \ottmv{y} \! :^{ { \color{black}{r} } }\!  \Box^{ { \color{black}{q} } }  \ottnt{A}     \vdash \ottnt{B} :  \textbf{Type}  }%
}{
 \Delta  ;  \Gamma_{{\mathrm{1}}}  \ottsym{+}  \Gamma_{{\mathrm{2}}}  \vdash  \ottkw{let}\, \ottkw{box} \, \ottmv{x} \,=\, \ottnt{a} \ \ottkw{in}\  \ottnt{b}  : \ottnt{B}  \ottsym{\{}  \ottnt{a}  \ottsym{/}  \ottmv{y}  \ottsym{\}} }{%
{\ottdrulename{T\_letbox}}{}%
}}

\newcommand{\ottdruleTXXsum}[1]{\ottdrule[#1]{%
\ottpremise{ \Delta  ;  \Gamma_{{\mathrm{1}}}  \vdash \ottnt{A_{{\mathrm{1}}}} :  \textbf{Type}  }%
\ottpremise{ \Delta  ;  \Gamma_{{\mathrm{2}}}  \vdash \ottnt{A_{{\mathrm{2}}}} :  \textbf{Type}  }%
}{
 \Delta  ;  \Gamma_{{\mathrm{1}}}  \ottsym{+}  \Gamma_{{\mathrm{2}}}  \vdash  \ottnt{A_{{\mathrm{1}}}}  \oplus  \ottnt{A_{{\mathrm{2}}}}  :  \textbf{Type}  }{%
{\ottdrulename{T\_sum}}{}%
}}

\newcommand{\ottdruleTXXinjOne}[1]{\ottdrule[#1]{%
\ottpremise{ \Delta  ;  \Gamma  \vdash \ottnt{a} : \ottnt{A_{{\mathrm{1}}}} }%
\ottpremise{ \Delta  ;  \Gamma_{{\mathrm{1}}}  \vdash \ottnt{A_{{\mathrm{2}}}} :  \textbf{Type}  }%
}{
 \Delta  ;  \Gamma  \vdash  \ottkw{inj}_1\,  \ottnt{a}  :  \ottnt{A_{{\mathrm{1}}}}  \oplus  \ottnt{A_{{\mathrm{2}}}}  }{%
{\ottdrulename{T\_inj1}}{}%
}}

\newcommand{\ottdruleTXXinjTwo}[1]{\ottdrule[#1]{%
\ottpremise{ \Delta  ;  \Gamma  \vdash \ottnt{a} : \ottnt{A_{{\mathrm{2}}}} }%
\ottpremise{ \Delta  ;  \Gamma_{{\mathrm{1}}}  \vdash \ottnt{A_{{\mathrm{1}}}} :  \textbf{Type}  }%
}{
 \Delta  ;  \Gamma  \vdash  \ottkw{inj}_2\,  \ottnt{a}  :  \ottnt{A_{{\mathrm{1}}}}  \oplus  \ottnt{A_{{\mathrm{2}}}}  }{%
{\ottdrulename{T\_inj2}}{}%
}}

\newcommand{\ottdruleTXXcase}[1]{\ottdrule[#1]{%
\ottpremise{ { \color{black}{1} }   \leq  { \color{black}{q} }}%
\ottpremise{ \Delta  ;  \Gamma_{{\mathrm{1}}}  \vdash \ottnt{a} :  \ottnt{A_{{\mathrm{1}}}}  \oplus  \ottnt{A_{{\mathrm{2}}}}  }%
\ottpremise{\ottnt{B_{{\mathrm{1}}}}  \ottsym{=}   \ottnt{B}  \ottsym{\{}   \ottkw{inj}_1\,  \ottmv{x}   \ottsym{/}  \ottmv{y}  \ottsym{\}} }%
\ottpremise{\ottnt{B_{{\mathrm{2}}}}  \ottsym{=}   \ottnt{B}  \ottsym{\{}   \ottkw{inj}_2\,  \ottmv{x}   \ottsym{/}  \ottmv{y}  \ottsym{\}} }%
\ottpremise{ \Delta  ;  \Gamma_{{\mathrm{2}}}  \vdash \ottnt{b_{{\mathrm{1}}}} :  \Pi  \ottmv{x} \!:^ { \color{black}{q} } \! \ottnt{A_{{\mathrm{1}}}} . \ottnt{B_{{\mathrm{1}}}}  }%
\ottpremise{ \Delta  ;  \Gamma_{{\mathrm{2}}}  \vdash \ottnt{b_{{\mathrm{2}}}} :  \Pi  \ottmv{x} \!:^ { \color{black}{q} } \! \ottnt{A_{{\mathrm{2}}}} . \ottnt{B_{{\mathrm{2}}}}  }%
\ottpremise{  \Delta ,   \ottmv{y} \!\!:\!\!  \ottnt{A_{{\mathrm{1}}}}  \oplus  \ottnt{A_{{\mathrm{2}}}}     ;   \Gamma_{{\mathrm{3}}} ,   \ottmv{y} \! :^{ { \color{black}{r} } }\!  \ottnt{A_{{\mathrm{1}}}}  \oplus  \ottnt{A_{{\mathrm{2}}}}     \vdash \ottnt{B} :  \textbf{Type}  }%
}{
 \Delta  ;    { \color{black}{q} }   \cdot   \Gamma_{{\mathrm{1}}}    \ottsym{+}  \Gamma_{{\mathrm{2}}}  \vdash  \ottkw{case}_ { \color{black}{q} } \,  \ottnt{a} \, \ottkw{of}\,  \ottnt{b_{{\mathrm{1}}}}  ;  \ottnt{b_{{\mathrm{2}}}}  : \ottnt{B}  \ottsym{\{}  \ottnt{a}  \ottsym{/}  \ottmv{y}  \ottsym{\}} }{%
{\ottdrulename{T\_case}}{}%
}}

\newcommand{\ottdruleTXXSigma}[1]{\ottdrule[#1]{%
\ottpremise{ \Delta  ;  \Gamma_{{\mathrm{1}}}  \vdash \ottnt{A} :  \textbf{Type}  }%
\ottpremise{  \Delta ,   \ottmv{x} \!\!:\!\! \ottnt{A}    ;   \Gamma_{{\mathrm{2}}} ,   \ottmv{x} \! :^{ { \color{black}{r} } }\! \ottnt{A}    \vdash \ottnt{B} :  \textbf{Type}  }%
}{
 \Delta  ;  \Gamma_{{\mathrm{1}}}  \ottsym{+}  \Gamma_{{\mathrm{2}}}  \vdash  \Sigma  \ottmv{x} \!\!:^ { \color{black}{q} } \!\! \ottnt{A} . \ottnt{B}  :  \textbf{Type}  }{%
{\ottdrulename{T\_Sigma}}{}%
}}

\newcommand{\ottdruleTXXTensor}[1]{\ottdrule[#1]{%
\ottpremise{ \Delta  ;  \Gamma_{{\mathrm{1}}}  \vdash \ottnt{a} : \ottnt{A} }%
\ottpremise{ \Delta  ;  \Gamma_{{\mathrm{2}}}  \vdash \ottnt{b} : \ottnt{B}  \ottsym{\{}  \ottnt{a}  \ottsym{/}  \ottmv{x}  \ottsym{\}} }%
\ottpremise{  \Delta ,   \ottmv{x} \!\!:\!\! \ottnt{A}    ;   \Gamma_{{\mathrm{3}}} ,   \ottmv{x} \! :^{ { \color{black}{r} } }\! \ottnt{A}    \vdash \ottnt{B} :  \textbf{Type}  }%
}{
 \Delta  ;    { \color{black}{q} }   \cdot   \Gamma_{{\mathrm{1}}}    \ottsym{+}  \Gamma_{{\mathrm{2}}}  \vdash \ottsym{(}  \ottnt{a}  \ottsym{,}  \ottnt{b}  \ottsym{)} :  \Sigma  \ottmv{x} \!\!:^ { \color{black}{q} } \!\! \ottnt{A} . \ottnt{B}  }{%
{\ottdrulename{T\_Tensor}}{}%
}}

\newcommand{\ottdruleTXXSpread}[1]{\ottdrule[#1]{%
\ottpremise{\ottnt{A}  \ottsym{=}   \Sigma  \ottmv{x} \!\!:^ { \color{black}{q} } \!\! \ottnt{A_{{\mathrm{1}}}} . \ottnt{A_{{\mathrm{2}}}} }%
\ottpremise{ \Delta  ;  \Gamma_{{\mathrm{1}}}  \vdash \ottnt{a} : \ottnt{A} }%
\ottpremise{  \Delta ,   \ottmv{x} \!\!:\!\! \ottnt{A_{{\mathrm{1}}}}    ;   \Gamma_{{\mathrm{2}}} ,   \ottmv{x} \! :^{ { \color{black}{q} } }\! \ottnt{A_{{\mathrm{1}}}}    \vdash \ottnt{b} :  \Pi  \ottmv{y} \!:^  { \color{black}{1} }  \! \ottnt{A_{{\mathrm{2}}}} .  \ottnt{B}  \ottsym{\{}  hack \, \ottmv{x} \, \ottmv{y}  \ottsym{/}  \ottmv{y}  \ottsym{\}}   }%
\ottpremise{  \Delta ,   \ottmv{y} \!\!:\!\! \ottnt{A}    ;   \Gamma_{{\mathrm{3}}} ,   \ottmv{y} \! :^{ { \color{black}{r} } }\! \ottnt{A}    \vdash \ottnt{B} :  \textbf{Type}  }%
}{
 \Delta  ;  \Gamma_{{\mathrm{1}}}  \ottsym{+}  \Gamma_{{\mathrm{2}}}  \vdash  \ottkw{spread}\,  \ottnt{a} \, \ottkw{to}\,  \ottmv{x} \, \ottkw{in}\,  \ottnt{b}  : \ottnt{B}  \ottsym{\{}  \ottnt{a}  \ottsym{/}  \ottmv{y}  \ottsym{\}} }{%
{\ottdrulename{T\_Spread}}{}%
}}

\newcommand{\ottdruleTXXSigmaElim}[1]{\ottdrule[#1]{%
\ottpremise{ \Delta  ;  \Gamma_{{\mathrm{1}}}  \vdash \ottnt{a} :  \Sigma  \ottmv{x} \!\!:^ { \color{black}{q} } \!\! \ottnt{A_{{\mathrm{1}}}} . \ottnt{A_{{\mathrm{2}}}}  }%
\ottpremise{    \Delta ,   \ottmv{x} \!\!:\!\! \ottnt{A_{{\mathrm{1}}}}    ,   \ottmv{y} \!\!:\!\! \ottnt{A_{{\mathrm{2}}}}    ;     \Gamma_{{\mathrm{2}}} ,   \ottmv{x} \! :^{ { \color{black}{q} } }\! \ottnt{A_{{\mathrm{1}}}}    ,   \ottmv{y} \! :^{  { \color{black}{1} }  }\! \ottnt{A_{{\mathrm{2}}}}    \vdash \ottnt{b} : \ottnt{B}  \ottsym{\{}  \ottsym{(}  \ottmv{x}  \ottsym{,}  \ottmv{y}  \ottsym{)}  \ottsym{/}  \ottmv{z}  \ottsym{\}} }%
\ottpremise{  \Delta ,   \ottmv{z} \!\!:\!\! \ottsym{(}   \Sigma  \ottmv{x} \!\!:^ { \color{black}{q} } \!\! \ottnt{A_{{\mathrm{1}}}} . \ottnt{A_{{\mathrm{2}}}}   \ottsym{)}    ;   \Gamma_{{\mathrm{3}}} ,   \ottmv{z} \! :^{ { \color{black}{r} } }\! \ottsym{(}   \Sigma  \ottmv{x} \!\!:^ { \color{black}{q} } \!\! \ottnt{A_{{\mathrm{1}}}} . \ottnt{A_{{\mathrm{2}}}}   \ottsym{)}    \vdash \ottnt{B} :  \textbf{Type}  }%
}{
 \Delta  ;  \Gamma_{{\mathrm{1}}}  \ottsym{+}  \Gamma_{{\mathrm{2}}}  \vdash \ottkw{let} \, \ottsym{(}  \ottmv{x}  \ottsym{,}  \ottmv{y}  \ottsym{)}  \ottsym{=}  \ottnt{a} \, \mathsf{in} \, \ottnt{b} : \ottnt{B}  \ottsym{\{}  \ottnt{a}  \ottsym{/}  \ottmv{z}  \ottsym{\}} }{%
{\ottdrulename{T\_SigmaElim}}{}%
}}

\newcommand{\ottdefnTyping}[1]{\begin{ottdefnblock}[#1]{$ \Delta  ;  \Gamma  \vdash \ottnt{a} : \ottnt{A} $}{\ottcom{Typing}}
\ottusedrule{\ottdruleTXXsub{}}
\ottusedrule{\ottdruleTXXtype{}}
\ottusedrule{\ottdruleTXXvar{}}
\ottusedrule{\ottdruleTXXdef{}}
\ottusedrule{\ottdruleTXXweak{}}
\ottusedrule{\ottdruleTXXweakXXdef{}}
\ottusedrule{\ottdruleTXXpi{}}
\ottusedrule{\ottdruleTXXlam{}}
\ottusedrule{\ottdruleTXXapp{}}
\ottusedrule{\ottdruleTXXconv{}}
\ottusedrule{\ottdruleTXXunit{}}
\ottusedrule{\ottdruleTXXUnit{}}
\ottusedrule{\ottdruleTXXUnitE{}}
\ottusedrule{\ottdruleTXXBox{}}
\ottusedrule{\ottdruleTXXbox{}}
\ottusedrule{\ottdruleTXXletbox{}}
\ottusedrule{\ottdruleTXXsum{}}
\ottusedrule{\ottdruleTXXinjOne{}}
\ottusedrule{\ottdruleTXXinjTwo{}}
\ottusedrule{\ottdruleTXXcase{}}
\ottusedrule{\ottdruleTXXSigma{}}
\ottusedrule{\ottdruleTXXTensor{}}
\ottusedrule{\ottdruleTXXSpread{}}
\ottusedrule{\ottdruleTXXSigmaElim{}}
\end{ottdefnblock}}

\newcommand{\ottdefnsJTyping}{
\ottdefnTyping{}}

% defns JSmallStep
%% defn SmallStep
\newcommand{\ottdruleSmallXXVar}[1]{\ottdrule[#1]{%
\ottpremise{  }%
\ottpremise{ { \color{black}{1} }   \leq  { \color{black}{r} }}%
}{
 [   \ottnt{H_{{\mathrm{1}}}}  ,     \ottmv{x}  \overset{ \ottsym{(}  { \color{black}{q} }  \ottsym{+}  { \color{black}{r} }  \ottsym{)} }{\mapsto} { \Gamma \vdash  \ottnt{a}  :  \ottnt{A} }   ,  \ottnt{H_{{\mathrm{2}}}}     ]\,  \ottmv{x}  \Rightarrow_{ \ottnt{S} }^{ { \color{black}{r} } } [   \ottnt{H_{{\mathrm{1}}}}  ,     \ottmv{x}  \overset{ { \color{black}{q} } }{\mapsto} { \Gamma \vdash  \ottnt{a}  :  \ottnt{A} }   ,  \ottnt{H_{{\mathrm{2}}}}    \, ;\,    \mathbf{0}^{| \ottnt{H_{{\mathrm{1}}}} |}   \mathop{\diamond}     { \color{black}{r} }   \mathop{\diamond}   \mathbf{0}^{| \ottnt{H_{{\mathrm{2}}}} |}     \, ;\,  \varnothing  ]\,  \ottnt{a} }{%
{\ottdrulename{Small\_Var}}{}%
}}

\newcommand{\ottdruleSmallXXAppL}[1]{\ottdrule[#1]{%
\ottpremise{ [  \ottnt{H}  ]\,  \ottnt{a}  \Rightarrow_{  \ottnt{S}  \, \cup    \,\text{fv}\,  \ottnt{b}    }^{ { \color{black}{r} } } [  \ottnt{H'} \, ;\,  \mathbf{u}' \, ;\,  \Gamma  ]\,  \ottnt{a'} }%
}{
 [  \ottnt{H}  ]\,  \ottnt{a} \, \ottnt{b}  \Rightarrow_{ \ottnt{S} }^{ { \color{black}{r} } } [  \ottnt{H'} \, ;\,  \mathbf{u}' \, ;\,  \Gamma  ]\,  \ottnt{a'} \, \ottnt{b} }{%
{\ottdrulename{Small\_AppL}}{}%
}}

\newcommand{\ottdruleSmallXXAppBeta}[1]{\ottdrule[#1]{%
\ottpremise{  }%
\ottpremise{ \ottmv{x} \ \not\in    \mathsf{Var} \,  \ottnt{H}   \, \cup       \,\text{fv}\,  \ottnt{b}   \, \cup     \,\text{fv}\,  \ottnt{a}   -   \{  \ottmv{y}  \}       \, \cup  \ottnt{S}    }%
\ottpremise{\ottnt{a'}  \ottsym{=}  \ottnt{a}  \ottsym{\{}  \ottmv{x}  \ottsym{/}  \ottmv{y}  \ottsym{\}}}%
}{
 [  \ottnt{H}  ]\,  \ottsym{(}   \lambda \ottmv{y} \!:^ { \color{black}{q} } \! \ottnt{A_{{\mathrm{1}}}} . \ottnt{a}   \ottsym{)} \, \ottnt{b}  \Rightarrow_{ \ottnt{S} }^{ { \color{black}{r} } } [   \ottnt{H}  ,   \ottmv{x}  \overset{  { \color{black}{r} }  \cdot  { \color{black}{q} }  }{\mapsto} { \Gamma \vdash  \ottnt{b}  :  \ottnt{A} }   \, ;\,    \mathbf{0}^{| \ottnt{H} |}   \mathop{\diamond}    { \color{black}{0} }    \, ;\,   \ottmv{x} \! :^{  { \color{black}{r} }  \cdot  { \color{black}{q} }  }\! \ottnt{A}   ]\,  \ottnt{a'} }{%
{\ottdrulename{Small\_AppBeta}}{}%
}}

\newcommand{\ottdruleSmallXXUnitL}[1]{\ottdrule[#1]{%
\ottpremise{ [  \ottnt{H}  ]\,  \ottnt{a}  \Rightarrow_{  \ottnt{S}  \, \cup    \,\text{fv}\,  \ottnt{b}    }^{ { \color{black}{r} } } [  \ottnt{H'} \, ;\,  \mathbf{u}' \, ;\,  \Gamma  ]\,  \ottnt{a'} }%
}{
 [  \ottnt{H}  ]\,   \ottkw{let}\,  \ottkw{unit} \,=\, \ottnt{a} \ \ottkw{in}\  \ottnt{b}   \Rightarrow_{ \ottnt{S} }^{ { \color{black}{r} } } [  \ottnt{H'} \, ;\,  \mathbf{u}' \, ;\,  \Gamma  ]\,   \ottkw{let}\,  \ottkw{unit} \,=\, \ottnt{a'} \ \ottkw{in}\  \ottnt{b}  }{%
{\ottdrulename{Small\_UnitL}}{}%
}}

\newcommand{\ottdruleSmallXXUnitBeta}[1]{\ottdrule[#1]{%
}{
 [  \ottnt{H}  ]\,   \ottkw{let}\,  \ottkw{unit} \,=\, \ottkw{unit} \ \ottkw{in}\  \ottnt{b}   \Rightarrow_{ \ottnt{S} }^{ { \color{black}{r} } } [  \ottnt{H} \, ;\,   \mathbf{0}^{| \ottnt{H} |}  \, ;\,  \varnothing  ]\,  \ottnt{b} }{%
{\ottdrulename{Small\_UnitBeta}}{}%
}}

\newcommand{\ottdruleSmallXXLetBoxL}[1]{\ottdrule[#1]{%
\ottpremise{ [  \ottnt{H}  ]\,  \ottnt{a}  \Rightarrow_{  \ottnt{S}  \, \cup    \,\text{fv}\,  \ottnt{b}    }^{ { \color{black}{r} } } [  \ottnt{H'} \, ;\,  \mathbf{u}' \, ;\,  \Gamma  ]\,  \ottnt{a'} }%
}{
 [  \ottnt{H}  ]\,   \ottkw{let}\, \ottkw{box} \, \ottmv{x} \,=\, \ottnt{a} \ \ottkw{in}\  \ottnt{b}   \Rightarrow_{ \ottnt{S} }^{ { \color{black}{r} } } [  \ottnt{H'} \, ;\,  \mathbf{u}' \, ;\,  \Gamma  ]\,   \ottkw{let}\, \ottkw{box} \, \ottmv{x} \,=\, \ottnt{a'} \ \ottkw{in}\  \ottnt{b}  }{%
{\ottdrulename{Small\_LetBoxL}}{}%
}}

\newcommand{\ottdruleSmallXXLetBoxBeta}[1]{\ottdrule[#1]{%
\ottpremise{  }%
\ottpremise{ \ottmv{x} \ \not\in    \mathsf{Var} \,  \ottnt{H}   \, \cup       \,\text{fv}\,  \ottnt{a}   \, \cup     \,\text{fv}\,  \ottnt{b}   -   \{  \ottmv{y}  \}       \, \cup  \ottnt{S}    }%
\ottpremise{\ottnt{b'}  \ottsym{=}  \ottnt{b}  \ottsym{\{}  \ottmv{x}  \ottsym{/}  \ottmv{y}  \ottsym{\}}}%
}{
 [  \ottnt{H}  ]\,   \ottkw{let}\, \ottkw{box} \, \ottmv{y} \,=\,  \ottkw{box} _ { \color{black}{q} } \, \ottnt{a}  \ \ottkw{in}\  \ottnt{b}   \Rightarrow_{ \ottnt{S} }^{ { \color{black}{r} } } [   \ottnt{H}  ,   \ottmv{x}  \overset{  { \color{black}{r} }  \cdot  { \color{black}{q} }  }{\mapsto} { \Gamma \vdash  \ottnt{a}  :  \ottnt{A} }   \, ;\,    \mathbf{0}^{| \ottnt{H} |}   \mathop{\diamond}    { \color{black}{0} }    \, ;\,   \ottmv{x} \! :^{  { \color{black}{r} }  \cdot  { \color{black}{q} }  }\! \ottnt{A}   ]\,  \ottnt{b'} }{%
{\ottdrulename{Small\_LetBoxBeta}}{}%
}}

\newcommand{\ottdruleSmallXXCaseL}[1]{\ottdrule[#1]{%
\ottpremise{ [  \ottnt{H}  ]\,  \ottnt{a}  \Rightarrow_{  \ottnt{S}  \, \cup     \,\text{fv}\,  \ottnt{b_{{\mathrm{1}}}}   \, \cup   \,\text{fv}\,  \ottnt{b_{{\mathrm{2}}}}     }^{  { \color{black}{r} }  \cdot  { \color{black}{q} }  } [  \ottnt{H'} \, ;\,  \mathbf{u}' \, ;\,  \Gamma  ]\,  \ottnt{a'} }%
}{
 [  \ottnt{H}  ]\,   \ottkw{case}_ { \color{black}{q} } \,  \ottnt{a} \, \ottkw{of}\,  \ottnt{b_{{\mathrm{1}}}}  ;  \ottnt{b_{{\mathrm{2}}}}   \Rightarrow_{ \ottnt{S} }^{ { \color{black}{r} } } [  \ottnt{H'} \, ;\,  \mathbf{u}' \, ;\,  \Gamma  ]\,   \ottkw{case}_ { \color{black}{q} } \,  \ottnt{a'} \, \ottkw{of}\,  \ottnt{b_{{\mathrm{1}}}}  ;  \ottnt{b_{{\mathrm{2}}}}  }{%
{\ottdrulename{Small\_CaseL}}{}%
}}

\newcommand{\ottdruleSmallXXCaseOne}[1]{\ottdrule[#1]{%
}{
 [  \ottnt{H}  ]\,   \ottkw{case}_ { \color{black}{q} } \,  \ottsym{(}   \ottkw{inj}_1\,  \ottnt{a}   \ottsym{)} \, \ottkw{of}\,  \ottnt{b_{{\mathrm{1}}}}  ;  \ottnt{b_{{\mathrm{2}}}}   \Rightarrow_{ \ottnt{S} }^{ { \color{black}{r} } } [  \ottnt{H} \, ;\,   \mathbf{0}^{| \ottnt{H} |}  \, ;\,  \varnothing  ]\,  \ottnt{b_{{\mathrm{1}}}} \, \ottnt{a} }{%
{\ottdrulename{Small\_Case1}}{}%
}}

\newcommand{\ottdruleSmallXXCaseTwo}[1]{\ottdrule[#1]{%
}{
 [  \ottnt{H}  ]\,   \ottkw{case}_ { \color{black}{q} } \,  \ottsym{(}   \ottkw{inj}_2\,  \ottnt{a}   \ottsym{)} \, \ottkw{of}\,  \ottnt{b_{{\mathrm{1}}}}  ;  \ottnt{b_{{\mathrm{2}}}}   \Rightarrow_{ \ottnt{S} }^{ { \color{black}{r} } } [  \ottnt{H} \, ;\,   \mathbf{0}^{| \ottnt{H} |}  \, ;\,  \varnothing  ]\,  \ottnt{b_{{\mathrm{2}}}} \, \ottnt{a} }{%
{\ottdrulename{Small\_Case2}}{}%
}}

\newcommand{\ottdruleSmallXXSub}[1]{\ottdrule[#1]{%
\ottpremise{ [  \ottnt{H_{{\mathrm{1}}}}  ]\,  \ottnt{a}  \Rightarrow_{ \ottnt{S} }^{ { \color{black}{r} } } [  \ottnt{H'} \, ;\,  \mathbf{u}' \, ;\,  \Gamma  ]\,  \ottnt{a'} }%
\ottpremise{ \ottnt{H_{{\mathrm{1}}}}  \leq  \ottnt{H_{{\mathrm{2}}}} }%
}{
 [  \ottnt{H_{{\mathrm{2}}}}  ]\,  \ottnt{a}  \Rightarrow_{ \ottnt{S} }^{ { \color{black}{r} } } [  \ottnt{H'} \, ;\,  \mathbf{u}' \, ;\,  \Gamma  ]\,  \ottnt{a'} }{%
{\ottdrulename{Small\_Sub}}{}%
}}

\newcommand{\ottdruleSmallXXProjL}[1]{\ottdrule[#1]{%
\ottpremise{ [  \ottnt{H}  ]\,  \ottnt{a}  \Rightarrow_{  \ottnt{S}  \, \cup    \,\text{fv}\,  \ottnt{b}    }^{ { \color{black}{r} } } [  \ottnt{H'} \, ;\,  \mathbf{u}' \, ;\,  \Gamma  ]\,  \ottnt{a'} }%
}{
 [  \ottnt{H}  ]\,  \ottkw{let} \, \ottsym{(}  \ottmv{x}  \ottsym{,}  \ottmv{y}  \ottsym{)}  \ottsym{=}  \ottnt{a} \, \mathsf{in} \, \ottnt{b}  \Rightarrow_{ \ottnt{S} }^{ { \color{black}{r} } } [  \ottnt{H'} \, ;\,  \mathbf{u}' \, ;\,  \Gamma  ]\,  \ottkw{let} \, \ottsym{(}  \ottmv{x}  \ottsym{,}  \ottmv{y}  \ottsym{)}  \ottsym{=}  \ottnt{a'} \, \mathsf{in} \, \ottnt{b} }{%
{\ottdrulename{Small\_ProjL}}{}%
}}

\newcommand{\ottdruleSmallXXProjBeta}[1]{\ottdrule[#1]{%
\ottpremise{  }%
\ottpremise{  }%
\ottpremise{ \ottmv{x'} \ \not\in    \mathsf{Var} \,  \ottnt{H}   \, \cup       \,\text{fv}\,  \ottnt{a_{{\mathrm{1}}}}   \, \cup   \,\text{fv}\,  \ottnt{a_{{\mathrm{2}}}}     \, \cup         \,\text{fv}\,  \ottnt{b}   -   \{  \ottmv{x}  \}     -   \{  \ottmv{y}  \}     \, \cup  \ottnt{S}      }%
\ottpremise{ \ottmv{y'} \ \not\in    \mathsf{Var} \,  \ottnt{H}   \, \cup       \,\text{fv}\,  \ottnt{a_{{\mathrm{1}}}}   \, \cup   \,\text{fv}\,  \ottnt{a_{{\mathrm{2}}}}     \, \cup         \,\text{fv}\,  \ottnt{b}   -   \{  \ottmv{x}  \}     -   \{  \ottmv{y}  \}     \, \cup    \ottnt{S}  \, \cup   \{  \ottmv{x'}  \}         }%
\ottpremise{\ottnt{b'}  \ottsym{=}  \ottnt{b}  \ottsym{\{}  \ottmv{x'}  \ottsym{/}  \ottmv{x}  \ottsym{\}}  \ottsym{\{}  \ottmv{y'}  \ottsym{/}  \ottmv{y}  \ottsym{\}}}%
}{
 [  \ottnt{H}  ]\,  \ottkw{let} \, \ottsym{(}  \ottmv{x}  \ottsym{,}  \ottmv{y}  \ottsym{)}  \ottsym{=}  \ottsym{(}  \ottnt{a_{{\mathrm{1}}}}  \ottsym{,}  \ottnt{a_{{\mathrm{2}}}}  \ottsym{)} \, \mathsf{in} \, \ottnt{b}  \Rightarrow_{ \ottnt{S} }^{ { \color{black}{r} } } [   \ottnt{H}  ,     \ottmv{x'}  \overset{ { \color{black}{r} } }{\mapsto} { \Gamma_{{\mathrm{1}}} \vdash  \ottnt{a_{{\mathrm{1}}}}  :  \ottnt{A_{{\mathrm{1}}}} }   ,   \ottmv{y'}  \overset{ { \color{black}{r} } }{\mapsto} { \Gamma_{{\mathrm{2}}} \vdash  \ottnt{a_{{\mathrm{2}}}}  :  \ottnt{A_{{\mathrm{2}}}} }     \, ;\,    \mathbf{0}^{| \ottnt{H} |}   \mathop{\diamond}      { \color{black}{0} }    \mathop{\diamond}    { \color{black}{0} }      \, ;\,    \ottmv{x'} \! :^{ { \color{black}{r} } }\! \ottnt{A_{{\mathrm{1}}}}  ,   \ottmv{y'} \! :^{ { \color{black}{r} } }\! \ottnt{A_{{\mathrm{2}}}}    ]\,  \ottnt{b'} }{%
{\ottdrulename{Small\_ProjBeta}}{}%
}}

\newcommand{\ottdruleSmallXXSpreadL}[1]{\ottdrule[#1]{%
\ottpremise{ [  \ottnt{H}  ]\,  \ottnt{a}  \Rightarrow_{  \ottnt{S}  \, \cup     \,\text{fv}\,  \ottnt{b}   -   \{  \ottmv{y}  \}     }^{ { \color{black}{r} } } [  \ottnt{H'} \, ;\,  \mathbf{u}' \, ;\,  \Gamma  ]\,  \ottnt{a'} }%
}{
 [  \ottnt{H}  ]\,   \ottkw{spread}\,  \ottnt{a} \, \ottkw{to}\,  \ottmv{y} \, \ottkw{in}\,  \ottnt{b}   \Rightarrow_{ \ottnt{S} }^{ { \color{black}{r} } } [  \ottnt{H'} \, ;\,  \mathbf{u}' \, ;\,  \Gamma  ]\,   \ottkw{spread}\,  \ottnt{a'} \, \ottkw{to}\,  \ottmv{y} \, \ottkw{in}\,  \ottnt{b}  }{%
{\ottdrulename{Small\_SpreadL}}{}%
}}

\newcommand{\ottdruleSmallXXSpread}[1]{\ottdrule[#1]{%
\ottpremise{ \ottmv{x} \ \not\in    \mathsf{Var} \,  \ottnt{H}   \, \cup       \,\text{fv}\,  \ottnt{a_{{\mathrm{1}}}}   \, \cup   \,\text{fv}\,  \ottnt{a_{{\mathrm{2}}}}     \, \cup       \,\text{fv}\,  \ottnt{b}   -   \{  \ottmv{y}  \}     \, \cup  \ottnt{S}      }%
\ottpremise{\ottnt{b'}  \ottsym{=}  \ottnt{b}  \ottsym{\{}  \ottmv{x}  \ottsym{/}  \ottmv{y}  \ottsym{\}}}%
}{
 [  \ottnt{H}  ]\,   \ottkw{spread}\,  \ottsym{(}  \ottnt{a_{{\mathrm{1}}}}  \ottsym{,}  \ottnt{a_{{\mathrm{2}}}}  \ottsym{)} \, \ottkw{to}\,  \ottmv{y} \, \ottkw{in}\,  \ottnt{b}   \Rightarrow_{ \ottnt{S} }^{ { \color{black}{r} } } [   \ottnt{H}  ,   \ottmv{x}  \overset{  { \color{black}{r} }  \cdot  { \color{black}{q} }  }{\mapsto} { \Gamma \vdash  \ottnt{a_{{\mathrm{1}}}}  :  \ottnt{A_{{\mathrm{1}}}} }   \, ;\,    \mathbf{0}^{| \ottnt{H} |}   \mathop{\diamond}    { \color{black}{0} }    \, ;\,   \ottmv{x} \! :^{  { \color{black}{r} }  \cdot  { \color{black}{q} }  }\! \ottnt{A_{{\mathrm{1}}}}   ]\,  \ottnt{b'} \, \ottnt{a_{{\mathrm{2}}}} }{%
{\ottdrulename{Small\_Spread}}{}%
}}

\newcommand{\ottdruleSmallXXDAppBeta}[1]{\ottdrule[#1]{%
\ottpremise{  }%
\ottpremise{ \ottmv{x} \ \not\in    \mathsf{Var} \,  \ottnt{H}   \, \cup       \,\text{fv}\,  \ottnt{b}   \, \cup     \,\text{fv}\,  \ottnt{a}   -   \{  \ottmv{y}  \}       \, \cup  \ottnt{S}    }%
\ottpremise{\ottnt{a'}  \ottsym{=}  \ottnt{a}  \ottsym{\{}  \ottmv{x}  \ottsym{/}  \ottmv{y}  \ottsym{\}}}%
}{
 [  \ottnt{H}  ]\,  \ottsym{(}   \lambda \ottmv{y} \!:^ { \color{black}{q} } \! \ottnt{A_{{\mathrm{1}}}} . \ottnt{a}   \ottsym{)} \, \ottnt{b}  \Rightarrow_{ \ottnt{S} }^{ { \color{black}{r} } } [   \ottnt{H}  ,   \ottmv{x}  \overset{  { \color{black}{r} }  \cdot  { \color{black}{q} }  }{\mapsto} { \Gamma \vdash  \ottnt{b}  :  \ottnt{A} }   \, ;\,    \mathbf{0}^{| \ottnt{H} |}   \mathop{\diamond}    { \color{black}{0} }    \, ;\,   \ottmv{x}  \! = \!  \ottnt{b}  \! :^{  { \color{black}{r} }  \cdot  { \color{black}{q} }  } \!  \ottnt{A}   ]\,  \ottnt{a'} }{%
{\ottdrulename{Small\_DAppBeta}}{}%
}}

\newcommand{\ottdefnSmallStep}[1]{\begin{ottdefnblock}[#1]{$ [  \ottnt{H}  ]\,  \ottnt{a}  \Rightarrow_{ \ottnt{S} }^{ { \color{black}{r} } } [  \ottnt{H'} \, ;\,  \mathbf{u}' \, ;\,  \Gamma  ]\,  \ottnt{b} $}{\ottcom{Small-step Reduction}}
\ottusedrule{\ottdruleSmallXXVar{}}
\ottusedrule{\ottdruleSmallXXAppL{}}
\ottusedrule{\ottdruleSmallXXAppBeta{}}
\ottusedrule{\ottdruleSmallXXUnitL{}}
\ottusedrule{\ottdruleSmallXXUnitBeta{}}
\ottusedrule{\ottdruleSmallXXLetBoxL{}}
\ottusedrule{\ottdruleSmallXXLetBoxBeta{}}
\ottusedrule{\ottdruleSmallXXCaseL{}}
\ottusedrule{\ottdruleSmallXXCaseOne{}}
\ottusedrule{\ottdruleSmallXXCaseTwo{}}
\ottusedrule{\ottdruleSmallXXSub{}}
\ottusedrule{\ottdruleSmallXXProjL{}}
\ottusedrule{\ottdruleSmallXXProjBeta{}}
\ottusedrule{\ottdruleSmallXXSpreadL{}}
\ottusedrule{\ottdruleSmallXXSpread{}}
\ottusedrule{\ottdruleSmallXXDAppBeta{}}
\end{ottdefnblock}}

\newcommand{\ottdefnsJSmallStep}{
\ottdefnSmallStep{}}

% defns JSmallStepSecond
%% defn SmallStepSecond
\newcommand{\ottdruleSmallSXXVar}[1]{\ottdrule[#1]{%
\ottpremise{ { \color{black}{1} }   \leq  { \color{black}{r} }}%
}{
 [   \ottnt{H_{{\mathrm{1}}}}  ,     \ottmv{x}  \overset{ \ottsym{(}  { \color{black}{q} }  \ottsym{+}  { \color{black}{r} }  \ottsym{)} }{\mapsto} { \Gamma \vdash  \ottnt{a}  :  \ottnt{A} }   ,  \ottnt{H_{{\mathrm{2}}}}     ]\,  \ottmv{x}  \! :  \ottnt{A}  \Rightarrow_{ \ottnt{S} } [   \ottnt{H_{{\mathrm{1}}}}  ,     \ottmv{x}  \overset{ { \color{black}{q} } }{\mapsto} { \Gamma \vdash  \ottnt{a}  :  \ottnt{A} }   ,  \ottnt{H_{{\mathrm{2}}}}    \, ;\,    \mathbf{0}^{| \ottnt{H_{{\mathrm{1}}}} |}   \mathop{\diamond}     { \color{black}{r} }   \mathop{\diamond}   \mathbf{0}^{| \ottnt{H_{{\mathrm{2}}}} |}     \, ;\,  \varnothing  ]\,  \ottnt{a}  \! :  \ottnt{A} }{%
{\ottdrulename{SmallS\_Var}}{}%
}}

\newcommand{\ottdruleSmallSXXAppL}[1]{\ottdrule[#1]{%
\ottpremise{ [  \ottnt{H}  ]\,  \ottnt{a}  \! :  \ottnt{A}  \Rightarrow_{  \ottnt{S}  \, \cup    \,\text{fv}\,  \ottnt{b}    } [  \ottnt{H'} \, ;\,  \mathbf{u}' \, ;\,  \Gamma  ]\,  \ottnt{a'}  \! :  \ottnt{B} }%
}{
 [  \ottnt{H}  ]\,  \ottnt{a} \, \ottnt{b}  \! :  \ottnt{A}  \Rightarrow_{ \ottnt{S} } [  \ottnt{H'} \, ;\,  \mathbf{u}' \, ;\,  \Gamma  ]\,  \ottnt{a'} \, \ottnt{b}  \! :  \ottnt{B} }{%
{\ottdrulename{SmallS\_AppL}}{}%
}}

\newcommand{\ottdruleSmallSXXApp}[1]{\ottdrule[#1]{%
\ottpremise{ \ottmv{x} \ \not\in    \mathsf{dom} \,  \ottnt{H}   \, \cup     \,\text{fv}\,  \ottnt{b}   \, \cup  \ottnt{S}    }%
}{
 [  \ottnt{H}  ]\,  \ottsym{(}   \lambda \ottmv{y} \!:^ { \color{black}{q} } \! \ottnt{A} . \ottnt{a}   \ottsym{)} \, \ottnt{b}  \! :  \ottnt{A}  \Rightarrow_{ \ottnt{S} } [   \ottnt{H}  ,   \ottmv{x}  \overset{  { \color{black}{q} }_{{\mathrm{1}}}  \cdot   { \color{black}{q} }  \cdot  { \color{black}{q} }_{{\mathrm{2}}}   }{\mapsto} { \Gamma \vdash  \ottnt{b}  :  \ottnt{A} }   \, ;\,    \mathbf{0}^{| \ottnt{H} |}   \mathop{\diamond}    { \color{black}{0} }    \, ;\,   \ottmv{x} \! :^{  { \color{black}{q} }_{{\mathrm{1}}}  \cdot   { \color{black}{q} }  \cdot  { \color{black}{q} }_{{\mathrm{2}}}   }\! \ottnt{A}   ]\,  \ottnt{a}  \ottsym{\{}  \ottmv{x}  \ottsym{/}  \ottmv{y}  \ottsym{\}}  \! :  \ottnt{A}  \ottsym{\{}  \ottmv{x}  \ottsym{/}  \ottmv{y}  \ottsym{\}} }{%
{\ottdrulename{SmallS\_App}}{}%
}}

\newcommand{\ottdruleSmallSXXUnitL}[1]{\ottdrule[#1]{%
\ottpremise{ [  \ottnt{H}  ]\,  \ottnt{a}  \! :  \ottnt{A}  \Rightarrow_{  \ottnt{S}  \, \cup    \,\text{fv}\,  \ottnt{b}    } [  \ottnt{H'} \, ;\,  \mathbf{u}' \, ;\,  \Gamma  ]\,  \ottnt{a'}  \! :  \ottnt{B} }%
}{
 [  \ottnt{H}  ]\,   \ottkw{let}\,  \ottkw{unit} \,=\, \ottnt{a} \ \ottkw{in}\  \ottnt{b}   \! :  \ottnt{A}  \Rightarrow_{ \ottnt{S} } [  \ottnt{H'} \, ;\,  \mathbf{u}' \, ;\,  \Gamma  ]\,   \ottkw{let}\,  \ottkw{unit} \,=\, \ottnt{a'} \ \ottkw{in}\  \ottnt{b}   \! :  \ottnt{B} }{%
{\ottdrulename{SmallS\_UnitL}}{}%
}}

\newcommand{\ottdruleSmallSXXUnitE}[1]{\ottdrule[#1]{%
}{
 [  \ottnt{H}  ]\,   \ottkw{let}\,  \ottkw{unit} \,=\, \ottkw{unit} \ \ottkw{in}\  \ottnt{b}   \! :  \ottnt{A}  \Rightarrow_{ \ottnt{S} } [  \ottnt{H} \, ;\,   \mathbf{0}^{| \ottnt{H} |}  \, ;\,  \varnothing  ]\,  \ottnt{b}  \! :  \ottnt{A} }{%
{\ottdrulename{SmallS\_UnitE}}{}%
}}

\newcommand{\ottdruleSmallSXXLetBoxL}[1]{\ottdrule[#1]{%
\ottpremise{ [  \ottnt{H}  ]\,  \ottnt{a}  \! :  \ottnt{A}  \Rightarrow_{  \ottnt{S}  \, \cup    \,\text{fv}\,  \ottnt{b}    } [  \ottnt{H'} \, ;\,  \mathbf{u}' \, ;\,  \Gamma  ]\,  \ottnt{a'}  \! :  \ottnt{B} }%
}{
 [  \ottnt{H}  ]\,   \ottkw{let}\, \ottkw{box} \, \ottmv{x} \,=\, \ottnt{a} \ \ottkw{in}\  \ottnt{b}   \! :  \ottnt{A}  \Rightarrow_{ \ottnt{S} } [  \ottnt{H'} \, ;\,  \mathbf{u}' \, ;\,  \Gamma  ]\,   \ottkw{let}\, \ottkw{box} \, \ottmv{x} \,=\, \ottnt{a'} \ \ottkw{in}\  \ottnt{b}   \! :  \ottnt{B} }{%
{\ottdrulename{SmallS\_LetBoxL}}{}%
}}

\newcommand{\ottdruleSmallSXXLetBox}[1]{\ottdrule[#1]{%
\ottpremise{  }%
\ottpremise{ \ottmv{x} \ \not\in     \mathsf{dom} \,  \ottnt{H}   \, \cup     \,\text{fv}\,  \ottnt{b}   \, \cup  \ottnt{S}     }%
\ottpremise{\ottnt{b'}  \ottsym{=}  \ottnt{b}  \ottsym{\{}  \ottmv{x}  \ottsym{/}  \ottmv{y}  \ottsym{\}}}%
\ottpremise{\ottnt{B'}  \ottsym{=}  \ottnt{B}  \ottsym{\{}  \ottmv{x}  \ottsym{/}  \ottmv{y}  \ottsym{\}}}%
}{
 [  \ottnt{H}  ]\,   \ottkw{let}\, \ottkw{box} \, \ottmv{y} \,=\,  \ottkw{box} _ { \color{black}{q} } \, \ottnt{a}  \ \ottkw{in}\  \ottnt{b}   \! :  \ottnt{B}  \Rightarrow_{ \ottnt{S} } [   \ottnt{H}  ,   \ottmv{x}  \overset{ { \color{black}{q} }' }{\mapsto} { \Gamma \vdash  \ottnt{a}  :  \ottnt{A} }   \, ;\,    \mathbf{0}^{| \ottnt{H} |}   \mathop{\diamond}    { \color{black}{0} }    \, ;\,   \ottmv{x} \! :^{ { \color{black}{q} }' }\! \ottnt{A}   ]\,  \ottnt{b'}  \! :  \ottnt{B'} }{%
{\ottdrulename{SmallS\_LetBox}}{}%
}}

\newcommand{\ottdruleSmallSXXCaseL}[1]{\ottdrule[#1]{%
\ottpremise{ [  \ottnt{H}  ]\,  \ottnt{a}  \! :  \ottnt{A}  \Rightarrow_{  \ottnt{S}  \, \cup     \,\text{fv}\,  \ottnt{b_{{\mathrm{1}}}}   \, \cup   \,\text{fv}\,  \ottnt{b_{{\mathrm{2}}}}     } [  \ottnt{H'} \, ;\,  \mathbf{u}' \, ;\,  \Gamma  ]\,  \ottnt{a'}  \! :  \ottnt{A'} }%
}{
 [  \ottnt{H}  ]\,   \ottkw{case}_ { \color{black}{q} } \,  \ottnt{a} \, \ottkw{of}\,  \ottnt{b_{{\mathrm{1}}}}  ;  \ottnt{b_{{\mathrm{2}}}}   \! :  \ottnt{A}  \Rightarrow_{ \ottnt{S} } [  \ottnt{H'} \, ;\,  \mathbf{u}' \, ;\,  \Gamma  ]\,   \ottkw{case}_ { \color{black}{q} } \,  \ottnt{a'} \, \ottkw{of}\,  \ottnt{b_{{\mathrm{1}}}}  ;  \ottnt{b_{{\mathrm{2}}}}   \! :  \ottnt{A'} }{%
{\ottdrulename{SmallS\_CaseL}}{}%
}}

\newcommand{\ottdruleSmallSXXCaseOne}[1]{\ottdrule[#1]{%
}{
 [  \ottnt{H}  ]\,   \ottkw{case}_ { \color{black}{q} } \,  \ottsym{(}   \ottkw{inj}_1\,  \ottnt{a}   \ottsym{)} \, \ottkw{of}\,  \ottnt{b_{{\mathrm{1}}}}  ;  \ottnt{b_{{\mathrm{2}}}}   \! :  \ottnt{A}  \Rightarrow_{ \ottnt{S} } [  \ottnt{H} \, ;\,   \mathbf{0}^{| \ottnt{H} |}  \, ;\,  \varnothing  ]\,  \ottnt{b_{{\mathrm{1}}}} \, \ottnt{a}  \! :  \ottnt{A} }{%
{\ottdrulename{SmallS\_Case1}}{}%
}}

\newcommand{\ottdruleSmallSXXCaseTwo}[1]{\ottdrule[#1]{%
}{
 [  \ottnt{H}  ]\,   \ottkw{case}_ { \color{black}{q} } \,  \ottsym{(}   \ottkw{inj}_2\,  \ottnt{a}   \ottsym{)} \, \ottkw{of}\,  \ottnt{b_{{\mathrm{1}}}}  ;  \ottnt{b_{{\mathrm{2}}}}   \! :  \ottnt{A}  \Rightarrow_{ \ottnt{S} } [  \ottnt{H} \, ;\,   \mathbf{0}^{| \ottnt{H} |}  \, ;\,  \varnothing  ]\,  \ottnt{b_{{\mathrm{2}}}} \, \ottnt{a}  \! :  \ottnt{A} }{%
{\ottdrulename{SmallS\_Case2}}{}%
}}

\newcommand{\ottdruleSmallSXXProjL}[1]{\ottdrule[#1]{%
\ottpremise{ [  \ottnt{H}  ]\,  \ottnt{a}  \! :  \ottnt{A}  \Rightarrow_{  \ottnt{S}  \, \cup    \,\text{fv}\,  \ottnt{b}    } [  \ottnt{H'} \, ;\,  \mathbf{u}' \, ;\,  \Gamma  ]\,  \ottnt{a'}  \! :  \ottnt{B} }%
}{
 [  \ottnt{H}  ]\,  \ottkw{let} \, \ottsym{(}  \ottmv{x}  \ottsym{,}  \ottmv{y}  \ottsym{)}  \ottsym{=}  \ottnt{a} \, \mathsf{in} \, \ottnt{b}  \! :  \ottnt{A}  \Rightarrow_{ \ottnt{S} } [  \ottnt{H'} \, ;\,  \mathbf{u}' \, ;\,  \Gamma  ]\,  \ottkw{let} \, \ottsym{(}  \ottmv{x}  \ottsym{,}  \ottmv{y}  \ottsym{)}  \ottsym{=}  \ottnt{a'} \, \mathsf{in} \, \ottnt{b}  \! :  \ottnt{B} }{%
{\ottdrulename{SmallS\_ProjL}}{}%
}}

\newcommand{\ottdruleSmallSXXProjBeta}[1]{\ottdrule[#1]{%
\ottpremise{  }%
\ottpremise{  }%
\ottpremise{ \ottmv{x'} \ \not\in     \mathsf{dom} \,  \ottnt{H}   \, \cup     \,\text{fv}\,  \ottnt{b}   \, \cup  \ottnt{S}     }%
\ottpremise{ \ottmv{y'} \ \not\in     \mathsf{dom} \,  \ottnt{H}   \, \cup     \,\text{fv}\,  \ottnt{b}   \, \cup     \{  \ottmv{x'}  \}   \, \cup  \ottnt{S}       }%
\ottpremise{\ottnt{b'}  \ottsym{=}  \ottnt{b}  \ottsym{\{}  \ottmv{x'}  \ottsym{/}  \ottmv{x}  \ottsym{\}}  \ottsym{\{}  \ottmv{y'}  \ottsym{/}  \ottmv{y}  \ottsym{\}}}%
\ottpremise{\ottnt{B'}  \ottsym{=}  \ottnt{B}  \ottsym{\{}  \ottmv{x'}  \ottsym{/}  \ottmv{x}  \ottsym{\}}  \ottsym{\{}  \ottmv{y'}  \ottsym{/}  \ottmv{y}  \ottsym{\}}}%
}{
 [  \ottnt{H}  ]\,  \ottkw{let} \, \ottsym{(}  \ottmv{x}  \ottsym{,}  \ottmv{y}  \ottsym{)}  \ottsym{=}  \ottsym{(}  \ottnt{a_{{\mathrm{1}}}}  \ottsym{,}  \ottnt{a_{{\mathrm{2}}}}  \ottsym{)} \, \mathsf{in} \, \ottnt{b}  \! :  \ottnt{B}  \Rightarrow_{ \ottnt{S} } [   \ottnt{H}  ,     \ottmv{x'}  \overset{ { \color{black}{q} } }{\mapsto} { \Gamma_{{\mathrm{1}}} \vdash  \ottnt{a_{{\mathrm{1}}}}  :  \ottnt{A_{{\mathrm{1}}}} }   ,   \ottmv{y'}  \overset{ { \color{black}{q} } }{\mapsto} { \Gamma_{{\mathrm{2}}} \vdash  \ottnt{a_{{\mathrm{2}}}}  :  \ottnt{A_{{\mathrm{2}}}} }     \, ;\,    \mathbf{0}^{| \ottnt{H} |}   \mathop{\diamond}      { \color{black}{0} }    \mathop{\diamond}    { \color{black}{0} }      \, ;\,    \ottmv{x'} \! :^{ { \color{black}{q} } }\! \ottnt{A_{{\mathrm{1}}}}  ,   \ottmv{y'} \! :^{ { \color{black}{q} } }\! \ottnt{A_{{\mathrm{2}}}}    ]\,  \ottnt{b'}  \! :  \ottnt{B'} }{%
{\ottdrulename{SmallS\_ProjBeta}}{}%
}}

\newcommand{\ottdruleSmallSXXSub}[1]{\ottdrule[#1]{%
\ottpremise{ [  \ottnt{H_{{\mathrm{1}}}}  ]\,  \ottnt{a}  \! :  \ottnt{A}  \Rightarrow_{ \ottnt{S} } [  \ottnt{H'} \, ;\,  \mathbf{u}' \, ;\,  \Gamma  ]\,  \ottnt{a'}  \! :  \ottnt{B} }%
\ottpremise{ \ottnt{H_{{\mathrm{1}}}}  \leq  \ottnt{H_{{\mathrm{2}}}} }%
}{
 [  \ottnt{H_{{\mathrm{2}}}}  ]\,  \ottnt{a}  \! :  \ottnt{A}  \Rightarrow_{ \ottnt{S} } [  \ottnt{H'} \, ;\,  \mathbf{u}' \, ;\,  \Gamma  ]\,  \ottnt{a'}  \! :  \ottnt{B} }{%
{\ottdrulename{SmallS\_Sub}}{}%
}}

\newcommand{\ottdefnSmallStepSecond}[1]{\begin{ottdefnblock}[#1]{$ [  \ottnt{H}  ]\,  \ottnt{a}  \! :  \ottnt{A}  \Rightarrow_{ \ottnt{S} } [  \ottnt{H'} \, ;\,  \mathbf{u}' \, ;\,  \Gamma  ]\,  \ottnt{b}  \! :  \ottnt{B} $}{\ottcom{Small-step Reduction}}
\ottusedrule{\ottdruleSmallSXXVar{}}
\ottusedrule{\ottdruleSmallSXXAppL{}}
\ottusedrule{\ottdruleSmallSXXApp{}}
\ottusedrule{\ottdruleSmallSXXUnitL{}}
\ottusedrule{\ottdruleSmallSXXUnitE{}}
\ottusedrule{\ottdruleSmallSXXLetBoxL{}}
\ottusedrule{\ottdruleSmallSXXLetBox{}}
\ottusedrule{\ottdruleSmallSXXCaseL{}}
\ottusedrule{\ottdruleSmallSXXCaseOne{}}
\ottusedrule{\ottdruleSmallSXXCaseTwo{}}
\ottusedrule{\ottdruleSmallSXXProjL{}}
\ottusedrule{\ottdruleSmallSXXProjBeta{}}
\ottusedrule{\ottdruleSmallSXXSub{}}
\end{ottdefnblock}}

\newcommand{\ottdefnsJSmallStepSecond}{
\ottdefnSmallStepSecond{}}

% defns JMultiStep
%% defn MultiStep
\newcommand{\ottdruleMultiXXOne}[1]{\ottdrule[#1]{%
\ottpremise{ [  \ottnt{H}  ]\,  \ottnt{a}  \Rightarrow_{ \ottnt{S} }^{ { \color{black}{r} } } [  \ottnt{H'} \, ;\,  \mathbf{u}' \, ;\,  \Gamma  ]\,  \ottnt{b} }%
}{
 [  \ottnt{H}  ]\,  \ottnt{a}  \Rightarrow\!\!\!\!\!\Rightarrow^{ { \color{black}{r} } }_{ \ottnt{S} } [  \ottnt{H'} \, ;\,  \mathbf{u}' \, ;\,  \Gamma  ]\,  \ottnt{b} }{%
{\ottdrulename{Multi\_One}}{}%
}}

\newcommand{\ottdruleMultiXXMany}[1]{\ottdrule[#1]{%
\ottpremise{ [  \ottnt{H}  ]\,  \ottnt{a}  \Rightarrow_{ \ottnt{S} }^{ { \color{black}{r} } } [  \ottnt{H'} \, ;\,  \mathbf{u}' \, ;\,  \Gamma_{{\mathrm{1}}}  ]\,  \ottnt{b_{{\mathrm{1}}}} }%
\ottpremise{ [  \ottnt{H'}  ]\,  \ottnt{b_{{\mathrm{1}}}}  \Rightarrow\!\!\!\!\!\Rightarrow^{ { \color{black}{r} } }_{ \ottnt{S} } [  \ottnt{H''} \, ;\,  \mathbf{u}'' \, ;\,  \Gamma_{{\mathrm{2}}}  ]\,  \ottnt{b} }%
}{
 [  \ottnt{H}  ]\,  \ottnt{a}  \Rightarrow\!\!\!\!\!\Rightarrow^{ { \color{black}{r} } }_{ \ottnt{S} } [  \ottnt{H''} \, ;\,     \mathbf{u}'  \mathop{\diamond}   \mathbf{0}^{| \Gamma_{{\mathrm{2}}} |}     +  \mathbf{u}''  \, ;\,   \Gamma_{{\mathrm{1}}} ,  \Gamma_{{\mathrm{2}}}   ]\,  \ottnt{b} }{%
{\ottdrulename{Multi\_Many}}{}%
}}

\newcommand{\ottdefnMultiStep}[1]{\begin{ottdefnblock}[#1]{$ [  \ottnt{H}  ]\,  \ottnt{a}  \Rightarrow\!\!\!\!\!\Rightarrow^{ { \color{black}{r} } }_{ \ottnt{S} } [  \ottnt{H'} \, ;\,  \mathbf{u}' \, ;\,  \Gamma  ]\,  \ottnt{b} $}{\ottcom{Multi-step Reduction}}
\ottusedrule{\ottdruleMultiXXOne{}}
\ottusedrule{\ottdruleMultiXXMany{}}
\end{ottdefnblock}}

\newcommand{\ottdefnsJMultiStep}{
\ottdefnMultiStep{}}

% defns JCompat
%% defn Compatibility
\newcommand{\ottdruleCompatXXEmpty}[1]{\ottdrule[#1]{%
}{
 \varnothing  \vdash  \varnothing ;  \varnothing }{%
{\ottdrulename{Compat\_Empty}}{}%
}}

\newcommand{\ottdruleCompatXXCons}[1]{\ottdrule[#1]{%
\ottpremise{ \ottnt{H}  \vdash  \Delta ;  \Gamma_{{\mathrm{1}}}  \ottsym{+}  \ottsym{(}   { \color{black}{q} }   \cdot   \Gamma_{{\mathrm{2}}}   \ottsym{)} }%
\ottpremise{ \Delta  ;  \Gamma_{{\mathrm{2}}}  \vdash \ottnt{a} : \ottnt{A} }%
\ottpremise{\ottmv{x} \, \not\in \, \mathsf{dom} \, \ottnt{H}}%
}{
  \ottnt{H}  ,   \ottmv{x}  \overset{ { \color{black}{q} } }{\mapsto} { \Gamma_{{\mathrm{2}}} \vdash  \ottnt{a}  :  \ottnt{A} }    \vdash   \Delta ,   \ottmv{x} \!\!:\!\! \ottnt{A}   ;   \Gamma_{{\mathrm{1}}} ,   \ottmv{x} \! :^{ { \color{black}{q} } }\! \ottnt{A}   }{%
{\ottdrulename{Compat\_Cons}}{}%
}}

\newcommand{\ottdruleCompatXXConsDef}[1]{\ottdrule[#1]{%
\ottpremise{ \ottnt{H}  \vdash  \Delta ;  \Gamma_{{\mathrm{1}}}  \ottsym{+}  \ottsym{(}   { \color{black}{q} }   \cdot   \Gamma_{{\mathrm{2}}}   \ottsym{)} }%
\ottpremise{ \Delta  ;  \Gamma_{{\mathrm{2}}}  \vdash \ottnt{a} : \ottnt{A} }%
\ottpremise{\ottmv{x} \, \not\in \, \mathsf{dom} \, \ottnt{H}}%
}{
  \ottnt{H}  ,   \ottmv{x}  \overset{ { \color{black}{q} } }{\mapsto} { \Gamma_{{\mathrm{2}}} \vdash  \ottnt{a}  :  \ottnt{A} }    \vdash   \Delta ,   \ottmv{x} \! = \!  \ottnt{a}  \! : \!  \ottnt{A}   ;   \Gamma_{{\mathrm{1}}} ,   \ottmv{x}  \! = \!  \ottnt{a}  \! :^{ { \color{black}{q} } } \!  \ottnt{A}   }{%
{\ottdrulename{Compat\_ConsDef}}{}%
}}

\newcommand{\ottdefnCompatibility}[1]{\begin{ottdefnblock}[#1]{$ \ottnt{H}  \vdash  \Delta ;  \Gamma $}{\ottcom{Compatibility}}
\ottusedrule{\ottdruleCompatXXEmpty{}}
\ottusedrule{\ottdruleCompatXXCons{}}
\ottusedrule{\ottdruleCompatXXConsDef{}}
\end{ottdefnblock}}

\newcommand{\ottdefnsJCompat}{
\ottdefnCompatibility{}}

% defns JExtra
%% defn ATyping
\newcommand{\ottdruleTXXUnitElim}[1]{\ottdrule[#1]{%
\ottpremise{ \Delta  ;  \Gamma_{{\mathrm{1}}}  \vdash \ottnt{a} : \ottkw{Unit} }%
\ottpremise{ \Delta  ;  \Gamma_{{\mathrm{2}}}  \vdash \ottnt{b} : \ottnt{B}  \ottsym{\{}  \ottkw{unit}  \ottsym{/}  \ottmv{y}  \ottsym{\}} }%
\ottpremise{  \Delta ,   \ottmv{y} \!\!:\!\! \ottkw{Unit}    ;   \Gamma_{{\mathrm{3}}} ,   \ottmv{y} \! :^{ { \color{black}{r} } }\! \ottkw{Unit}    \vdash \ottnt{B} :  \textbf{Type}  }%
}{
\Delta  \mathsf{;}  \Gamma_{{\mathrm{1}}}  \ottsym{+}  \Gamma_{{\mathrm{2}}}  \vdash   \ottkw{let}\,  \ottkw{unit} \,=\, \ottnt{a} \ \ottkw{in}\  \ottnt{b}   \ottsym{:}  \ottnt{B}  \ottsym{\{}  \ottnt{a}  \ottsym{/}  \ottmv{y}  \ottsym{\}}}{%
{\ottdrulename{T\_UnitElim}}{}%
}}

\newcommand{\ottdruleTXXSumElim}[1]{\ottdrule[#1]{%
\ottpremise{ { \color{black}{1} }   \leq  { \color{black}{q} }}%
\ottpremise{ \Delta  ;  \Gamma_{{\mathrm{1}}}  \vdash \ottnt{a} :  \ottnt{A_{{\mathrm{1}}}}  \oplus  \ottnt{A_{{\mathrm{2}}}}  }%
\ottpremise{ \Delta  ;  \Gamma_{{\mathrm{2}}}  \vdash \ottnt{b_{{\mathrm{1}}}} :  \Pi  \ottmv{x} \!:^ { \color{black}{q} } \! \ottnt{A_{{\mathrm{1}}}} .  \ottnt{B}  \ottsym{\{}   \ottkw{inj}_1\,  \ottmv{x}   \ottsym{/}  \ottmv{y}  \ottsym{\}}   }%
\ottpremise{ \Delta  ;  \Gamma_{{\mathrm{2}}}  \vdash \ottnt{b_{{\mathrm{2}}}} :  \Pi  \ottmv{x} \!:^ { \color{black}{q} } \! \ottnt{A_{{\mathrm{2}}}} .  \ottnt{B}  \ottsym{\{}   \ottkw{inj}_2\,  \ottmv{x}   \ottsym{/}  \ottmv{y}  \ottsym{\}}   }%
\ottpremise{  \Delta ,   \ottmv{y} \!\!:\!\!  \ottnt{A_{{\mathrm{1}}}}  \oplus  \ottnt{A_{{\mathrm{2}}}}     ;   \Gamma_{{\mathrm{3}}} ,   \ottmv{y} \! :^{ { \color{black}{r} } }\!  \ottnt{A_{{\mathrm{1}}}}  \oplus  \ottnt{A_{{\mathrm{2}}}}     \vdash \ottnt{B} :  \textbf{Type}  }%
}{
\Delta  \mathsf{;}    { \color{black}{q} }   \cdot   \Gamma_{{\mathrm{1}}}    \ottsym{+}  \Gamma_{{\mathrm{2}}}  \vdash    \ottkw{case}_ { \color{black}{q} } \,  \ottnt{a} \, \ottkw{of}\,  \ottnt{b_{{\mathrm{1}}}}  ;  \ottnt{b_{{\mathrm{2}}}}    \ottsym{:}   \ottnt{B}  \ottsym{\{}  \ottnt{a}  \ottsym{/}  \ottmv{y}  \ottsym{\}} }{%
{\ottdrulename{T\_SumElim}}{}%
}}

\newcommand{\ottdruleTXXSpreadElim}[1]{\ottdrule[#1]{%
\ottpremise{ \Delta  ;  \Gamma_{{\mathrm{1}}}  \vdash \ottnt{a} :  \Sigma  \ottmv{x} \!\!:^ { \color{black}{q} } \!\! \ottnt{A_{{\mathrm{1}}}} . \ottnt{A_{{\mathrm{2}}}}  }%
\ottpremise{    \Delta ,   \ottmv{x} \!\!:\!\! \ottnt{A_{{\mathrm{1}}}}    ,   \ottmv{y} \!\!:\!\! \ottnt{A_{{\mathrm{2}}}}    ;     \Gamma_{{\mathrm{2}}} ,   \ottmv{x} \! :^{ { \color{black}{q} } }\! \ottnt{A_{{\mathrm{1}}}}    ,   \ottmv{y} \! :^{  { \color{black}{1} }  }\! \ottnt{A_{{\mathrm{2}}}}    \vdash \ottnt{b} : \ottnt{B}  \ottsym{\{}  \ottsym{(}  \ottmv{x}  \ottsym{,}  \ottmv{y}  \ottsym{)}  \ottsym{/}  \ottmv{z}  \ottsym{\}} }%
\ottpremise{  \Delta ,   \ottmv{z} \!\!:\!\!  \Sigma  \ottmv{x} \!\!:^ { \color{black}{q} } \!\! \ottnt{A_{{\mathrm{1}}}} . \ottnt{A_{{\mathrm{2}}}}     ;   \Gamma_{{\mathrm{3}}} ,   \ottmv{z} \! :^{ { \color{black}{r} } }\!  \Sigma  \ottmv{x} \!\!:^ { \color{black}{q} } \!\! \ottnt{A_{{\mathrm{1}}}} . \ottnt{A_{{\mathrm{2}}}}     \vdash \ottnt{B} :  \textbf{Type}  }%
}{
\Delta  \mathsf{;}  \Gamma_{{\mathrm{1}}}  \ottsym{+}  \Gamma_{{\mathrm{2}}}  \vdash  \ottkw{let} \, \ottsym{(}  \ottmv{x}  \ottsym{,}  \ottmv{y}  \ottsym{)}  \ottsym{=}  \ottnt{a} \, \mathsf{in} \, \ottnt{b}  \ottsym{:}  \ottnt{B}  \ottsym{\{}  \ottnt{a}  \ottsym{/}  \ottmv{z}  \ottsym{\}}}{%
{\ottdrulename{T\_SpreadElim}}{}%
}}

\newcommand{\ottdruleTXXSpreadOne}[1]{\ottdrule[#1]{%
\ottpremise{ \Delta  ;  \Gamma_{{\mathrm{1}}}  \vdash \ottnt{a} :  \Sigma  \ottmv{x} \!\!:\!\! \ottnt{A_{{\mathrm{1}}}} . \ottnt{A_{{\mathrm{2}}}}  }%
\ottpremise{    \Delta ,   \ottmv{x} \!\!:\!\! \ottnt{A_{{\mathrm{1}}}}    ,   \ottmv{y} \!\!:\!\! \ottnt{A_{{\mathrm{2}}}}    ;     \Gamma_{{\mathrm{2}}} ,   \ottmv{x} \! :^{  { \color{black}{1} }  }\! \ottnt{A_{{\mathrm{1}}}}    ,   \ottmv{y} \! :^{  { \color{black}{1} }  }\! \ottnt{A_{{\mathrm{2}}}}    \vdash \ottnt{b} : \ottnt{B}  \ottsym{\{}  \ottsym{(}  \ottmv{x}  \ottsym{,}  \ottmv{y}  \ottsym{)}  \ottsym{/}  \ottmv{z}  \ottsym{\}} }%
\ottpremise{  \Delta ,   \ottmv{z} \!\!:\!\!  \Sigma  \ottmv{x} \!\!:\!\! \ottnt{A_{{\mathrm{1}}}} . \ottnt{A_{{\mathrm{2}}}}     ;   \Gamma_{{\mathrm{3}}} ,   \ottmv{z} \! :^{ { \color{black}{r} } }\!  \Sigma  \ottmv{x} \!\!:\!\! \ottnt{A_{{\mathrm{1}}}} . \ottnt{A_{{\mathrm{2}}}}     \vdash \ottnt{B} :  \textbf{Type}  }%
}{
\Delta  \mathsf{;}  \Gamma_{{\mathrm{1}}}  \ottsym{+}  \Gamma_{{\mathrm{2}}}  \vdash  \ottkw{let} \, \ottsym{(}  \ottmv{x}  \ottsym{,}  \ottmv{y}  \ottsym{)}  \ottsym{=}  \ottnt{a} \, \mathsf{in} \, \ottnt{b}  \ottsym{:}  \ottnt{B}  \ottsym{\{}  \ottnt{a}  \ottsym{/}  \ottmv{z}  \ottsym{\}}}{%
{\ottdrulename{T\_Spread1}}{}%
}}

\newcommand{\ottdruleTXXSpreadTwo}[1]{\ottdrule[#1]{%
\ottpremise{ \Delta  ;  \Gamma_{{\mathrm{1}}}  \vdash \ottnt{a} :  \Sigma  \ottmv{x} \!\!:\!\! \ottnt{A_{{\mathrm{1}}}} . \ottnt{A_{{\mathrm{2}}}}  }%
\ottpremise{    \Delta ,   \ottmv{x} \!\!:\!\! \ottnt{A_{{\mathrm{1}}}}    ,   \ottmv{y} \!\!:\!\! \ottnt{A_{{\mathrm{2}}}}    ;     \Gamma_{{\mathrm{2}}} ,   \ottmv{x} \! :^{ { \color{black}{q} } }\! \ottnt{A_{{\mathrm{1}}}}    ,   \ottmv{y} \! :^{ { \color{black}{q} } }\! \ottnt{A_{{\mathrm{2}}}}    \vdash \ottnt{b} : \ottnt{B}  \ottsym{\{}  \ottsym{(}  \ottmv{x}  \ottsym{,}  \ottmv{y}  \ottsym{)}  \ottsym{/}  \ottmv{z}  \ottsym{\}} }%
\ottpremise{  \Delta ,   \ottmv{z} \!\!:\!\!  \Sigma  \ottmv{x} \!\!:\!\! \ottnt{A_{{\mathrm{1}}}} . \ottnt{A_{{\mathrm{2}}}}     ;   \Gamma_{{\mathrm{3}}} ,   \ottmv{z} \! :^{ { \color{black}{r} } }\!  \Sigma  \ottmv{x} \!\!:\!\! \ottnt{A_{{\mathrm{1}}}} . \ottnt{A_{{\mathrm{2}}}}     \vdash \ottnt{B} :  \textbf{Type}  }%
}{
\Delta  \mathsf{;}    { \color{black}{q} }   \cdot   \Gamma_{{\mathrm{1}}}    \ottsym{+}  \Gamma_{{\mathrm{2}}}  \vdash  \ottkw{let} \, \ottsym{(}  \ottmv{x}  \ottsym{,}  \ottmv{y}  \ottsym{)}  \ottsym{=}  \ottnt{a} \, \mathsf{in} \, \ottnt{b}  \ottsym{:}  \ottnt{B}  \ottsym{\{}  \ottnt{a}  \ottsym{/}  \ottmv{z}  \ottsym{\}}}{%
{\ottdrulename{T\_Spread2}}{}%
}}

\newcommand{\ottdruleTXXSpreadThree}[1]{\ottdrule[#1]{%
\ottpremise{ \Delta  ;  \Gamma_{{\mathrm{1}}}  \vdash \ottnt{a} :  \Sigma  \ottmv{x} \!\!:^ { \color{black}{q} }_{{\mathrm{1}}} \!\! \ottnt{A_{{\mathrm{1}}}} . \ottnt{A_{{\mathrm{2}}}}  }%
\ottpremise{    \Delta ,   \ottmv{x} \!\!:\!\! \ottnt{A_{{\mathrm{1}}}}    ,   \ottmv{y} \!\!:\!\! \ottnt{A_{{\mathrm{2}}}}    ;     \Gamma_{{\mathrm{2}}} ,   \ottmv{x} \! :^{ \ottsym{(}  { \color{black}{q} }_{{\mathrm{2}}}  \cdot  { \color{black}{q} }_{{\mathrm{1}}}  \ottsym{)} }\! \ottnt{A_{{\mathrm{1}}}}    ,   \ottmv{y} \! :^{ { \color{black}{q} }_{{\mathrm{2}}} }\! \ottnt{A_{{\mathrm{2}}}}    \vdash \ottnt{b} : \ottnt{B}  \ottsym{\{}  \ottsym{(}  \ottmv{x}  \ottsym{,}  \ottmv{y}  \ottsym{)}  \ottsym{/}  \ottmv{z}  \ottsym{\}} }%
\ottpremise{  \Delta ,   \ottmv{z} \!\!:\!\!  \Sigma  \ottmv{x} \!\!:^ { \color{black}{q} }_{{\mathrm{1}}} \!\! \ottnt{A_{{\mathrm{1}}}} . \ottnt{A_{{\mathrm{2}}}}     ;   \Gamma_{{\mathrm{3}}} ,   \ottmv{z} \! :^{ { \color{black}{r} } }\!  \Sigma  \ottmv{x} \!\!:^ { \color{black}{q} }_{{\mathrm{1}}} \!\! \ottnt{A_{{\mathrm{1}}}} . \ottnt{A_{{\mathrm{2}}}}     \vdash \ottnt{B} :  \textbf{Type}  }%
}{
\Delta  \mathsf{;}    { \color{black}{q} }_{{\mathrm{2}}}   \cdot   \Gamma_{{\mathrm{1}}}    \ottsym{+}  \Gamma_{{\mathrm{2}}}  \vdash  \ottkw{let} \, \ottsym{(}  \ottmv{x}  \ottsym{,}  \ottmv{y}  \ottsym{)}  \ottsym{=}  \ottnt{a} \, \mathsf{in} \, \ottnt{b}  \ottsym{:}  \ottnt{B}  \ottsym{\{}  \ottnt{a}  \ottsym{/}  \ottmv{z}  \ottsym{\}}}{%
{\ottdrulename{T\_Spread3}}{}%
}}

\newcommand{\ottdruleTXXTensorOne}[1]{\ottdrule[#1]{%
\ottpremise{ \Delta  ;  \Gamma_{{\mathrm{1}}}  \vdash \ottnt{a} : \ottnt{A} }%
\ottpremise{ \Delta  ;  \Gamma_{{\mathrm{2}}}  \vdash \ottnt{b} : \ottnt{B}  \ottsym{\{}  \ottnt{a}  \ottsym{/}  \ottmv{x}  \ottsym{\}} }%
\ottpremise{  \Delta ,   \ottmv{x} \!\!:\!\! \ottnt{A}    ;   \Gamma_{{\mathrm{3}}} ,   \ottmv{x} \! :^{ { \color{black}{r} } }\! \ottnt{A}    \vdash \ottnt{B} :  \textbf{Type}  }%
}{
\Delta  \mathsf{;}  \Gamma_{{\mathrm{1}}}  \ottsym{+}  \Gamma_{{\mathrm{2}}}  \vdash  \ottsym{(}  \ottnt{a}  \ottsym{,}  \ottnt{b}  \ottsym{)}  \ottsym{:}   \Sigma  \ottmv{x} \!\!:\!\! \ottnt{A} . \ottnt{B} }{%
{\ottdrulename{T\_Tensor1}}{}%
}}

\newcommand{\ottdruleTXXconvert}[1]{\ottdrule[#1]{%
\ottpremise{ \Delta  ;  \Gamma_{{\mathrm{1}}}  \vdash \ottnt{a} : \ottnt{A} }%
\ottpremise{ \Delta  ;  \Gamma_{{\mathrm{2}}}  \vdash \ottnt{B} :  \textbf{Type}  }%
\ottpremise{\ottnt{A}  \equiv  \ottnt{B}}%
}{
\Delta  \mathsf{;}  \Gamma_{{\mathrm{1}}}  \vdash  \ottnt{a}  \ottsym{:}  \ottnt{B}}{%
{\ottdrulename{T\_convert}}{}%
}}

\newcommand{\ottdefnATyping}[1]{\begin{ottdefnblock}[#1]{$\Delta  \mathsf{;}  \Gamma  \vdash  \ottnt{a}  \ottsym{:}  \ottnt{A}$}{}
\ottusedrule{\ottdruleTXXUnitElim{}}
\ottusedrule{\ottdruleTXXSumElim{}}
\ottusedrule{\ottdruleTXXSpreadElim{}}
\ottusedrule{\ottdruleTXXSpreadOne{}}
\ottusedrule{\ottdruleTXXSpreadTwo{}}
\ottusedrule{\ottdruleTXXSpreadThree{}}
\ottusedrule{\ottdruleTXXTensorOne{}}
\ottusedrule{\ottdruleTXXconvert{}}
\end{ottdefnblock}}

%% defn ASTyping
\newcommand{\ottdruleSTXXLetPairElim}[1]{\ottdrule[#1]{%
\ottpremise{ \Delta \ ;\  \Gamma_{{\mathrm{1}}}  \vdash \ottnt{a} :  \ottnt{A_{{\mathrm{1}}}}  \otimes  \ottnt{A_{{\mathrm{2}}}}  }%
\ottpremise{ \Delta \ ;\   \Gamma_{{\mathrm{2}}} ,  \ottsym{(}    \ottmv{x} \! :^{ { \color{black}{q} } }\! \ottnt{A_{{\mathrm{1}}}}  ,   \ottmv{y} \! :^{ { \color{black}{q} } }\! \ottnt{A_{{\mathrm{2}}}}    \ottsym{)}   \vdash \ottnt{b} : \ottnt{B} }%
}{
\Delta  \mathsf{;}  \ottsym{(}   { \color{black}{q} }   \cdot   \Gamma_{{\mathrm{1}}}   \ottsym{)}  \ottsym{+}  \Gamma_{{\mathrm{2}}}  \vdash  \ottkw{let} \, \ottsym{(}  \ottmv{x}  \ottsym{,}  \ottmv{y}  \ottsym{)}  \ottsym{=}  \ottnt{a} \, \mathsf{in} \, \ottnt{b}  \ottsym{:}  \ottnt{B}}{%
{\ottdrulename{ST\_LetPairElim}}{}%
}}

\newcommand{\ottdruleSTXXLetBoxElim}[1]{\ottdrule[#1]{%
\ottpremise{ \Delta \ ;\  \Gamma_{{\mathrm{1}}}  \vdash \ottnt{a} :  \Box^{ { \color{black}{q} } }  \ottnt{A}  }%
\ottpremise{  \Delta ,   \ottmv{x} \!\!:\!\! \ottnt{A}   \ ;\   \Gamma_{{\mathrm{2}}} ,   \ottmv{x} \! :^{ { \color{black}{r} }  \cdot  { \color{black}{q} } }\! \ottnt{A}    \vdash \ottnt{b} : \ottnt{B} }%
}{
\Delta  \mathsf{;}  \ottsym{(}   { \color{black}{r} }   \cdot   \Gamma_{{\mathrm{1}}}   \ottsym{)}  \ottsym{+}  \Gamma_{{\mathrm{2}}}  \vdash   \ottkw{let}\, \ottkw{box} \, \ottmv{x} \,=\, \ottnt{a} \ \ottkw{in}\  \ottnt{b}   \ottsym{:}  \ottnt{B}}{%
{\ottdrulename{ST\_LetBoxElim}}{}%
}}

\newcommand{\ottdefnASTyping}[1]{\begin{ottdefnblock}[#1]{$\Delta  \mathsf{;}  \Gamma  \vdash  \ottnt{a}  \ottsym{:}  \ottnt{A}$}{}
\ottusedrule{\ottdruleSTXXLetPairElim{}}
\ottusedrule{\ottdruleSTXXLetBoxElim{}}
\end{ottdefnblock}}

\newcommand{\ottdefnsJExtra}{
\ottdefnATyping{}\ottdefnASTyping{}}

\newcommand{\ottdefnss}{
\ottdefnsJOp
\ottdefnsJSimpleTyping
\ottdefnsJPath
\ottdefnsJTyping
\ottdefnsJSmallStep
\ottdefnsJSmallStepSecond
\ottdefnsJMultiStep
\ottdefnsJCompat
\ottdefnsJExtra
}

\newcommand{\ottall}{\ottmetavars\\[0pt]
\ottgrammar\\[5.0mm]
\ottdefnss}